\documentclass[11pt,a4paper]{article}
\usepackage{amssymb}
\usepackage{amsmath}
\usepackage{amsfonts}
\usepackage{dsfont}
\usepackage{amsthm}
\usepackage{mathrsfs}
\usepackage{hyperref}
\usepackage{color}
\usepackage[margin=2.41cm]{geometry}
\usepackage[all,cmtip]{xy}
\usepackage[utf8]{inputenc}
\usepackage{graphicx}
\usepackage{varwidth}
\usepackage{comment}

\usepackage{upgreek}
\usepackage{rotating}

\usepackage{tikz}
\usetikzlibrary{shapes.geometric}

\usepackage{tikz-cd}
\usetikzlibrary{matrix,arrows,calc,decorations.pathmorphing,fit,shapes.geometric,decorations.pathreplacing,positioning}

\usepackage[shortlabels]{enumitem}

\definecolor{darkred}{rgb}{0.8,0.1,0.1}
\hypersetup{
     colorlinks=false,         
     linkcolor=darkred,
     citecolor=blue,
}

\theoremstyle{plain}
\newtheorem{theo}{Theorem}[section]
\newtheorem{lem}[theo]{Lemma}
\newtheorem{propo}[theo]{Proposition}
\newtheorem{cor}[theo]{Corollary}

\theoremstyle{definition}
\newtheorem{defi}[theo]{Definition}

\newtheorem{remm}[theo]{Remark}
\newenvironment{rem}
  {\pushQED{\qed}\remm}
  {\popQED\endremm}

\numberwithin{equation}{section}

\newcommand{\CongTo}{\xrightarrow{\raisebox{-0.3 em}{\smash{\ensuremath{\text{\tiny$\cong$}}}}}}
\newcommand{\SimTo}{\xrightarrow{\raisebox{-0.3 em}{\smash{\ensuremath{\text{\tiny$\sim$}}}}}}

\def\nn{\nonumber}

\def\bbK{\mathbb{K}}
\def\bbR{\mathbb{R}}
\def\bbC{\mathbb{C}}

\def\bbZ{\mathbb{Z}}
\def\bbT{\mathbb{T}}
\def\bbS{\mathbb{S}}

\def\Hom{\mathrm{Hom}}

\def\End{\mathrm{End}}

\def\Der{\mathrm{Der}}

\def\Sym{\mathrm{Sym}}

\DeclareMathOperator{\ad}{ad}

\def\id{\mathrm{id}}

\def\dd{\mathrm{d}}

\def\cc{\mathrm{c}}

\def\dim{\mathrm{dim}}
\def\1{\mathbf{1}}
\def\oone{\mathds{1}}
\def\op{\mathrm{op}}

\def\pr{\mathrm{pr}}

\def\Set{\mathbf{Set}}
\def\Alg{\mathbf{Alg}}
\def\CoAlg{\mathbf{CoAlg}}

\def\Ch{\mathbf{Ch}}

\def\Op{\mathbf{Op}}
\def\CoOp{\mathbf{CoOp}}

\def\SymSeq{\mathbf{SymSeq}}

\def\dR{\mathrm{dR}}
\def\CE{\mathrm{CE}}

\def\X{\mathsf{X}}
\def\Y{\mathsf{Y}}

\def\g{\mathfrak{g}}

\def\CC{\mathcal{C}}

\def\PP{\mathcal{P}}

\def\FFF{\mathfrak{F}}

\def\H{\mathsf{H}}

\def\conil{\mathrm{conil}}
\def\aug{\mathrm{aug}}

\newcommand\und[1]{\underline{#1}}

\def\sk{\vspace{1mm}}

\makeatletter
\let\@fnsymbol\@alph
\makeatother

\title{%
Universal first-order Massey product of a prefactorization algebra
}

\author{%
Simen Bruinsma$^{1,a}$, Alexander Schenkel$^{2,b}$\ and\ Beno\^{\i}t Vicedo$^{1,c}$\vspace{4mm}\\
{\small ${}^1$ Department of Mathematics, University of York,}\\
{\small Heslington, York YO10 5DD, United Kingdom.}\vspace{2mm}\\
{\small ${}^2$ School of Mathematical Sciences, University of Nottingham,}\\
{\small University Park, Nottingham NG7 2RD, United Kingdom.}\vspace{4mm}\\
{\small \begin{tabular}{ll}
Email: & ${}^a$~\texttt{simen.bruinsma@york.ac.uk}\\
& ${}^b$~\texttt{alexander.schenkel@nottingham.ac.uk}\\
& ${}^c$~\texttt{benoit.vicedo@gmail.com}\\
\vspace{2mm}
\end{tabular}
}
}

\date{July 2023}

\begin{document}

\maketitle

\begin{abstract}
\noindent This paper studies the universal first-order Massey product of a prefactorization algebra, which encodes higher algebraic operations on the cohomology. Explicit computations of these structures are carried out in the locally constant case, with applications to factorization envelopes on $\mathbb{R}^m$ and a compactification of linear Chern-Simons theory on $\mathbb{R}^2\times \mathbb{S}^1$.
\end{abstract}



\renewcommand{\baselinestretch}{0.8}\normalsize
\tableofcontents
\renewcommand{\baselinestretch}{1.0}\normalsize



\section{\label{sec:intro}Introduction and summary}
Prefactorization algebras \cite{CostelloGwilliam,CostelloGwilliam2}
are a modern and versatile approach to quantum field theory (QFT) with a broad 
range of applications, e.g.\ in topological \cite{GwilliamGrady,GwilliamWilliamsTFT,ElliottSafronov,ElliottGwilliam}, 
holomorphic \cite{Williams,Williams2,GwilliamWilliams}
and Lorentzian \cite{GwilliamRejzner,BPS,BMS,GwilliamRejzner2} settings.
They are designed to axiomatize the algebraic structure of
observables in a QFT on a manifold $M$, possibly
with additional geometric structure, such as an orientation, a metric or a complex structure. 
In its most basic form, this algebraic structure is relatively simple:
To every suitable open subset $U\subseteq M$
is associated a cochain complex $\FFF(U)$ of observables and to every 
mutually disjoint family $(U_1,\dots,U_n) \subseteq U$ of suitable open subsets is associated 
a cochain map $\bigotimes_{i=1}^n\FFF(U_i)\to \FFF(U)$ that combines observables
in the small opens $U_i$ to an observable in the big open $U$. 
From a mathematical perspective, prefactorization algebras 
are algebras over a colored dg-operad $\PP_M$ whose objects are suitable opens 
$U\subseteq M$ and whose operations are mutually disjoint families of inclusions $(U_1,\dots,U_n) \subseteq U$. 
Depending on the flavor of QFT that one intends to describe, there
may be additional axioms. Most notably, in a topological QFT,
one demands that the structure map $\FFF(U)\SimTo \FFF(U^\prime)$
is a quasi-isomorphism for every isotopy equivalence $U\subseteq U^\prime$,
which formalizes the intuition that observables in a topological field
theory depend only on the `shape', but not on the `size', of the subset $U\subseteq M$.
Prefactorization algebras $\FFF$ with this property
are called locally constant and they are related to 
factorization homology \cite{AyalaFrancis,CFM}.
\sk

In the discussion above, we have intentionally kept vague the term `suitable opens' $U\subseteq M$ 
for the objects of the prefactorization operad $\PP_M$ in order to 
accommodate for the different choices which appear in the literature. 
In the books of Costello and Gwilliam \cite{CostelloGwilliam,CostelloGwilliam2}
and follow-up works such as \cite{PFAKoszul}, the default choice is to 
consider \textit{all} open subsets $U\subseteq M$. In contrast to this,
Lurie considers only topological open disks $D\subseteq M$, i.e.\ open subsets which 
are homeomorphic to $\bbR^m$ with $m=\dim(M)$, see \cite[Definition 5.4.5.6 and Remark 5.4.5.7]{HigherAlgebra}.
The latter choice is consistent with the one of Ayala and Francis 
by noting that what they define in \cite[Definition 2.9]{AyalaFrancis}
is \textit{not} the operad $\PP_M$ but rather its monoidal envelope $\PP_M^\otimes$,
which is a universally constructed symmetric monoidal category whose objects are tuples 
of the objects of $\PP_M$. (In this case these are finite disjoint unions of disks.) 
In our paper we follow \cite{HigherAlgebra,AyalaFrancis}
and focus on disks $D\subseteq M$ rather than general open subsets $U\subseteq M$.
\sk

Cochain complexes appear in 
prefactorization algebras as a manifestation of the BV formalism
from theoretical physics. They are necessary to capture the rich 
and interesting homological phenomena that arise from gauge 
symmetries and the complicated dynamical behavior of QFTs.
The world of cochain complexes is naturally $\infty$-categorical,
with higher morphisms given by (higher) cochain homotopies and
equivalences given by quasi-isomorphisms, which leads to conceptual
and also practical difficulties when working with prefactorization algebras.
The main complication is that cochain complexes may contain redundancies,
e.g.\ one can replace the complex of observables $\FFF(U)$ with a much bigger quasi-isomorphic 
complex (in physics terminology, this corresponds to introducing `auxiliary fields'),
hence it is difficult to give a concrete interpretation of elements in $\FFF(U)$.
One way to circumvent such issues is to take the cohomology $\H\FFF$ of the prefactorization
algebra, which produces a prefactorization algebra that takes values in graded vector spaces (i.e.\ 
cochain complexes with trivial differential). Unfortunately, this construction in general 
forgets/truncates some of the structure of the original prefactorization algebra $\FFF$,
except in very special cases where $\FFF$ is formal. 
Using more sophisticated technology from operad theory, there are ways to 
recover the entire structure of $\FFF$ at the level of the cohomology $\H\FFF$.
The relevant concept is the so-called \textit{homotopy transfer theorem} (see e.g.\ \cite{LodayVallette}), 
or more specifically the \textit{minimal model} construction, which allows one to transfer the prefactorization 
algebra structure on $\FFF$ to an $\infty$-prefactorization algebra structure on $\H\FFF$, such that
there exists an equivalence $\H\FFF \stackrel{\sim}{\rightsquigarrow} \FFF$. 
\sk

The main goal of this paper is to describe the first, potentially non-trivial, 
higher structure of the cohomology $\infty$-prefactorization algebra $\H\FFF$,
and to illustrate this structure by simple examples. Abstractly, it is given by the
\textit{universal first-order Massey product} \cite{UniversalMassey1,UniversalMassey2}.
It is realized by a cohomology class, constructed from any choice of minimal model, 
that describes the first-order obstruction of $\FFF$ being formal \cite{Dimitrova}.
The advantage of this cohomology class, in contrast to a minimal model, is that it is insensitive to the 
choice of minimal model, hence it does not suffer from any redundancy or dependence of
auxiliary choices. We would like to mention that there also exist successive higher obstruction classes 
\cite{Dimitrova}, or in other words higher-order Massey products \cite{UniversalMassey3},
which encode higher-order obstructions to $\FFF$ being formal. In the present paper,
we restrict our attention to the simplest case of universal \textit{first-order} Massey products,
but we hope to come back to their higher-order variants in a future work.
\sk

The main mathematical tools that we use for formulating and proving our results
are from homotopical algebra and operad theory, see e.g.\ the influential works \cite{Hinich1,Hinich2}
and the comprehensive monograph \cite{LodayVallette}. Similar techniques have been used previously
in the context of prefactorization algebras. For instance, Carmona, Flores and Muro \cite{CFM} 
describe the abstract homotopy theory of prefactorization algebras and factorization algebras, 
which are variants satisfying a descent condition, 
using model category theory and Bousfield localizations. 
A more concrete and computational approach, based on Koszul duality of dg-operads, 
was studied recently by Idrissi and Rabinovich in \cite{PFAKoszul}. Their main
result is a proof that a certain variant of the prefactorization operad (whose objects are
all open subsets and not only disks) is Koszul,
which implies that there exists a quite concrete model for homotopy-coherent
prefactorization algebras, a homotopy transfer theorem, and hence a concept of 
minimal models in this case.
\sk

The outline of the remainder of our paper is as follows. In Section \ref{sec:prelim},
we recall the relevant background on operadic homological algebra, following mainly
the presentation of Loday and Vallette \cite{LodayVallette}, but slightly generalizing
their constructions to the case of colored dg-operads. More specifically,
we recall the operadic bar-cobar adjunction, which allows us to determine
a semi-free resolution $\PP_{\infty}\SimTo\PP$ of any augmented colored dg-operad $\PP$,
and hence a concept of homotopy-coherent $\PP$-algebras, also known as $\PP_{\infty}$-algebras.
With these methods we also obtain a homotopy transfer theorem, as well as a minimal
model construction. We conclude this section
with an explicit and computationally accessible description of the cohomology class 
that describes the universal first-order
Massey product, slightly generalizing the constructions in 
\cite{Dimitrova} and \cite{UniversalMassey1,UniversalMassey2} to the case of colored operads.
\sk

In Section \ref{sec:PFA}, we apply these homological techniques to prefactorization algebras
on a manifold $M$, thereby obtaining an explicit description of the minimal model
and the universal first-order Massey product of any prefactorization algebra $\FFF$.
In the special case of a locally constant prefactorization algebra $\FFF$ on $M=\bbR^m$, which
as explained above describes a topological QFT on the Cartesian space, we prove some non-trivial 
results that provide a simplified model for the cocycle that determines
the universal first-order Massey product, see in particular Proposition \ref{prop: improved Masseys}.
In Subsection \ref{subsec:2dinvariant} we determine a very simple and computationally
accessible invariant for a locally constant prefactorization algebra on the $2$-dimensional
Cartesian space $\bbR^2$, which provides sufficient conditions for proving
non-triviality of the universal first-order Massey product. In contrast to  
the entire cohomology class that defines the universal first-order Massey product, 
our invariant has pleasant algebraic properties, namely it is given by a degree $-1$
Poisson bracket, i.e.\ a $\mathbb{P}_2$-algebra structure, see Theorem \ref{theo:LPoisson}. This allows us 
to make contact with the result of Lurie \cite[Theorem 5.4.5.9]{HigherAlgebra}, 
see also \cite{AyalaFrancis} and \cite{CFM}, that locally constant 
prefactorization algebras on $\bbR^m$ are equivalent to $\mathbb{E}_m$-algebras.
(See Remark \ref{rem:PFAEmPm} for further explanations and comments on this point.)
We would like to emphasize that our Poisson bracket invariant has the practical advantage 
that it is described very explicitly, hence it can be computed in examples.
A more physical approach towards such higher structures, in the context of supersymmetric
QFTs, has previously appeared in \cite{Beem}.
\sk

In Section \ref{sec:examples}, we illustrate and apply our results 
to the factorization envelope $\FFF = \mathfrak{U} \g^{\bbR^m}$ on the 
$m$-dimensional Cartesian space $\bbR^m$, where $\g^{\bbR^m} = \g\otimes \Omega_{\bbR^m}^{\bullet}$
denotes the local dg-Lie algebra determined by an ordinary Lie algebra $\g$
and the sheaf of de Rham complexes $\Omega_{\bbR^m}^{\bullet}$ on $\bbR^m$.
In $m=1$ dimensions, we prove that $\mathfrak{U} \g^{\bbR}$ is formal, i.e.\
there do not exist higher structures. We also show by an explicit computation,
using homological perturbation theory and homotopy transfer, that the cohomology 
of $\mathfrak{U} \g^{\bbR}$ is equivalent to the associative and unital algebra $\big(\Sym(\g),\star,1\big)$ 
with multiplication $\star$ given by the Gutt star-product \cite{Gutt}. This provides
an alternative proof for the result in \cite[Proposition 3.4.1]{CostelloGwilliam}.
In $m\geq 3$ dimensions, we prove that the cohomology of  $\mathfrak{U} \g^{\bbR^m}$ 
is simply given by the associative, unital and commutative algebra $\Sym(\g[1-m])$
and that the universal first-order Massey product is trivial.
We expect that in this case there are non-trivial higher obstruction classes 
\cite{Dimitrova}, i.e.\ higher-order Massey products \cite{UniversalMassey3},
but we do not attempt to describe these in our work. The most interesting case for us is $m=2$
dimensions, where we prove that the cohomology of $\mathfrak{U} \g^{\bbR^2}$ is
the associative, unital and commutative algebra $\Sym(\g[-1])$ and that the 
universal first-order Massey product is non-trivial. This non-triviality follows
from a detailed investigation of the simple $2$-dimensional Poisson bracket 
invariant from Subsection \ref{subsec:2dinvariant}. Hence, with our methods
we are able to show that the $2$-dimensional factorization envelope $\mathfrak{U} \g^{\bbR^2}$
is a non-formal prefactorization algebra, and we are able to compute explicitly
to first order the higher structures that cause this non-formality.
\sk

In Section \ref{sec:CS} we study linear Chern-Simons theory $\FFF_{\mathrm{CS}}^{}$ 
with structure group $G=\bbR$, 
another simple example of a locally constant prefactorization algebra. 
When defined on $M=\bbR^3$, we find similarly to
the case of factorization envelopes in Section \ref{sec:examples} that the universal first-order 
Massey product is trivial, but we again expect that there are non-trivial higher obstruction classes,
i.e.\ higher-order Massey products. In order to exhibit a non-trivial universal first-order Massey product,
we consider the compactification of linear Chern-Simons theory on $M = \bbR^2\times \bbS^1$,
which we regard as a $2$-dimensional prefactorization algebra along $\bbR^2$. We show that
the cohomology of $\FFF_{\mathrm{CS}}^{}$ is the associative, unital and commutative algebra
$\Sym\big(\H^\bullet_{\dR}(\bbS^1)\big)$, with $\H^\bullet_{\dR}(\bbS^1)$ the de Rham cohomology of the circle $\bbS^1$,
and that the universal first-order Massey product is non-trivial. 
For the latter conclusion we again compute explicitly 
the $2$-dimensional Poisson bracket invariant from Subsection \ref{subsec:2dinvariant}.


\section{\label{sec:prelim}Background on operadic homological algebra}
In this section we recall some relevant concepts and homological
tools from operad theory that are needed for this work. This also allows us to fix our notations 
and conventions. We refer the reader to \cite{LodayVallette} for further details.

\subsection{Cochain complexes} 
Let us fix a field $\bbK$ of characteristic $0$. We work with
cohomological degree conventions and denote by $\Ch$
the category of cochain complexes of $\bbK$-vector spaces. 
This category is closed symmetric monoidal with respect to the 
usual structures that we shall briefly recall. The tensor product $V\otimes W\in\Ch$
of two cochain complexes $V$ and $W$ is the cochain complex defined by
\begin{subequations}
\begin{flalign}
(V\otimes W)^i\,:=\,\bigoplus_{j\in\bbZ}\Big( V^j\otimes W^{i-j}\Big)\quad,
\end{flalign}
for all $i\in\bbZ$, and the differential 
\begin{flalign}
\dd_{V\otimes W}^{}(v\otimes w) \,:= \,
(\dd_V^{} v)\otimes w + (-1)^{\vert v\vert} \, v\otimes (\dd_W^{} w)\quad,
\end{flalign} 
\end{subequations}
for all homogeneous $v\in V$ and $w\in W$, 
where $\vert v\vert\in\bbZ$ denotes the cohomological degree.
The monoidal unit is $\bbK\in\Ch$, regarded as a cochain complex concentrated 
in degree $0$ and endowed with the trivial differential. The symmetric braiding
is given by the Koszul sign rule
\begin{flalign}\label{eqn:braiding}
\tau\,:\, V \otimes W\,\longrightarrow \, W\otimes V~,~~v\otimes w\,\longmapsto\,(-1)^{\vert v\vert\,\vert w\vert}\,w\otimes v
\end{flalign}
on homogeneous elements $v\in V$ and $w\in W$. The internal hom $[V,W]\in\Ch$
between two cochain complexes $V$ and $W$ is the cochain complex defined by
\begin{subequations}\label{eqn:Chinternalhom}
\begin{flalign}
[V,W]^i\,:=\,\prod_{j\in\bbZ} \Hom_\bbK\big(V^j,W^{j+i}\big)\quad,
\end{flalign}
for all $i\in\bbZ$, where $\Hom_\bbK$ denotes the vector space of linear maps,
and the differential
\begin{flalign}
\partial L\,:=\,\dd_W^{}\, L - (-1)^{\vert L\vert}\,L\, \dd_V^{}\quad,
\end{flalign}
\end{subequations}
for all homogeneous $L\in[V,W]$. 
\sk

Given any integer $p\in\bbZ$ and cochain complex $V\in\Ch$,
our convention for the $p$-shifted cochain complex $V[p]\in\Ch$
is $V[p]^i := V^{i+p}$, for all $i\in\bbZ$, and $\dd_{V[p]}^{}:= (-1)^{p}\,\dd_V^{}$.

\subsection{Operads and cooperads} 
Let us fix a non-empty set $C\in\Set$, which 
will play the role of the set of colors (or objects) for our operads. 
We denote by $\Sigma_C$ the groupoid whose objects are (possibly empty) tuples
$\und{c}:=(c_1,\dots,c_n)$ of elements in $C$ and whose morphisms
are right permutations $\und{c}= (c_1,\dots,c_n) \to \und{c}\sigma:= (c_{\sigma(1)},\dots,c_{\sigma(n)})$,
where $\sigma\in \Sigma_n$ is an element of the permutation group on $n$ letters.
The category of symmetric sequences is defined as the functor category
\begin{flalign}
\SymSeq_C\,:=\, \Ch^{\Sigma_C\times C}\quad.
\end{flalign}
Explicitly, this means that an object is a functor $X: \Sigma_C\times C\to \Ch$
assigning to each $(\und{c},c) = ((c_1,\dots,c_n),c)$ a cochain complex $X\big(\substack{c\\\und{c}}\big)\in\Ch$
and to each permutation $(\und{c},c) \to (\und{c}\sigma,c)$ a cochain map 
$X\big(\substack{c\\\und{c}}\big)\to X\big(\substack{c\\\und{c}\sigma}\big)$,
and that a morphism $X\to Y$ is a family of cochain maps $X\big(\substack{c\\\und{c}}\big)\to
Y\big(\substack{c\\\und{c}}\big)$, for all $(\und{c},c)$, that is equivariant 
with respect to permutations.
The category $\SymSeq_C$ is monoidal with respect to the circle-product $X\circ Y\in\SymSeq_C$,
which is defined object-wise by the coend formula
\begin{flalign}\label{eqn:circle}
(X\circ Y)\big(\substack{c\\ \und{c}}\big)\,:=\, 
\int^{\und{a}\in\Sigma_C} \int^{(\und{b}_1,\dots,\und{b}_k)\in\Sigma_C^k}\Sigma_C\big((\und{b}_1,\dots,\und{b}_k),\und{c}\big)\otimes X\big(\substack{c\\ \und{a}}\big)\otimes Y\big(\substack{a_1\\ \und{b}_1}\big)\otimes\cdots\otimes Y\big(\substack{a_k\\\und{b}_k}\big)\quad,
\end{flalign}
for all $(\und{c},c)$,
where $k$ denotes the length of the tuple $\und{a} = (a_1,\dots,a_k)$
and $\Sigma_C\big((\und{b}_1,\dots,\und{b}_k),\und{c}\big)\in\Set$ 
the Hom-set in $\Sigma_C$.  The $\Set$-tensoring on cochain complexes
is defined by $\otimes:\Set\times\Ch\to\Ch\,,~S\times V\mapsto S\otimes V := \bigoplus_{s\in S}V$.
The monoidal unit $I_\circ\in\SymSeq_C$ for the circle-product is
\begin{flalign}
I_\circ\big(\substack{c\\ \und{c}}\big)\,:=\,\Sigma_C(\und{c},c)\otimes \bbK\,=\,\begin{cases}
\bbK&~,~~\text{if }\und{c}=c~~,\\
0&~,~~\text{else }~~,
\end{cases}
\end{flalign}
for all $(\und{c},c)$.
\begin{defi}
\begin{itemize}
\item[(a)] A ($C$-colored symmetric dg-)operad is an associative and unital monoid
$\PP = \big(\PP,\gamma:\PP\circ\PP\to\PP,\oone:I_\circ\to\PP\big)$ 
in the monoidal category $(\SymSeq_C,\circ,I_\circ)$.
We denote the category of operads by $\Op_C$. An augmented operad is
a pair $(\PP,\epsilon)$ consisting of an operad $\PP$
and an operad morphism $\epsilon : \PP\to I_\circ$. We denote the category
of augmented operads and augmentation preserving morphisms by $\Op_C^{\aug}$.

\item[(b)] A ($C$-colored symmetric dg-)cooperad is a coassociative and counital
comonoid $\CC=\big(\CC,\Delta:\CC\to\CC\circ\CC,\epsilon:\CC\to I_\circ\big)$ 
in the monoidal category $(\SymSeq_C,\circ,I_\circ)$.
We denote the category of cooperads by $\CoOp_C$. A coaugmented
cooperad is a pair $(\CC,\oone)$ consisting of a cooperad $\CC$
and a cooperad morphism $\oone : I_\circ\to \CC$. We denote the category
of coaugmented cooperads and coaugmentation preserving morphisms by
$\CoOp_{C}^{\mathrm{coaug}}$. The full subcategory of conilpotent
cooperads is denoted by $\CoOp_C^{\conil}\subseteq \CoOp_{C}^{\mathrm{coaug}}$, 
see e.g.\ \cite[Section 5.8]{LodayVallette} for the definition of a conilpotent cooperad.
\end{itemize}
\end{defi}

We briefly recall the construction of free augmented operads and of cofree conilpotent cooperads
in terms of tree modules, see \cite[Section 5]{LodayVallette} for the details. 
Given any $X\in \SymSeq_C$, we define inductively a family of $\SymSeq_C$-objects by
\begin{subequations}\label{eqn:treemodule}
\begin{flalign}
\bbT_0 X\,:=\, I_\circ\quad,\qquad \bbT_{n}X\,:=\, I_\circ\oplus \big(X\circ \bbT_{n-1}X\big)\quad,
\end{flalign}
for all $n\geq 1$, and a family of $\SymSeq_C$-morphism $\iota_n : \bbT_{n-1}X\to \bbT_{n}X$
by the inclusion map $\iota_1:\bbT_0 X= I_\circ\to I_\circ\oplus X=\bbT_1 X$ and
\begin{flalign}
\iota_n\,:=\id\oplus \big(\id\circ \iota_{n-1}\big)\,:\, I_\circ\oplus \big(X\circ \bbT_{n-2}X\big)\,\longrightarrow\,
I_\circ\oplus \big(X\circ \bbT_{n-1}X\big)  \quad, 
\end{flalign}
for all $n\geq 2$. The tree module is then defined as the colimit
\begin{flalign}
\bbT X\,:=\, \mathrm{colim}_{n\geq 0}^{~} \bbT_n X\,\in\SymSeq_C
\end{flalign}
\end{subequations}
in the category of symmetric sequences. This admits a graphical 
interpretation in terms of $C$-colored rooted trees whose vertices are ordered and decorated
by elements of $X$ that match the color pattern given by the tree.
\sk

A model for the free augmented operad associated with $X\in\SymSeq_C$ is given by
\begin{flalign}\label{eqn:freeOp}
\bbT X \,:=\, \big(\bbT X,\gamma,\oone,\epsilon\big)\,\in\,\Op_C^{\aug}
\end{flalign}
with composition $\gamma:\bbT X\circ \bbT X\to \bbT X$ given by grafting of trees, unit 
$\oone : I_\circ\to \bbT X$ the inclusion of $I_\circ  = \bbT_0 X$ into the colimit,
and augmentation $\epsilon:\bbT X\to I_\circ$ induced by the projections
$\pr_1 : \bbT_n X =  I_\circ\oplus \big(X\circ \bbT_{n-1}X\big) \to I_\circ$.
The universal property of the free augmented operad functor
$\bbT: \SymSeq_C\to \Op_C^{\aug}$ is that it is \textit{left} adjoint to
the forgetful functor $\mathbb{U} : \Op_C^{\aug}\to\SymSeq_C\,,~(\PP,\gamma,\oone,\epsilon)\mapsto
\overline{\PP}:= \ker\big(\epsilon:\PP\to I_\circ\big)$ that assigns the augmentation ideal.
\sk 

A model for the cofree conilpotent cooperad associated with $X\in\SymSeq_C$ is given by
\begin{flalign}\label{eqn:cofreeCoOp}
\bbT^\cc X\,:=\,\big(\bbT X,\Delta,\epsilon,\oone\big)\, \in\, \CoOp_C^{\conil}
\end{flalign}
with decomposition $\Delta : \bbT X\to\bbT X\circ \bbT X$ given by degrafting
of trees, counit $\epsilon : \bbT X\to I_\circ$ induced by the projections
$\pr_1 : \bbT_n X =  I_\circ\oplus \big(X\circ \bbT_{n-1}X\big) \to I_\circ$,
and coaugmentation $\oone : I_\circ\to \bbT X$ the inclusion of $I_\circ  = \bbT_0 X$ into the colimit.
The universal property of the cofree conilpotent cooperad functor
$\bbT^\cc: \SymSeq_C\to \CoOp_C^{\conil}$ is that it is \textit{right} adjoint to
the forgetful functor $\mathbb{U}^\cc : \CoOp_C^{\conil}\to\SymSeq_C\,,~(\CC,\Delta,\oone,\epsilon)\mapsto
\overline{\CC}:= \mathrm{coker}(\oone: I_\circ \to \CC)$ that assigns the coaugmentation coideal.

\subsection{Operadic bar-cobar adjunction} 
The operadic bar-cobar adjunction
\begin{flalign}\label{eqn:barcobar}
\Omega\,:\,\xymatrix@C=3.5em{
\CoOp_C^{\conil} \ar@<1.25ex>[r]_-{\perp}~&~\ar@<1.25ex>[l] \Op_C^{\aug}
}\,:\,\mathsf{B}
\end{flalign}
introduces a fruitful interplay between
augmented operads and conilpotent cooperads.
Its relevance in the context of our paper is that it allows us to construct a
semi-free resolution 
\begin{flalign}
\PP_\infty\,:=\,\Omega\mathsf{B}\PP\stackrel{\sim}{\longrightarrow}\PP
\end{flalign}
of an augmented operad $\PP\in\Op_C^{\aug}$, which provides a concept of homotopy-coherent
$\PP$-algebras (in terms of ordinary $\PP_\infty=\Omega\mathsf{B}\PP$-algebras) 
that is stable under homotopy transfer. 
A detailed description of the bar-cobar adjunction is spelled out in \cite[Section 6.5]{LodayVallette},
but we will briefly sketch the relevant steps in order to fix our notations.
\sk

The bar construction $\mathsf{B}$ assigns to an
augmented operad $\PP =(\PP,\gamma,\oone,\epsilon)\in \Op_C^{\aug}$
the semi-cofree conilpotent cooperad
\begin{flalign}\label{eqn:barconstruction}
\mathsf{B}\PP\,:=\, \bbT^\cc \overline{\PP}[1]_\gamma\,:=\,\big(\bbT^\cc\overline{\PP}[1], \dd_{\bbT^\cc \overline{\PP}[1]} 
+ \dd_\gamma\big)\,\in\,\CoOp_C^{\conil}
\end{flalign}
that is given by forming the cofree conilpotent cooperad \eqref{eqn:cofreeCoOp}
over the $1$-shifted augmentation ideal $\overline{\PP}[1] = \ker(\epsilon)[1]\in\SymSeq_C$
and then deforming its differential by a coderivation $\dd_\gamma$ that is constructed
out of the operadic composition map $\gamma$ of $\PP$.
\sk

The cobar construction $\Omega$ assigns to a conilpotent cooperad
$\CC = (\CC,\Delta,\epsilon,\oone)\in\CoOp_C^{\conil}$ the semi-free
augmented operad
\begin{flalign}\label{eqn:cobarconstruction}
\Omega\CC\,:=\, \bbT\overline{\CC}[-1]_\Delta\,:=\,\big(\bbT\overline{\CC}[-1] , \dd_{\bbT\overline{\CC}[-1]}+ \dd_{\Delta}\big)\,\in\,\Op_C^{\aug}
\end{flalign}
that is given by forming the free augmented operad \eqref{eqn:freeOp}
over the $(-1)$-shifted coaugmentation coideal $\overline{\CC}[-1] = \mathrm{coker}(\oone)[-1]\in\SymSeq_C$
and then deforming its differential by a derivation $\dd_\Delta$ that is constructed
out of the cooperadic decomposition map $\Delta$ of $\CC$.
\sk

One proves that \eqref{eqn:barcobar} is indeed an adjunction 
by using an intermediate step
\begin{flalign}\label{eqn:barcobarhom}
\Hom_{\Op_C^{\aug}}^{}\big(\Omega\CC, \PP\big)\,\cong\,\mathrm{Tw}\big(\CC,\PP\big)
\,\cong\, \Hom_{\CoOp_C^{\conil}}^{}\big(\CC,\mathsf{B}\PP\big)
\end{flalign}
that relates morphisms in these categories to so-called operadic twisting morphisms.
The latter are degree $1$ elements $\alpha\in[\CC,\PP]^1$ in the cochain complex
\begin{flalign}
[\CC,\PP]\,:=\, \int_{(\und{c},c)\in \Sigma_C\times C}\,\big[\CC\big(\substack{c\\ \und{c}}\big),
\PP\big(\substack{c\\ \und{c}}\big)\big]\,\in\,\Ch
\end{flalign}
of natural transformations between $\CC,\PP\in\SymSeq_C$, constructed as usual in terms of 
an end over the internal hom complexes \eqref{eqn:Chinternalhom}, that annihilate
the (co)augmentations, i.e.\  $\alpha\, \oone_\CC =0$ and $\epsilon_\PP\,\alpha=0$,
and satisfy the Maurer-Cartan equation
\begin{subequations}\label{eqn:MCtwisting}
\begin{flalign}
\partial \alpha + \alpha\star\alpha \,=\,0
\end{flalign}
with respect to the convolution product
\begin{flalign}
\alpha\star\alpha\,:\,\xymatrix@C=3em{
\CC \ar[r]^-{\Delta_{(1)}}~&~\CC\circ_{(1)}\CC \ar[r]^-{\alpha\circ_{(1)}\alpha}~&~ \PP\circ_{(1)}\PP \ar[r]^-{\gamma_{(1)}}~&~\PP
}\quad.
\end{flalign}
\end{subequations}
The subscripts ${}_{(1)}$ denote the infinitesimal composites from  
\cite[Section 6.1]{LodayVallette}, which are linearizations of the 
circle-product \eqref{eqn:circle} in the right entry.
The identifications in \eqref{eqn:barcobarhom} then work as follows:
By semi-freeness of $\Omega\CC$, the datum of an $\Op_C^{\aug}$-morphism 
$\Omega\CC\to \PP$ is equivalent to a degree $0$ map $\overline{\CC}[-1]\to \overline{\PP}$,
or equivalently a degree $1$ map $\CC\to \overline{\CC}\to \overline{\PP}\to \PP$
annihilating the (co)augmentations,
that due to the definition of the differential on \eqref{eqn:cobarconstruction}
satisfies the Maurer-Cartan equation \eqref{eqn:MCtwisting}.
Similarly, by semi-cofreeness of $\mathsf{B}\PP$, the datum of a $\CoOp_{C}^{\conil}$-morphism
$\CC\to\mathsf{B}\PP$ is equivalent to a degree $0$ map $\overline{\CC}\to \overline{\PP}[1]$,
or equivalently a degree $1$ map $\CC\to \overline{\CC}\to \overline{\PP}\to \PP$
annihilating the (co)augmentations,
that due to the definition of the differential on \eqref{eqn:barconstruction}
satisfies the Maurer-Cartan equation \eqref{eqn:MCtwisting}.
\begin{rem}\label{rem:nonaugmented}
The bijections in \eqref{eqn:barcobarhom} can be generalized to the case where the 
operads are not augmented. Given any non-augmented operad $\PP=(\PP,\gamma,\oone)\in \Op_C$,
one defines similarly to \eqref{eqn:barconstruction}, but without taking an augmentation ideal,
\begin{flalign}
\mathsf{B}^{\mathrm{na}}\PP\,:=\, 
\bbT^\cc \PP[1]_\gamma\,:=\,\big(\bbT^\cc\PP[1], \dd_{\bbT^\cc \PP[1]} 
+ \dd_\gamma\big)\,\in\,\CoOp_C^{\conil} \quad.
\end{flalign}
We further write $\Omega\CC\in \Op_C$ for the non-augmented
operad given by forgetting the augmentation of the cobar construction
in \eqref{eqn:cobarconstruction}. Then one has bijections
\begin{flalign}\label{eqn:barcobarhomna}
\Hom_{\Op_C}^{}\big(\Omega\CC, \PP\big)\,\cong\,\mathrm{Tw}^{\mathrm{na}}\big(\CC,\PP\big)
\,\cong\, \Hom_{\CoOp_C^{\conil}}^{}\big(\CC,\mathsf{B}^{\mathrm{na}}\PP\big)\quad,
\end{flalign}
where the relevant twisting morphisms in this case
are degree $1$ elements $\alpha\in[\CC,\PP]^1$ that annihilate the 
coaugmentation $\alpha\,\oone_{\CC}=0$, but of course not an augmentation since there is none,
and satisfy the Maurer-Cartan equation \eqref{eqn:MCtwisting}.
It is important to stress that, even though one has an adjunction
$\Omega\dashv \mathsf{B}^{\mathrm{na}}$ in the non-augmented case, 
its counit $\Omega\mathsf{B}^{\mathrm{na}}\PP \stackrel{\not\sim}{\longrightarrow}\PP$
does \textit{not} define a resolution of the non-augmented operad $\PP\in\Op_C$. 
This is inessential for the purpose of our paper since we are using
\eqref{eqn:barcobarhomna} only as a tool to relate the different equivalent models for 
homotopy-coherent algebras over an augmented operad. As a side-remark, resolutions of non-augmented
operads can be constructed from the curved bar-cobar adjunction \cite{Hirsh,LeGrignou}, but 
this is not needed in our paper.
\end{rem}

\subsection{Algebras and coalgebras} 
We have to introduce some further terminology before 
we can define the important concepts of (co)algebras over (co)operads.
A symmetric sequence
$N\in \SymSeq_C$ is called a $C$-colored object if
\begin{flalign}
N\big(\substack{c\\ \und{c}}\big) \,=\, 0\,\in\,\Ch \quad,
\end{flalign}
for all $(\und{c},c)$ with $\und{c}=(c_1,\dots,c_n)$ a tuple of length $n\geq 1$.
The full subcategory of $C$-colored objects is thus given by the functor category 
$\Ch^C \subseteq \SymSeq_C$, i.e.\ an object $N\in\Ch^C$ is simply a family of cochain complexes 
$\{N(c)= N\big(\substack{c\\ \varnothing}\big)\in\Ch\}_{c\in C}$.
From the definition of the circle-product \eqref{eqn:circle}, one immediately observes that
it restricts to a functor
\begin{subequations}
\begin{flalign}
\circ\,:\,\SymSeq_C\times\Ch^C~\longrightarrow~\Ch^C\quad,
\end{flalign}
which for $X\in\SymSeq_C$ and $N\in \Ch^C$ is given explicitly by
\begin{flalign}
(X\circ N)(c)\,=\,\int^{\und{a}\in\Sigma_C} X\big(\substack{c \\ \und{a}}\big)\otimes N(a_1)\otimes\cdots\otimes N(a_k)\quad,
\end{flalign}
\end{subequations}
for all $c\in C$. This implies that $(\Ch^C,\circ)$ is a left-module category over the monoidal category
$(\SymSeq_C,\circ,I_\circ)$.
\begin{defi}
\begin{itemize}
\item[(a)] Let $\PP = (\PP,\gamma,\oone)\in\Op_C$ be an operad. A $\PP$-algebra
is a pair $(A,\gamma_A)$ consisting of an object $A\in\Ch^C$ and a $\Ch^C$-morphism
$\gamma_A : \PP\circ A\to A$ (called left action) that satisfies the axioms of a left module, i.e.\
\begin{flalign}
\begin{gathered}
\xymatrix{
\ar[d]_-{\gamma\circ\id}\PP\circ\PP\circ A \ar[r]^-{\id\circ \gamma_A} ~&~\PP\circ A\ar[d]^-{\gamma_A} ~&~I_\circ\circ A \ar[r]^-{\oone\circ \id}\ar[dr]_-{\cong} ~&~\PP\circ A\ar[d]^-{\gamma_A}\\
\PP\circ A \ar[r]_-{\gamma_A}~&~A ~&~ ~&~A
}
\end{gathered}\qquad.
\end{flalign}
A morphism $f : (A,\gamma_A)\to (B,\gamma_B)$ of $\PP$-algebras is a $\Ch^C$-morphism
$f:A\to B$ that preserves the left actions, i.e.\ $\gamma_B\,(\id\circ f) = f\,\gamma_A$.
We denote the category of $\PP$-algebras by $\Alg_\PP$.
If $\PP = (\PP,\gamma,\oone,\epsilon)\in\Op_C^{\aug}$ is an augmented operad,
the category of $\PP$-algebras is defined as the category $\Alg_{\PP}$
of algebras over the underlying operad $(\PP,\gamma,\oone)$. 

\item[(b)]  Let $\CC = (\CC,\Delta,\epsilon)\in\CoOp_C$ be a cooperad. A $\CC$-coalgebra
is a pair $(D,\Delta_D)$ consisting of an object $D\in\Ch^C$ and a $\Ch^C$-morphism
$\Delta_D : D\to \CC\circ D$ (called left coaction) that satisfies the axioms of a left comodule, i.e.\
\begin{flalign}
\begin{gathered}
\xymatrix{
\ar[d]_-{\Delta_D}D \ar[r]^-{\Delta_D}~&~ \CC\circ D\ar[d]^-{\id\circ\Delta_D} ~&~ \ar[dr]_-{\cong}D\ar[r]^-{\Delta_D} ~&~\CC\circ D\ar[d]^-{\epsilon\circ \id}\\
\CC\circ D \ar[r]_-{\Delta\circ\id}~&~\CC\circ\CC\circ D ~&~ ~&~I_\circ \circ D
}
\end{gathered}\qquad.
\end{flalign}
A morphism $f : (D,\Delta_D)\to (E,\Delta_E)$ of $\CC$-coalgebras is a $\Ch^C$-morphism
$f:D\to E$ that preserves the left coactions, i.e.\ $(\id\circ f)\,\Delta_D = \Delta_E\,f$.
We denote the category of $\CC$-coalgebras by $\CoAlg_\CC$. If 
$\CC = (\CC,\Delta,\epsilon,\oone)\in\CoOp_C^{\conil/\mathrm{coaug}}$ 
is a conilpotent (or coaugmented) cooperad,
the category of $\CC$-coalgebras is defined as the category $\CoAlg_{\CC}$
of coalgebras over the underlying cooperad $(\CC,\Delta,\epsilon)$. 
\end{itemize}
\end{defi}
\begin{rem}\label{rem:Endoperad}
It is sometimes useful to work with an equivalent description of
$\PP$-algebras that uses the concept of endomorphism operads. For this we recall
that the monoidal category $(\SymSeq_\CC,\circ,I_\circ)$ is right closed,
i.e.\ there exists an internal hom functor $[-,-]_\circ : \SymSeq^{\op}_C\times\SymSeq_C\to\SymSeq_C$
such that
\begin{flalign}
\Hom_{\SymSeq_C}\big(X\circ Y,Z\big)\,\cong\, \Hom_{\SymSeq_C}\big(X,[Y,Z]_\circ\big)\quad,
\end{flalign}
for all $X,Y,Z\in\SymSeq_C$. When applied to a left action $\gamma_A : \PP\circ A\to A$,
this defines an $\Op_C$-morphism (denoted with abuse of notation by the same symbol)
\begin{flalign}
\gamma_A \,:\,\PP\,\longrightarrow\, \End_A \,:=\, [A,A]_\circ
\end{flalign}
to the endomorphism operad associated with $A\in\Ch^C$. The endomorphism operad 
$\End_A = (\End_A,\gamma,\oone)$
is given concretely by the symmetric sequence that is defined by
\begin{flalign}
\End_A\big(\substack{c\\ \und{c}}\big)\,=\, \big[A(c_1)\otimes\cdots\otimes A(c_n), A(c)\big]\,\in\,\Ch\quad,
\end{flalign}
for all $(\und{c},c)$, where on the right-hand side
$[-,-]$ denotes the internal hom \eqref{eqn:Chinternalhom} in $\Ch$,
with $\gamma:\End_A\circ \End_A\to \End_A$ defined by composing maps and 
$\oone:I_\circ \to \End_A$ assigning the identity maps.
\end{rem}

\subsection{Homotopy-coherent $\PP$-algebras} 
Let $\PP = (\PP,\gamma,\oone,\epsilon)\in\Op_C^{\aug}$
be an augmented operad. Using the bar-cobar adjunction \eqref{eqn:barcobar}, one obtains a semi-free resolution
$\PP_\infty:=\Omega\mathsf{B}\PP\SimTo\PP$ of $\PP$ that allows
one to introduce a concept of homotopy-coherent $\PP$-algebras.
\begin{defi}
A $\PP_\infty$-algebra is an algebra $(A,\gamma_A : \PP_\infty\circ A\to A)\in\Alg_{\PP_\infty}$ 
over the augmented operad $\PP_\infty:=\Omega\mathsf{B}\PP\in\Op_C^{\aug}$ given by the
semi-free resolution $\Omega\mathsf{B}\PP \SimTo \PP$.
A strict morphism $f : (A,\gamma_A)\to (B,\gamma_B)$ of $\PP_\infty$-algebras
is an ordinary $\PP_\infty$-algebra morphism, i.e.\ a $\Ch^C$-morphism
$f:A\to B$ satisfying $ \gamma_B\,(\id\circ f)= f\,\gamma_A$.
\end{defi}

Although conceptually very clear, this definition is impractical
for concrete computations because $\PP_\infty=\Omega\mathsf{B}\PP\in\Op_C^{\aug}$
is a complicated operad that is obtained from $\PP$ by taking twice a tree module,
see \eqref{eqn:barconstruction} and \eqref{eqn:cobarconstruction}, i.e.\ 
it consists of trees whose vertices are decorated by trees whose vertices are decorated by elements of $\PP$.
This can be simplified considerably by making use of the endomorphism operads from Remark \ref{rem:Endoperad}
and the variant of the bar-cobar adjunction from Remark \ref{rem:nonaugmented}. 
Indeed, the datum of a left action $\gamma_A : \Omega\mathsf{B}\PP\circ A\to A$
is equivalent to an $\Op_C$-morphism (denoted with abuse of notation 
by the same symbol) $\gamma_A : \Omega\mathsf{B}\PP\to\End_A$
to the endomorphism operad associated with $A$. Via bar-cobar, the latter can be expressed
in the following three equivalent ways
\begin{flalign}\label{eqn:RosettaStone}
\Hom_{\Op_C}\big(\Omega\mathsf{B}\PP,\End_A\big)\,\cong\,\mathrm{Tw}^{\mathrm{na}}\big(\mathsf{B}\PP,\End_A\big)\,\cong\,
\Hom_{\CoOp_C^{\conil}}\big(\mathsf{B}\PP,\mathsf{B}^{\mathrm{na}}\End_A\big)\quad.
\end{flalign}
The phenomenon of having available multiple useful descriptions
of $\PP_\infty$-algebra structures is known as the 
Rosetta Stone in \cite[Theorem 10.1.13]{LodayVallette}. The latter comes with one further
equivalent description that we will briefly recall. Given a $\PP_\infty$-algebra
structure that is described, say, by a twisting morphism $\alpha\in \mathrm{Tw}^{\mathrm{na}}\big(\mathsf{B}\PP,\End_A\big)$,
i.e.\ a degree $1$ element $\alpha\in [\mathsf{B}\PP,\End_A]^1$ that annihilates the coaugmentation
$\alpha\,\oone_{\mathsf{B}\PP} =0$ and satisfies the Maurer-Cartan equation \eqref{eqn:MCtwisting},
one can use that $\End_A =[A,A]_\circ$ is given by the internal hom in $\SymSeq_C$ 
in order to identify $\alpha$ with a degree $1$ element (denoted by the same symbol)
$\alpha\in [\mathsf{B}\PP\circ A,A]^1$. Interpreting $\mathsf{B}\PP\circ A$
as a cofree coalgebra over $\mathsf{B}\PP\in \CoOp_C^{\conil}$,
with left coaction $\Delta_{\mathsf{B}\PP} \circ \id : \mathsf{B}\PP\circ A \to
\mathsf{B}\PP\circ \mathsf{B}\PP\circ A$, one can extend the element
$(\epsilon_{\mathsf{B}\PP} \circ\id)\, \dd_{\mathsf{B}\PP\circ A} + \alpha \in [\mathsf{B}\PP\circ A,A]^1$
given by adding $\alpha$ to the differential on $\mathsf{B}\PP\circ A$ to a deformed
differential on $\mathsf{B}\PP\circ A$. (This differential squares to zero
as a consequence of the Maurer-Cartan equation, see \cite[Proposition 10.1.11]{LodayVallette}.)
Hence, the datum of a $\PP_\infty$-algebra structure on $A$ is equivalent
to a deformation of the cofree coalgebra $\mathsf{B}\PP\circ A\in\CoAlg_{\mathsf{B}\PP}$
to a semi-cofree one, which we denote by $(\mathsf{B}\PP\circ A)_\alpha\in\CoAlg_{\mathsf{B}\PP}$. 
The latter point of view is useful to introduce a weaker concept of morphisms between $\PP_\infty$-algebras.
\begin{defi}
Consider two $\PP_\infty$-algebras $(A,\gamma_A),(B,\gamma_B)\in \Alg_{\PP_\infty}$
and denote by $\alpha,\,\beta$ the twisting morphisms corresponding to the left actions $\gamma_A,\,\gamma_B$.
An $\infty$-morphism $\zeta : (A,\gamma_A)\rightsquigarrow (B,\gamma_B)$ between $\PP_\infty$-algebras
is an ordinary $\CoAlg_{\mathsf{B}\PP}$-morphism 
$\zeta: (\mathsf{B}\PP\circ A)_{\alpha}\to (\mathsf{B}\PP\circ B)_{\beta}$
between the associated semi-cofree coalgebras. An $\infty$-morphism is called
an $\infty$-quasi-isomorphism if the composite morphism 
\begin{flalign}
\xymatrix{
A \,\cong\,I_\circ \circ A \ar[r]^-{\oone_{\mathsf{B}\PP}\circ\id} ~&~(\mathsf{B}\PP\circ A)_\alpha 
\ar[r]^-{\zeta}~&~(\mathsf{B}\PP\circ B)_\beta\ar[r]^-{\epsilon_{\mathsf{B}\PP}\circ\id}~&~I_\circ\circ B\,\cong \, B
}
\end{flalign}
is a quasi-isomorphism in $\Ch^C$, i.e.\ an object-wise quasi-isomorphism of cochain complexes. 
\end{defi}

There exist other equivalent descriptions of $\infty$-morphisms, see \cite[Section 10.2]{LodayVallette}, 
but these will not be needed in our paper.

\subsection{\label{subsec:minmod}Homotopy transfer and minimal models} 
Consider any homotopy retract
\begin{equation}\label{eqn:homotopyretract}
\begin{tikzcd}
B \ar[r,shift right=-1ex,"i"] & \ar[l,shift right=-1ex,"p"] A \ar[loop,out=-30,in=30,distance=30,swap,"h"]
\end{tikzcd}
\end{equation}
in the category $\Ch^C$ of $C$-colored objects. This consists of two quasi-isomorphisms
$i: B\to A$ and $p: A\to B$ in $\Ch^C$, and a homotopy $h\in [A,A]^{-1}$, such
that $\partial h = i\,p-\id$.
\begin{theo}[Homotopy transfer theorem]\label{theo:HTT}
Let $\PP= (\PP,\gamma,\oone,\epsilon)\in \Op_C^{\aug}$ be an augmented operad.
Any $\PP_\infty$-algebra structure on $A\in\Ch^C$
can be transferred along the homotopy retract \eqref{eqn:homotopyretract} 
into a $\PP_\infty$-algebra structure on $B\in\Ch^C$
such that $i$ extends to an $\infty$-quasi-isomorphism.
\end{theo}

The proof of the homotopy transfer theorem is constructive and rather explicit, 
see \cite[Theorem 10.3.1]{LodayVallette}, but it heavily uses the identifications
given by the Rosetta Stone \eqref{eqn:RosettaStone}. 
Since we are interested in explicit computations, we have to spell out 
concretely how the transferred $\PP_\infty$-algebra structure on $B$ can be computed
from the given $\PP_\infty$-algebra structure on $A$. For the applications
considered in this paper, it suffices to consider the case where 
$A$ is endowed with a strict $\PP_\infty$-algebra structure $\gamma_A : \PP\to \End_A$. 
Applying the bar construction we obtain a $\CoOp_C^{\conil}$-morphism 
$\widehat{\alpha} : \mathsf{B}\PP\to \mathsf{B}^{\mathrm{na}}\End_A$ that, via the Rosetta
Stone, gives an equivalent description of the given $\PP$-algebra structure $\gamma_A$. Explicitly,
when interpreted in terms of trees, $\widehat{\alpha}$ sends a tree that is decorated 
with operations in $\overline{\PP}[1]$ to the same tree, but now decorated with the corresponding 
operations in $\End_A[1]$ that are obtained from the left $\PP$-action $\gamma_A : \PP\to \End_A$. 
(See \eqref{eqn:alphatree} below for an illustration.)
The transferred $\PP_\infty$-algebra structure on $B\in\Ch^C$ is then given by
the composite $\CoOp_C^{\conil}$-morphism 
\begin{flalign}
\widehat{\beta}\,:\, 
\xymatrix{
\mathsf{B}\PP\ar[r]^-{\widehat{\alpha}}~&~\mathsf{B}^{\mathrm{na}}\End_A \ar[r]^-{\Psi} ~&~\mathsf{B}^{\mathrm{na}}\End_B
}\quad,
\end{flalign}
where $\Psi$ is determined from the homotopy retract \eqref{eqn:homotopyretract} by
the construction given in \cite[Proposition 10.3.2]{LodayVallette}.
Because $\mathsf{B}^{\mathrm{na}}\End_B\in\CoOp_C^{\conil}$ is semi-cofree, this morphism
is fully determined by projecting onto cogenerators
\begin{flalign}\label{eqn:betatransfer}
\Big(\beta\,:\, \xymatrix{
\mathsf{B}\PP\ar[r]^-{\widehat{\alpha}}~&~\mathsf{B}^{\mathrm{na}}\End_A \ar[r]^-{\Psi} ~&~\mathsf{B}^{\mathrm{na}}\End_B
\ar[r]^-{\pr}~&~\End_B[1]
}\Big)\,\in \,\big[\mathsf{B}\PP, \End_{B}[1]\big]^0\quad,
\end{flalign}
which equivalently defines a degree $1$ element $\beta\in [\mathsf{B}\PP, \End_{B}]^1$ that satisfies
$\beta\,\oone_{\mathsf{B}\PP}=0$ and, as a consequence of the Rosetta Stone, 
the Maurer-Cartan equation \eqref{eqn:MCtwisting}. This defines the transferred $\PP_\infty$-algebra
structure on $B\in \Ch^C$.
\sk

Let us illustrate how the map $\beta$ in \eqref{eqn:betatransfer} can be computed by using an example.
Recall that an element in $\mathsf{B}\PP$ is a ($C$-colored rooted) tree whose vertices are ordered and 
decorated by elements in $\overline{\PP}[1]$. The following picture gives an example
\begin{flalign}\label{eqn:tree}
t(\mu_1,\mu_2,\mu_3,\mu_4)~=~
\parbox{5cm}{\begin{tikzpicture}[cir/.style={circle,draw=black,inner sep=0pt,minimum size=2mm},
        poin/.style={rectangle, inner sep=2pt,minimum size=0mm},scale=0.8, every node/.style={scale=0.8}]
\node[poin] (i)   at (0,4) {};
\node[poin] (v1)  at (0,3) {$\mu_1$};
\node[poin] (v2)  at (-1.5,2) {$\mu_2$};
\node[poin] (v3)  at (1.5,2) {$\mu_3$};
\node[poin] (v4)  at (2.5,1) {$\mu_4$};
\node[poin] (o1)  at (-2.5,1) {};
\node[poin] (o2)  at (-1.5,1) {};
\node[poin] (o3)  at (-0.5,1) {};
\node[poin] (o4)  at (0.5,1) {};
\node[poin] (o5)  at (1.5,0) {};
\node[poin] (o6)  at (2.5,0) {};
\node[poin] (o7)  at (3.5,0) {};
\draw[thick] (i) -- (v1);
\draw[thick] (v1) -- (v2);
\draw[thick] (v1) -- (v3);
\draw[thick] (v3) -- (v4);
\draw[thick] (v2) -- (o1);
\draw[thick] (v2) -- (o2);
\draw[thick] (v2) -- (o3);
\draw[thick] (v3) -- (o4);
\draw[thick] (v4) -- (o5);
\draw[thick] (v4) -- (o6);
\draw[thick] (v4) -- (o7);
\end{tikzpicture}}\quad,
\end{flalign}
where for notational convenience we suppress the $C$-colors. (Recall that each $\mu_i$ 
has a tuple $(c_1,\dots,c_n)$ of input colors and an output color $c$.
For each tree the colors on the edges have to match.) The map $\widehat{\alpha}$ 
in \eqref{eqn:betatransfer} assigns to such tree the same tree,
but now decorated by the elements $\gamma_A(\mu_i)$ in $\End_A[1]$ that are obtained 
by applying the left $\PP$-action $\gamma_A : \PP\to\End_A$, i.e.\
\begin{flalign}\label{eqn:alphatree}
\widehat{\alpha}\big(t(\mu_1,\mu_2,\mu_3,\mu_4)\big)~=~
\parbox{5cm}{\begin{tikzpicture}[cir/.style={circle,draw=black,inner sep=0pt,minimum size=2mm},
        poin/.style={rectangle, inner sep=2pt,minimum size=0mm},scale=0.8, every node/.style={scale=0.8}]
\node[poin] (i)   at (0,4) {};
\node[poin] (v1)  at (0,3) {$\gamma_A(\mu_1)$};
\node[poin] (v2)  at (-1.5,2) {$\gamma_A(\mu_2)$};
\node[poin] (v3)  at (1.5,2) {$\gamma_A(\mu_3)$};
\node[poin] (v4)  at (3,1) {$\gamma_A(\mu_4)$};
\node[poin] (o1)  at (-2.5,1) {};
\node[poin] (o2)  at (-1.5,1) {};
\node[poin] (o3)  at (-0.5,1) {};
\node[poin] (o4)  at (0.5,1) {};
\node[poin] (o5)  at (2,0) {};
\node[poin] (o6)  at (3,0) {};
\node[poin] (o7)  at (4,0) {};
\draw[thick] (i) -- (v1);
\draw[thick] (v1) -- (v2);
\draw[thick] (v1) -- (v3);
\draw[thick] (v3) -- (v4);
\draw[thick] (v2) -- (o1);
\draw[thick] (v2) -- (o2);
\draw[thick] (v2) -- (o3);
\draw[thick] (v3) -- (o4);
\draw[thick] (v4) -- (o5);
\draw[thick] (v4) -- (o6);
\draw[thick] (v4) -- (o7);
\end{tikzpicture}}\quad.
\end{flalign}
The transferred $\PP_\infty$-algebra structure $\beta$ in \eqref{eqn:betatransfer} assigns to this tree
the element in $\End_B[1]$ that is obtained by composing the tree
\begin{subequations}\label{eqn:betatree}
\begin{flalign}
\beta(\mu_1,\mu_2,\mu_3,\mu_4)\,:=\, \beta\big(t(\mu_1,\mu_2,\mu_3,\mu_4)\big)~=~
\parbox{5cm}{\begin{tikzpicture}[cir/.style={circle,draw=black,inner sep=0pt,minimum size=2mm},
        poin/.style={rectangle, inner sep=2pt,minimum size=0mm},scale=0.8, every node/.style={scale=0.8},
		every edge quotes/.style = {auto, font=\footnotesize}]
\node[poin] (i)   at (0,4) {};
\node[poin] (v1)  at (0,3) {$\gamma_A(\mu_1)$};
\node[poin] (v2)  at (-1.5,2) {$\gamma_A(\mu_2)$};
\node[poin] (v3)  at (1.5,2) {$\gamma_A(\mu_3)$};
\node[poin] (v4)  at (3,1) {$\gamma_A(\mu_4)$};
\node[poin] (o1)  at (-2.5,1) {};
\node[poin] (o2)  at (-1.5,1) {};
\node[poin] (o3)  at (-0.5,1) {};
\node[poin] (o4)  at (0.5,1) {};
\node[poin] (o5)  at (2,0) {};
\node[poin] (o6)  at (3,0) {};
\node[poin] (o7)  at (4,0) {};
\draw[thick] (i)  edge["$p$"] (v1);
\draw[thick] (v1) edge["$h$"'] (v2);
\draw[thick] (v1) edge["$h$"] (v3);
\draw[thick] (v3) edge["$h$"] (v4);
\draw[thick] (v2) edge["$i$"] (o1);
\draw[thick] (v2) edge["$i$"] (o2);
\draw[thick] (v2) edge["$i$"] (o3);
\draw[thick] (v3) edge["$i$"] (o4);
\draw[thick] (v4) edge["$i$"] (o5);
\draw[thick] (v4) edge["$i$"] (o6);
\draw[thick] (v4) edge["$i$"] (o7);
\end{tikzpicture}}\quad
\end{flalign}
that is formed by decorating the edges of \eqref{eqn:alphatree} with the data of the 
homotopy retract \eqref{eqn:homotopyretract}.
(Note that incoming edges are decorated by $i$, the outgoing edge by $p$ and internal edges by $h$.)
Written in a more standard notation, the result of this composition reads as
\begin{flalign}
\beta(\mu_1,\mu_2,\mu_3,\mu_4)\,=\, p\,\gamma_A(\mu_1)\,
\big(h \gamma_A(\mu_2)\otimes \id\big)\,
\big(\id^{\otimes 3}\otimes h \gamma_A(\mu_3)\big)\,
\big(\id^{\otimes 4}\otimes h\gamma_A(\mu_4)\big)\,i^{\otimes 7}\quad,
\end{flalign}
\end{subequations}
where it is important to note that the order in which the terms are composed is dictated 
by the ordering of vertices of the tree $t(\mu_1,\mu_2,\mu_3,\mu_4)$.
\sk

Given any $C$-colored object $A\in\Ch^C$, we can take 
object-wise cohomology and define a new object $\H A =\{\H A(c)\in\Ch\}_{c\in C}\in \Ch^C$
that has a trivial differential $\dd_{\H A}=0$. Since we are working over a field $\bbK$ of characteristic
$0$, there exists a strong deformation retract
\begin{subequations}\label{eqn:deformationretract}
\begin{equation}\label{eqn:deformationretract1}
\begin{tikzcd}
\H A \ar[r,shift right=-1ex,"i"] & \ar[l,shift right=-1ex,"p"] A \ar[loop,out=-30,in=30,distance=30,swap,"h"]
\end{tikzcd}
\end{equation}
in the category $\Ch^C$, i.e.\ 
\begin{flalign}\label{eqn:deformationretract2}
\partial i\,=\,0~,~~ \partial p \,=\,0 ~,~~ p\, i\,=\,\id ~,~~
\partial h \,=\, i\, p-\id~,~~ h\,i\,=\,0~,~~ p\, h\,=\,0~,~~ h^2\,=\,0\quad.
\end{flalign}
\end{subequations}
The following result is an immediate consequence of the Homotopy Transfer Theorem \ref{theo:HTT}.
\begin{cor}\label{cor:minmod}
Let $\PP= (\PP,\gamma,\oone,\epsilon)\in \Op_C^{\aug}$ be an augmented operad.
Any $\PP_\infty$-algebra structure on $A\in\Ch^C$
can be transferred along the strong deformation retract \eqref{eqn:deformationretract} 
into a $\PP_\infty$-algebra structure on the cohomology $\H A\in\Ch^C$. The latter 
is called a minimal model for the given $\PP_\infty$-algebra $A$. Minimal models are
unique up to $\infty$-isomorphism.
\end{cor}

We conclude this section with a basic observation that will be useful
in the main part of our paper. For this we consider the special case of
an augmented operad $\PP= (\PP,\gamma,\oone,\epsilon)\in \Op_C^{\aug}$ that is concentrated
in cohomological degree zero, i.e.\ all cochain complexes
of operations $\PP\big(\substack{c\\\und{c}}\big)\in\Ch$
are concentrated in degree $0$ and hence, in particular, have a trivial differential $\dd_\PP =0$.
Furthermore, let $A\in\Ch^C$ be endowed with a strict $\PP_\infty$-algebra structure
$\gamma_A : \PP\to \End_A$. We define a $\SymSeq_C$-morphism
\begin{subequations}\label{eqn:transferred0}
\begin{flalign}
\gamma_{\H A}\,:\, \PP \,\longrightarrow\,\End_{\H A}
\end{flalign}
by using the transferred $\PP_\infty$-algebra structure \eqref{eqn:betatree} on $\H A$
and the decomposition $\PP = \overline{\PP}\oplus I_\circ$ that is obtained
from the projector $\oone_\PP\,\epsilon_\PP:\PP\to\PP$
determined by the augmentation and unit of $\PP$. Explicitly, we set
\begin{flalign}
\gamma_{\H A}(\oone_\PP)\,:=\,\oone_{\End_{\H A}}~~,\quad
\gamma_{\H A}(\mu)\,:=\, \beta(\mu)\,=\, p_{c}\,\gamma_A(\mu)\,\bigotimes_{i=1}^ni_{c_i}\,=:\,
 p_{c}\,\gamma_A(\mu)\,i_{\und{c}}\quad,
\end{flalign}
\end{subequations}
for all $\mu\in \overline{\PP}\big(\substack{c\\ (c_1,\dots,c_n)}\big) = 
\overline{\PP}\big(\substack{c\\ \und{c}}\big) $.
\begin{propo}\label{prop:underlyingstrict}
Let $\PP\in \Op_C^{\aug}$ be an augmented operad
concentrated in degree $0$ and $(A,\gamma_A: \PP\to\End_A)\in\Alg_{\PP_\infty}$ a strict $\PP_\infty$-algebra.
Then the map \eqref{eqn:transferred0} is an $\Op_C$-morphism, i.e.\ the cohomology
$\H A\in\Ch^C$ carries as part of its transferred $\PP_\infty$-algebra structure
a strict $\PP_\infty$-algebra structure $\gamma_{\H A} : \PP\to \End_{\H A}$.
\end{propo}
\begin{proof}
The only slightly non-trivial step is to show that $\gamma_{\H A}$ preserves operadic compositions.
For this it is sufficient to consider composable non-identity
operations $\mu\in \overline{\PP}\big(\substack{c\\ \und{c}}\big)$
and $\mu_i\in \overline{\PP}\big(\substack{c_i\\ \und{c}_i}\big)$, for $i=1,\dots,n$, and compute
\begin{flalign}
\nn &\gamma_{\End_{\H A}}\Big(\gamma_{\H A}(\mu), \big(\gamma_{\H A}(\mu_1),\dots,\gamma_{\H A}(\mu_n)\big) \Big)
\,=\, \gamma_{\H A}(\mu)\,\Big(\gamma_{\H A}(\mu_1)\otimes\cdots\otimes \gamma_{\H A}(\mu_n) \Big)\\[4pt]
\,&\qquad\quad =\, p_{c} \,\gamma_A(\mu)\, \Big(i_{c_1}\,p_{c_1}\,\gamma_{A}(\mu_1)\otimes\cdots\otimes
i_{c_n}\,p_{c_n}\,\gamma_{A}(\mu_n) \Big)~
i_{(\und{c}_1,\dots,\und{c}_n)}\quad.
\end{flalign}
It is important to observe that all maps in this expression are cochain maps, 
i.e.\ they are annihilated by the differential $\partial$.
(For this our assumption that $\PP$ has a trivial differential $\dd_\PP=0$ is crucial, because otherwise
$\partial\gamma_A(\mu) = \gamma_A(\dd_\PP \mu)\neq 0$ could be non-trivial.) We can now iteratively insert
$i_{c_i}\,p_{c_i} = \id + \partial h_{c_i}$ and find that, as a consequence of the previous observation,
the terms with $\partial h_{c_i}$ vanish. Hence, we obtain
\begin{flalign}
\nn &\gamma_{\End_{\H A}}\Big(\gamma_{\H A}(\mu), \big(\gamma_{\H A}(\mu_1),\dots,\gamma_{\H A}(\mu_n)\big) \Big)\,=\,
p_{c} \,\gamma_A(\mu)\, \Big(\gamma_{A}(\mu_1)\otimes\cdots\otimes
\gamma_{A}(\mu_n) \Big)~
i_{(\und{c}_1,\dots,\und{c}_n)}\\
&\qquad\quad =\, 
p_{c} \,\gamma_A\Big(\gamma_\PP\big(\mu,(\mu_1,\dots,\mu_n)\big)\Big)~
i_{(\und{c}_1,\dots,\und{c}_n)}\, =\,  \gamma_{\H A}\Big(\gamma_\PP\big(\mu,(\mu_1,\dots,\mu_n)\big)\Big)\quad,
\end{flalign}
where in the second step we used that $\gamma_A$ is a $\PP$-algebra structure.
\end{proof}

\subsection{\label{sec: universal Massey}Universal first-order Massey product}
It is important to emphasize that the strict $\PP_\infty$-algebra structure 
$\gamma_{\H A}:\PP\to\End_{\H A}$  on the cohomology $\H A$ from \eqref{eqn:transferred0}
does in general \textit{not} capture the whole $\PP_\infty$-algebra structure
of the minimal model. There often exist higher structures, called Massey products,
that are related to the values of the transferred $\PP_\infty$-algebra
structure \eqref{eqn:betatree} on trees that have multiple vertices.
Note that, for a tree $t(\mu_1,\dots,\mu_n)\in\mathsf{B}\PP$ with $n$ vertices,
such higher structures take the form of degree $1-n$ elements 
$\beta(\mu_1,\dots,\mu_n)\in \End_{\H A}$ in the endomorphism operad of $\H A$
when $\PP$ is concentrated in degree zero.
\sk

The description and interpretation of these higher structures is rather 
subtle because the individual components $\beta(\mu_1,\dots,\mu_n)\in \End_{\H A}$
are {\em not} invariants of the $\PP_\infty$-algebra $(A,\gamma_A:\PP\to \End_A)\in\Alg_{\PP_\infty}$. In particular, 
they transform non-trivially under $\infty$-isomorphisms when changing the strong deformation retract 
\eqref{eqn:deformationretract} used in the minimal model construction from Corollary \ref{cor:minmod}, 
which involves mixing different components among each other. This issue has been overcome by Dimitrova \cite{Dimitrova}.
She has constructed from the transferred $\PP_\infty$-algebra structure $\beta$
a cohomology class $\big[\beta^{(2)}\big]\in\H^1_\Gamma\big(\H A,\gamma_{\H A}\big)$ in operadic Gamma-cohomology 
of the underlying \textit{strict} $\PP_\infty$-algebra $(\H A,\gamma_{\H A})\in\Alg_{\PP_\infty}$ 
from \eqref{eqn:transferred0} that is an invariant of $(A,\gamma_A)$.
In particular, this cohomology class is independent of the choice of strong deformation retract
used in the minimal model construction.
More recently \cite{UniversalMassey1,UniversalMassey2}, Muro has established a relationship between
this cohomology class and the ordinary concept of first-order Massey products,
thereby coining and justifying the terminology \textit{universal first-order Massey product} 
for Dimitrova's cohomology class $\big[\beta^{(2)}\big]\in\H^1_\Gamma\big(\H A,\gamma_{\H A}\big)$.
A higher-order generalization is given by the
successive obstruction classes in \cite[Theorem 4.6]{Dimitrova}, see also \cite{UniversalMassey3}.
\sk

Let us briefly recall how Dimitrova's class is computed from the transferred
$\PP_\infty$-algebra structure $\beta$ on a minimal model $\H A$.
We will focus mainly on computational aspects, since these will be relevant for
our work, and refer the reader to \cite{Dimitrova} for more details.
A model for the operadic Gamma-cohomology of the \textit{strict} $\PP_\infty$-algebra
$(\H A,\gamma_{\H A})\in\Alg_{\PP_\infty}$ underlying the minimal model \eqref{eqn:transferred0}
is given by the cochain complex
\begin{flalign}
\Gamma(\H A,\gamma_{\H A})\,:=\, \bbR\Der(\H A)\,:=\,\Der\big(Q\H A,\H A\big)\,\in\,\Ch
\end{flalign}
that is obtained from the derived functor of derivations. A concrete model for the 
cofibrant resolution $Q\H A  \SimTo \H A$ in $\Alg_{\PP_\infty}$
can be obtained by using the bar-cobar adjunction
$\Omega_j : \CoAlg_{\mathsf{B}\PP} \rightleftarrows \Alg_{\PP_{\infty}} : \mathsf{B}_j$
for (co)operad (co)algebras that is associated with the canonical
twisting morphism $j : \mathsf{B}\PP\to \Omega\mathsf{B}\PP=\PP_\infty$, see \cite[Section 11]{LodayVallette}.
Explicitly, this cofibrant resolution is given by the counit $\Omega_j\mathsf{B}_j \H A  \SimTo  \H A$
of this adjunction. Using that, by definition, the functors
$\Omega_j = \PP_{\infty}\circ_{j} (-)$ and $\mathsf{B}_j = \mathsf{B}\PP\circ_{j} (-)$
assign semi-(co)free (co)algebras, one finds that
\begin{flalign}\label{eqn:operadicGamma}
\Gamma(\H A,\gamma_{\H A})\,=\,\Der\big(\Omega_j\mathsf{B}_j\H A,\H A\big)\,\cong\, \Big(\big[\mathsf{B}\PP,\End_{\H A}\big]\,,~\partial_\Gamma\Big)\,\in\,\Ch
\end{flalign}
for a suitable differential $\partial_\Gamma$ that is constructed out of the differential $\dd_{\mathsf{B}\PP}$
on the bar construction, the twisting morphism $j : \mathsf{B}\PP\to \Omega\mathsf{B}\PP=\PP_\infty$ and the strict
$\PP_\infty$-algebra structure $\gamma_{\H A}:\PP\to\End_{\H A}$ on $\H A$. 
The details are spelled out in \cite[Section 4.1]{Dimitrova}.
\sk

From the transferred $\PP_\infty$-algebra structure $\beta \in \big[\mathsf{B}\PP,\End_{\H A}\big]^1$ 
on the minimal model,
we define a $1$-cochain $\beta^{(2)}\in \Gamma^1(\H A,\gamma_{\H A})$ in this complex by defining
for each tree $t(\mu_1,\dots,\mu_n)\in \mathsf{B}\PP$ with $n\geq 0$ vertices
\begin{flalign}\label{eqn:Masseycocycle}
\beta^{(2)}(\mu_1,\dots,\mu_n)\,:=\, \begin{cases}
\beta(\mu_1,\mu_2)&~,~~\text{for }n=2\quad,\\
0 &~,~~\text{for } n\neq 2\quad.
\end{cases}
\end{flalign}
Since this cochain is only non-trivial on trees with $2$ vertices,
one finds from the explicit form of $\partial_\Gamma$ 
that the cocycle condition $\partial_\Gamma \beta^{(2)} =0$ 
is only non-trivial on trees with $3$ vertices. There exist
two different types of such $3$-vertex trees, given by
\begin{subequations}\label{eqn:genfirstMasseyrelations}
\begin{flalign}\label{eqn:genfirstMasseyrelations1}
t(\mu_1,\mu_2,\mu_3)~~=~~
\parbox{5cm}{\begin{tikzpicture}[cir/.style={circle,draw=black,inner sep=0pt,minimum size=2mm},
        poin/.style={rectangle, inner sep=2pt,minimum size=0mm},scale=0.8, every node/.style={scale=0.8}]
\node[poin] (i)   at (0,0) {};
\node[poin] (v1)  at (0,-1) {{$\mu_1$}};
\node[poin] (v2)  at (0,-2) {{$\mu_2$}};
\node[poin] (v3)  at (0,-3) {{$\mu_3$}};
\node[poin] (o11)  at (-2,-2) {};
\node[poin] (o12)  at (-1.25,-2) {};
\node[poin] (o13)  at (1.25,-2) {};
\node[poin] (o14)  at (2,-2) {};
\node[poin] (o21)  at (-0.75,-3) {};
\node[poin] (o23)  at (0.75,-3) {};
\node[poin] (o31)  at (-0.5,-4) {};
\node[poin] (o32)  at (0,-4) {};
\node[poin] (o33)  at (0.5,-4) {};
\draw[thick] (i) -- (v1);
\draw[thick] (v1) -- (v2);
\draw[thick] (v1) -- (o11);
\draw[thick] (v1) -- (o12);
\draw[thick] (v1) -- (o13);
\draw[thick] (v1) -- (o14);
\draw[thick] (v2) -- (o21);
\draw[thick] (v2) -- (o23);
\draw[thick] (v2) -- (v3);
\draw[thick] (v3) -- (o31);
\draw[thick] (v3) -- (o32);
\draw[thick] (v3) -- (o33);
\end{tikzpicture}}
\end{flalign} 
and
\begin{flalign}\label{eqn:genfirstMasseyrelations2}
t(\mu_1,\mu_{21},\mu_{22})~~=~~
\parbox{3.5cm}{\begin{tikzpicture}[cir/.style={circle,draw=black,inner sep=0pt,minimum size=2mm},
        poin/.style={rectangle, inner sep=2pt,minimum size=0mm},scale=0.8, every node/.style={scale=0.8}]
\node[poin] (i)   at (0,0) {};
\node[poin] (v1)  at (0,-1) {{$\mu_1$}};
\node[poin] (v2)  at (-1,-2) {{$\mu_{21}$}};
\node[poin] (v3)  at (1,-2) {{$\mu_{22}$}};
\node[poin] (o11)  at (0,-2) {};
\node[poin] (o12)  at (-2,-2) {};
\node[poin] (o13)  at (2,-2) {};
\node[poin] (o21)  at (-1.5,-3) {};
\node[poin] (o22)  at (-1,-3) {};
\node[poin] (o23)  at (-0.5,-3) {};
\node[poin] (o31)  at (1.5,-3) {};
\node[poin] (o32)  at (1,-3) {};
\node[poin] (o33)  at (0.5,-3) {};
\draw[thick] (i) -- (v1);
\draw[thick] (v1) -- (v2);
\draw[thick] (v1) -- (v3);
\draw[thick] (v1) -- (o11);
\draw[thick] (v1) -- (o12);
\draw[thick] (v1) -- (o13);
\draw[thick] (v2) -- (o21);
\draw[thick] (v2) -- (o22);
\draw[thick] (v2) -- (o23);
\draw[thick] (v3) -- (o31);
\draw[thick] (v3) -- (o32);
\draw[thick] (v3) -- (o33);
\end{tikzpicture}}\quad.
\end{flalign} 
\end{subequations}
Evaluating the cocycle condition
$\partial_\Gamma\beta^{(2)}=0$ on these two types 
of trees gives the identities
\begin{subequations}\label{eqn:genMasseyrelations}
\begin{flalign}
\nn0\,&=\, \beta^{(2)}\big(\mu_1,\gamma_{\scriptscriptstyle{(1)}}(\mu_2,\mu_3)\big) -
\beta^{(2)}\big(\gamma_{\scriptscriptstyle{(1)}}(\mu_1,\mu_2),\mu_3\big) \\[3pt]
&\qquad \quad +
\gamma_{\H A}(\mu_1)~\beta^{(2)}(\mu_2,\mu_3) - \beta^{(2)}(\mu_1,\mu_2)\,\gamma_{\H A}(\mu_3)
\label{eqn:genMasseyrelations1}
\end{flalign}
and
\begin{flalign}
\nn 0\,&=\,\beta^{(2)}\big(\gamma_{\scriptscriptstyle{(1)}}(\mu_1,\mu_{22}),\mu_{21}\big) 
-\beta^{(2)}\big(\gamma_{\scriptscriptstyle{(1)}}(\mu_1,\mu_{21}),\mu_{22}\big) \\[3pt]
&\qquad \quad + \beta^{(2)}(\mu_1,\mu_{22})~\gamma_{\H A}(\mu_{21}) 
- \beta^{(2)}(\mu_1,\mu_{21})~\gamma_{\H A}(\mu_{22})
\quad,\label{eqn:genMasseyrelations2}
\end{flalign}
\end{subequations}
which hold true as a consequence of the Maurer-Cartan equation \eqref{eqn:MCtwisting} for the
transferred $\PP_\infty$-algebra structure $\beta$. (Recall that $\gamma_{\scriptscriptstyle{(1)}}$ 
denotes the infinitesimal composition of the operad $\PP$. Graphically, this 
corresponds to composing two vertices in the trees \eqref{eqn:genfirstMasseyrelations}
along a common edge.)
\begin{defi}\label{def:universalMassey}
The universal first-order Massey product of a strict $\PP_\infty$-algebra $(A,\gamma_A:\PP\to\End_A)\in\Alg_{\PP_\infty}$
is defined as the cohomology class
\begin{flalign}
\big[\beta^{(2)}\big]\,\in\,\H^1_\Gamma\big(\H A,\gamma_{\H A}\big)
\end{flalign}
in operadic Gamma-cohomology \eqref{eqn:operadicGamma} that is defined
by the $1$-cocycle \eqref{eqn:Masseycocycle} associated with any choice
of minimal model $(\H A,\beta)\in\Alg_{\PP_\infty}$. This cohomology class 
is an invariant of $(A,\gamma_A)\in\Alg_{\PP}$, in particular
it does not depend on the choice of minimal model that is 
used to compute it, see \cite[Theorem 4.2]{Dimitrova}.
\end{defi}

As usual, two $1$-cocycles $\beta^{(2)},\beta^{\prime(2)}\in \Gamma^1\big(\H A,\gamma_{\H A}\big)$ 
represent the same cohomology class $\big[\beta^{(2)}\big]=\big[\beta^{\prime(2)}\big]
\in \H^1_\Gamma\big(\H A,\gamma_{\H A}\big)$
if and only if they differ by a coboundary, i.e.\
\begin{flalign}\label{eqn:gaugetransformationgeneral}
\beta^{\prime(2)}\,=\, \beta^{(2)} + \partial_{\Gamma} \chi
\end{flalign}
for some $\chi\in \Gamma^0\big(\H A,\gamma_{\H A}\big)$. Requiring that both representatives
$\beta^{(2)}$ and $\beta^{\prime(2)}$ vanish as in \eqref{eqn:Masseycocycle} on all trees
with $n\neq 2$ vertices constrains $\chi$ to be of the form
\begin{flalign}
\chi(\mu_1,\dots,\mu_n)\,=\,\begin{cases}
\chi(\mu_1)&~,~~\text{for }n=1\quad,\\
0 &~,~~\text{for }n\neq 1\quad.
\end{cases}
\end{flalign}
Evaluating both sides of \eqref{eqn:gaugetransformationgeneral} on a tree $t(\mu_1,\dots,\mu_n)$
with $n$ vertices then vanishes for $n\neq 2$, and for $n=2$ it gives the explicit transformation formula
\begin{flalign}\label{eqn:gaugetransformationcomponents}
\beta^{\prime(2)}(\mu_1,\mu_2)\,=\, \beta^{(2)}(\mu_1,\mu_2) -\chi\big(\gamma_{\scriptscriptstyle{(1)}}(\mu_1,\mu_2)\big)
+ \chi(\mu_1)~\gamma_{\H A}(\mu_2) + \gamma_{\H A}(\mu_1)~\chi(\mu_2)
\end{flalign}
for the representing cocycles of the universal first-order Massey product.
As a side-remark, we would like to note that 
this transformation formula can also be interpreted in terms of gauge transformations
of Maurer-Cartan elements as introduced e.g.\ in \cite{DSV}.


\section{\label{sec:PFA}Application to prefactorization algebras}

\subsection{Prefactorization operad and algebras}
We will briefly recall the definition of a prefactorization algebra 
on a smooth $m$-dimensional manifold $M$. More details can be found in the textbooks 
of Costello and Gwilliam \cite{CostelloGwilliam,CostelloGwilliam2}.
Let us start by introducing some relevant terminology. 
An open subset $D\subseteq M$ is called a disk 
if it is diffeomorphic $D\cong \bbR^m$ to the $m$-dimensional Cartesian space. 
We denote by $\mathrm{Disk}_M=\big\{D\subseteq M\big\}$ the set of all disks in $M$.
\begin{defi}\label{def:PFAoperad}
The prefactorization operad $\PP_M^{}= (\PP_M^{},\gamma,\oone)\in\Op_{\mathrm{Disk}_M}^{}$ on a smooth $m$-dimensional
manifold $M$ is the ($\mathrm{Disk}_M$-colored dg-)operad that is defined as follows: 
The underlying symmetric sequence
$\PP_M^{}\in\SymSeq_{\mathrm{Disk}_M}$ is given object-wise by
\begin{flalign}
\PP_M^{}\big(\substack{D\\ \und{D}}\big)\,:=\,\begin{cases}
\bbK\big[\iota_{\und{D}}^{D}\big]\,\in\,\Ch&~,~~\text{if }D_i\subseteq D~\forall i\text{ and }
D_i\cap D_j=\emptyset~\forall {i\neq j}\quad,\\[4pt]
0\,\in\,\Ch&~,~~\text{else}\quad,
\end{cases}
\end{flalign}
for all tuples $(\und{D},D) = ((D_1,\dots,D_n),D)$ of disks,
where $\bbK\big[\iota_{\und{D}}^{D}\big]$ denotes the cochain complex (with trivial differential)
spanned by a degree $0$ element $\iota_{\und{D}}^{D}$, together with
the permutation action
$\PP_M^{}\big(\substack{D\\ \und{D}}\big)\to \PP_M^{}\big(\substack{D\\ \und{D}\sigma}\big)
\,,~\iota_{\und{D}}^{D}\mapsto \iota_{\und{D}\sigma}^{D}$.
The operadic composition $\gamma: \PP_M^{}\circ\PP_M^{}\to \PP_M^{}$ is given by
\begin{flalign}
\gamma\Big(\iota_{\und{D}}^{D},\Big(\iota_{\und{D}_1}^{D_1},\dots,\iota_{\und{D}_{n}}^{D_n}\Big)\Big)\,:=\,\iota^{D}_{(\und{D}_1,\dots,\und{D}_n)}\quad,
\end{flalign} 
and the operadic unit $\oone : I_\circ \to \PP_M^{}$ sends $1 \in I_\circ\big(\substack{D\\ D}\big)$ to $\iota_D^{D} \in \PP_M^{}\big(\substack{D\\ D}\big)$.
\end{defi}

The prefactorization operad is canonically augmented via the augmentation map
\begin{flalign}
\epsilon\,:\, \PP_M^{}\,\longrightarrow\,I_\circ~,~~\iota_{\und{D}}^{D}\,\longmapsto\,\begin{cases}
1&~,~~\text{if }\und{D}=D\quad,\\
0&~,~~\text{else}\quad.
\end{cases}
\end{flalign}
The corresponding augmentation ideal $\overline{\PP}_M^{} = \ker(\epsilon)$ is then given
by all non-identity operations $\iota_{\und{D}}^{D}$, for $\und{D}\neq D$.
Since the prefactorization operad $\PP_M=(\PP_M,\gamma,\oone,\epsilon)\in\Op_{\mathrm{Disk}_M}^{\aug}$ 
is an augmented operad, the homological techniques from Section \ref{sec:prelim} apply to this example.
\begin{defi}\label{def:PFA}
A prefactorization algebra $\FFF$ on a smooth $m$-manifold $M$ is an algebra over the prefactorization operad
$\PP_M=(\PP_M,\gamma,\oone,\epsilon)\in\Op_{\mathrm{Disk}_M}^{\aug}$. More explicitly,
a prefactorization algebra $\FFF\in\Alg_{\PP_M^{}}$ consists of the following data:
\begin{itemize}
\item[(i)] For each disk $D\subseteq M$, a cochain complex $\FFF(D)\in \Ch$.

\item[(ii)] For each tuple $(\und{D},D) = ((D_1,\dots,D_n),D)$ of disks in $M$,
such that $D_i\subseteq D$, for all $i=1,\dots,n$, and $D_{i}\cap D_{j}=\emptyset$, 
for all $i\neq j$, a cochain map
\begin{flalign}
\FFF\big(\iota_{\und{D}}^{D}\big)\,:\, \bigotimes_{i=1}^n \FFF(D_i)~\longrightarrow~\FFF(D)\quad.
\end{flalign}
This includes a cochain map $\FFF\big(\iota_{\varnothing}^D\big) : \bbK\to \FFF(D)$ for the empty tuple $\varnothing=()$ 
and each $D$.
\end{itemize}
These data have to satisfy the following axioms:
\begin{subequations}
\begin{itemize}
\item[(1)] \textit{Preservation of compositions:}
\begin{flalign}\label{eqn:PFAcomposition}
\xymatrix@R=3.5em{
\ar[dr]_-{\FFF\big(\iota_{(\und{D}_1,\dots,\und{D}_{n})}^{D}\big)~~~}
\bigotimes\limits_{i=1}^{n}\bigotimes\limits_{j=1}^{n_i}\FFF(D_{ij})\ar[rr]^-{\bigotimes_i \FFF\big(\iota_{\und{D}_i}^{D_i}\big)}~&&~\bigotimes\limits_{i=1}^{n}\FFF(D_i)
\ar[dl]^-{\FFF\big(\iota_{\und{D}}^{D}\big)}\\
~&~\FFF(D)~&~
}
\end{flalign}
\item[(2)] \textit{Preservation of identities:}
\begin{flalign}
\FFF\big(\iota_{D}^{D}\big)\,=\,\id_{\FFF(D)}\,:\,\FFF(D)~\longrightarrow~\FFF(D)
\end{flalign}

\item[(3)] \textit{Equivariance under permutation actions:}
\begin{flalign}\label{eqn:PFAequivariance}
\xymatrix@R=3.5em{
\ar[dr]_-{\FFF\big(\iota_{\und{D}}^{D}\big)~~}\bigotimes\limits_{i=1}^n\FFF(D_{i}) \ar[rr]^-{\tau_\sigma}~&&~\bigotimes\limits_{i=1}^n\FFF(D_{\sigma(i)}) \ar[ld]^-{\FFF\big(\iota_{\und{D}\sigma}^{D}\big)}\\
~&~\FFF(D)~&~
}
\end{flalign}
with the permutation $\tau_\sigma$ acting by the symmetric braiding \eqref{eqn:braiding}.
\end{itemize}
\end{subequations}
\end{defi}
\begin{rem}\label{rem:operadchoice}
We have already highlighted in the second paragraph of the introduction
that our prefactorization operad from Definition \ref{def:PFAoperad}
agrees with the operad used in the context of factorization homology, 
see e.g.\ \cite[Definition 5.4.5.6 and Remark 5.4.5.7]{HigherAlgebra} and \cite[Definition 2.9]{AyalaFrancis}. 
(To avoid confusion, let us reemphasize that Ayala and Francis consider the symmetric monoidal envelope
$\PP_M^\otimes$ of the operad $\PP_M$, which is the origin of their multidisks.)
Costello and Gwilliam \cite{CostelloGwilliam,CostelloGwilliam2}
consider also alternative (inequivalent) variants of prefactorization operads whose
objects are not only disks $D\subseteq M$, but also multidisks 
$D_1\sqcup \cdots\sqcup D_n\subseteq M$ or even general open subsets $U\subseteq M$.
Allowing for multidisks as part of the objects has the advantage that 
(the non-unital version of) the resulting operad admits a very simple quadratic
presentation that can be shown to be Koszul \cite{PFAKoszul}, while it is presently not known if
the operad from Definition \ref{def:PFAoperad} is Koszul.\footnote{Koszulness is a property
that a dg-operad may or may not have. If an operad $\PP$ is Koszul, then there exists an alternative cofibrant
resolution that is considerably smaller than the bar-cobar resolution from \eqref{eqn:barcobar}.
This leads to a simplified description of $\PP_\infty$-algebras, their $\infty$-morphisms, homotopy transfer 
and minimal models. Hence, the advantage of Koszul operads is that their homological algebra is practically better manageable.}
A disadvantage of such operads including multidisks is that they describe more operations than those one is interested
in. This has led Costello and Gwilliam to introduce a \textit{multiplicativity property}
\cite[Definition 6.1.3]{CostelloGwilliam} which demands that the structure map
\begin{flalign}
\FFF\big(\iota_{(D_1,D_2)}^{D_1 \sqcup D_2}\big) \,:\, \FFF(D_1)\otimes \FFF(D_2)\,\stackrel{\sim}{\longrightarrow}\,\FFF(D_1\sqcup D_2)
\end{flalign}
is a quasi-isomorphism for every pair of disjoint disks $D_1,D_2\subseteq M$.
It is not clear to us how this multiplicativity property interacts with the Koszul property from \cite{PFAKoszul}.
\end{rem}

We will later consider a special class of prefactorization algebras
that correspond to topological quantum field theories.
The precise definition is as follows.
\begin{defi}\label{eqn:PFAlocconst}
A prefactorization algebra $\FFF\in \Alg_{\PP_M^{}}$ is called locally constant
if it assigns to every disk inclusion $D\subseteq D^\prime$ a quasi-isomorphism of cochain complexes
$\FFF(\iota_{D}^{D^\prime}) : \FFF(D)  \SimTo  \FFF(D^\prime)$.
\end{defi}

\subsection{\label{sec: Masseys for PFA}Minimal model and universal first-order Massey product}
Since the prefactorization operad $\PP_M=(\PP_M,\gamma,\oone,\epsilon)\in\Op_{\mathrm{Disk}_M}^{\aug}$ 
is canonically augmented, the construction of minimal models from Corollary \ref{cor:minmod} can be applied 
to any (not necessarily locally constant) prefactorization algebra $\FFF$. Let us briefly specialize
the relevant constructions from Section \ref{sec:prelim} to this specific case.
To compute a minimal model for $\FFF$, we have to choose a strong deformation retract
\begin{equation}\label{eqn:PFASDR}
\begin{tikzcd}
\H\FFF \ar[r,shift right=-1ex,"i"] & \ar[l,shift right=-1ex,"p"] \FFF \ar[loop,out=-30,in=30,distance=30,swap,"h"]
\end{tikzcd}
\end{equation}
in the category $\Ch^{\mathrm{Disk}_{M}}$ to the cohomology
$\H\FFF = \big\{\H\FFF(D)\in\Ch\,:\,D\subseteq M\big\}\in\Ch^{\mathrm{Disk}_{M}}$,
which we regard as a family of cochain complexes with trivial differentials $\dd_{\H\FFF(D)}=0$.
Applying Corollary \ref{cor:minmod} to the strong deformation retract \eqref{eqn:PFASDR} 
yields a transferred $(\PP_{M}^{})_{\infty}^{}$-algebra structure 
$\beta :\mathsf{B}\PP_{M}^{}\to\End_{\H\FFF}$
on the cohomology $\H\FFF\in\Ch^{\mathrm{Disk}_{M}}$.
This means that associated to each $\mathrm{Disk}_{M}$-colored rooted
tree, whose vertices are ordered and decorated by elements in $\overline{\PP}_{M}^{}[1]$,
we have an element in $\End_{\H\FFF}$, i.e.\ an operation on $\H\FFF$ whose arity and colors
matches that of the tree. To illustrate this better, let us spell out a simple
example in full detail, following the presentation in \eqref{eqn:betatree}:  To the tree
\begin{subequations}\label{eqn:PFAtree}
\begin{flalign}
t\big(\iota_{(D_1,D_2)}^{D},\iota^{D_1}_{(D_{11},D_{12},D_{13})}\big)~~=~~~
\parbox{5cm}{\begin{tikzpicture}[cir/.style={circle,draw=black,inner sep=0pt,minimum size=2mm},
        poin/.style={rectangle, inner sep=2pt,minimum size=0mm},scale=0.8, every node/.style={scale=0.8}]
\node[poin] (i)   at (0,0) {};
\node[poin] (v1)  at (0,-1) {{\footnotesize $\iota_{(D_1,D_2)}^{D}$}};
\node[poin] (v2)  at (-1.5,-2) {\footnotesize{$\iota^{D_1}_{(D_{11},D_{12},D_{13})}$}};
\node[poin] (o1)  at (-2.5,-3) {};
\node[poin] (o2)  at (-1.5,-3) {};
\node[poin] (o3)  at (-0.5,-3) {};
\node[poin] (o4)  at (1.5,-2) {};
\draw[thick] (i) -- (v1);
\draw[thick] (v1) -- (v2);
\draw[thick] (v2) -- (o1);
\draw[thick] (v2) -- (o2);
\draw[thick] (v2) -- (o3);
\draw[thick] (v1) -- (o4);
\end{tikzpicture}}
\end{flalign}
is associated the degree
$[1-\text{(number of vertices)}]= -1$ linear map
\begin{flalign}
\beta\big(\iota_{(D_1,D_2)}^{D},\iota^{D_1}_{(D_{11},D_{12},D_{13})}\big)\,:\, \H\FFF(D_{11})\otimes  \H\FFF(D_{12})\otimes   \H\FFF(D_{13})\otimes  \H\FFF(D_{2})~\longrightarrow~\H\FFF(D)
\end{flalign}
that is given by the composite
\begin{flalign}
\nn &\beta\big(\iota_{(D_1,D_2)}^{D},\iota^{D_1}_{(D_{11},D_{12},D_{13})}\big)~=~\\[3pt]
&\qquad~ p_{D}\,\FFF\big(\iota_{(D_1,D_2)}^{D}\big)\,
\Big(h_{D_1} \FFF\big(\iota^{D_1}_{(D_{11},D_{12},D_{13})}\big)\otimes \id_{\FFF(D_2)}\Big)\,\Big(i_{D_{11}}\otimes i_{D_{12}}\otimes i_{D_{13}}\otimes i_{D_2}\Big)\quad.
\end{flalign}
\end{subequations}

Recall that by Proposition \ref{prop:underlyingstrict} the minimal
model $\big(\H \FFF,\beta\big)$ has an underlying strict $(\PP_{M}^{})_\infty^{}$-algebra
structure that is given by evaluating $\beta$ on trees with a single vertex. 
(This does not include the higher structure given by the Massey products,
which we will spell out below.) In the present context, this means that the cohomology $\H \FFF$ 
is a prefactorization algebra with respect to the structure maps
\begin{flalign}\label{eqn:strictPFA}
\H\FFF\big(\iota_{\und{D}}^D\big)\,:=\beta\big(\iota_{\und{D}}^D\big) 
\,=\,\,p_{D}\,\FFF\big(\iota_{\und{D}}^{D}\big)\,i_{\und{D}}
\,:\, \H\FFF(\und{D})~\longrightarrow~\H \FFF(D)\quad,
\end{flalign}
where we introduce the convenient abbreviations $\H\FFF(\und{D}):=
\bigotimes_{i=1}^n\H\FFF(D_i)$ and $i_{\und{D}} := \bigotimes_{i=1}^n i_{D_i}$.
\sk

On top of this strict prefactorization algebra $\H \FFF$
there are potentially non-trivial higher structures given by the Massey products.
The universal first-order Massey product from Definition \ref{def:universalMassey}
is described by the cohomology class $\big[\beta^{(2)}\big]\in \H_\Gamma^1(\H \FFF)$
defined by a certain $1$-cocycle $\beta^{(2)}\in\Gamma^1(\H\FFF)$ that is non-trivial only on trees with $n=2$ vertices. 
In the present case of prefactorization algebras, these trees take the form
\begin{flalign}\label{eqn:universalMasseyPFAtrees}
t\big(\iota_{\und{D}}^{D}, \iota_{\und{D}_i}^{D_i}\big)~~=~~
\parbox{5cm}{\begin{tikzpicture}[cir/.style={circle,draw=black,inner sep=0pt,minimum size=2mm},
        poin/.style={rectangle, inner sep=2pt,minimum size=0mm},scale=0.8, every node/.style={scale=0.8}]
\node[poin] (i)   at (0,0) {};
\node[poin] (v1)  at (0,-1) {{$\iota_{\und{D}}^{D}$}};
\node[poin] (v2)  at (0,-2) {{$\iota^{D_i}_{\und{D}_i}$}};
\node[poin] (o11)  at (-2,-2) {};
\node[poin] (o12)  at (-1.25,-2) {};
\node[poin] (o13)  at (1.25,-2) {};
\node[poin] (o14)  at (2,-2) {};
\node[poin] (o21)  at (-0.5,-3) {};
\node[poin] (o22)  at (0,-3) {};
\node[poin] (o23)  at (0.5,-3) {};
\draw[thick] (i) -- (v1);
\draw[thick] (v1) -- (v2);
\draw[thick] (v1) -- (o11);
\draw[thick] (v1) -- (o12);
\draw[thick] (v1) -- (o13);
\draw[thick] (v1) -- (o14);
\draw[thick] (v2) -- (o21);
\draw[thick] (v2) -- (o22);
\draw[thick] (v2) -- (o23);
\end{tikzpicture}}
\end{flalign} 
and the non-vanishing components of the $1$-cocycle reads as
\begin{flalign}\label{eqn:universalMasseyPFA}
\beta^{(2)}\big(\iota_{\und{D}}^{D}, \iota_{\und{D}_i}^{D_i}\big)
\,=\,\beta\big(\iota_{\und{D}}^{D}, \iota_{\und{D}_i}^{D_i}\big)\,=\,
p_{D}~\FFF\big(\iota_{\und{D}}^{D}\big)~h_{D_i}~\FFF\big(\iota_{\und{D}_i}^{D_i}\big)~i_{(D_1,\dots,\und{D}_i,\dots, D_n)}\quad,
\end{flalign}
where here and below we use a condensed notation suppressing all tensor products
with identity morphisms. As illustrated in \eqref{eqn:gaugetransformationcomponents}, 
different representatives (that vanish on all trees with $n\neq 2$ vertices) 
of the cohomology class $\big[\beta^{(2)}\big]\in \H_\Gamma^1(\H \FFF)$
are related by the transformation formula
\begin{flalign}\label{eqn:gaugetransformationPFA}
\beta^{\prime(2)}\big(\iota_{\und{D}}^{D}, \iota_{\und{D}_i}^{D_i}\big)\,=\, 
\beta^{(2)}\big(\iota_{\und{D}}^{D}, \iota_{\und{D}_i}^{D_i}\big)
-\chi\big(\iota^D_{(D_1,\dots,\und{D}_i,\dots,D_n)}\big)
+ \chi\big(\iota_{\und{D}}^{D}\big)~\H\FFF(\iota_{\und{D}_i}^{D_i}\big) 
+ \H\FFF\big(\iota_{\und{D}}^{D}\big)~\chi\big(\iota_{\und{D}_i}^{D_i}\big) \quad,
\end{flalign}
where $\chi\in \Gamma^0(\H\FFF)$ is any $0$-cochain that vanishes on all trees with $n\neq 1$ vertices.
\sk

In our present high level of generality, it is difficult (if not impossible)
to make any non-trivial statements about the universal first-order Massey
product associated with a prefactorization algebra $\FFF$. The defining cohomology
class $\big[\beta^{(2)}\big]\in \H_\Gamma^1(\H \FFF)$ can (in principle) be determined
for any example of $\FFF$ by choosing a strong deformation retract \eqref{eqn:PFASDR}
and using the formula in \eqref{eqn:universalMasseyPFA}. The general theory
from Section \ref{sec:prelim} implies that the result will be an invariant
of $\FFF$, in particular it will not depend on any of the choices made.
Of course, it will strongly depend on specific details of 
the prefactorization algebra $\FFF$ whether or not the class
$\big[\beta^{(2)}\big]\in \H_\Gamma^1(\H \FFF)$ is trivial.

\subsection{\label{sec: lcPFA}Locally constant prefactorization algebras on $\bbR^m$}
Some non-trivial statements about the minimal model and
universal first-order Massey product can be made for
locally constant prefactorization algebras on the Cartesian space $\bbR^m$.
The aim of this subsection is to explore some general consequences
arising from local constancy, which will become useful in our study
of explicit examples in Sections \ref{sec:examples} and \ref{sec:CS} below.
Throughout this subsection, we fix any locally constant 
prefactorization algebra $\FFF\in\Alg_{\PP_{\bbR^m}}$ on the 
$m$-dimensional Cartesian space $\bbR^m$. 
We also choose an orientation of $\bbR^m$.
\sk

As a simple first observation, we note that the
transferred (strict) prefactorization algebra $\H\FFF$ in \eqref{eqn:strictPFA} 
is locally constant as a consequence of local constancy of $\FFF$. More explicitly,
given any disk inclusion $D\subset D^\prime$, the cochain map
$\FFF\big(\iota_D^{D^\prime}\big):\FFF(D)\SimTo \FFF(D^\prime)$ is a quasi-isomorphism,
and hence by definition the induced map on cohomology
\begin{flalign} \label{eqn:strictPFAisos}
\H\FFF\big(\iota_{D}^{D^\prime}\big) : \H\FFF(D)\overset{\cong}\longrightarrow \H\FFF(D^\prime) \quad,
\end{flalign}
given by \eqref{eqn:strictPFA}, is an isomorphism.
Thus $\H\FFF$ is locally constant even in the stricter sense
given by replacing in Definition \ref{eqn:PFAlocconst} the concept 
of quasi-isomorphisms by actual isomorphisms.
\sk

Using the isomorphisms \eqref{eqn:strictPFAisos}, we can canonically identify
$\H\FFF\big(\iota_D^{\bbR^m}\big): \H\FFF(D) \CongTo \H\FFF(\bbR^m)$
the object $\H\FFF(D)\in\Ch$, for any disk $D\subseteq \bbR^m$, with the object
$\H\FFF(\bbR^m)\in\Ch$ that is assigned to the whole Cartesian space $\bbR^m$. 
Considering the commutative diagram 
\begin{flalign}\label{eqn:higherarityoperations}
\xymatrix@C=5em@R=3em{
\H\FFF(\bbR^m)^{\otimes n} \ar[r]^-{\mu_{\und{D}}}~&~ \H\FFF(\bbR^m)\\
\H\FFF(\und{D})\ar[u]_-{\cong}^-{\bigotimes_i\H\FFF(\iota_{D_i}^{\bbR^m})}\ar[r]_-{\H\FFF(\iota_{\und{D}}^D)}
\ar[ru]^-{\H\FFF(\iota_{\und{D}}^{\bbR^m})}~&~\H\FFF(D)\ar[u]^{\cong}_-{\H\FFF(\iota_D^{\bbR^m})}
}
\end{flalign}
allows us to regard the structure maps of $\H\FFF$ from 
\eqref{eqn:strictPFA} equivalently as operations $\mu_{\und{D}}$ on $\H\FFF(\bbR^m)$ whose
arity is given by the length of the tuple of disks $\und{D}= (D_1,\dots,D_n)$.
Using the composition property \eqref{eqn:PFAcomposition} of a prefactorization algebra,
one directly checks from the defining diagram \eqref{eqn:higherarityoperations} 
that $\mu_{\und{D}}=\mu_{\und{\widetilde{D}}}$ for every
pair of tuples  $\und{D} = (D_1,\dots,D_n)$ and $\und{\widetilde{D}} = (\widetilde{D}_1,\dots,\widetilde{D}_n)$
of mutually disjoint disks such that $\widetilde{D}_i\subseteq D_i$, for all $i=1,\dots,n$.
As a consequence, we find that two operations $\mu_{\und{D}}$ and $\mu_{\und{D}^\prime}$ 
coincide whenever the tuples of mutually disjoint disks $\und{D} = (D_1,\dots,D_n)$ and
$\und{D}^\prime = (D^\prime_1,\dots,D^\prime_n)$ can be connected by a finite chain of zig-zags
\begin{flalign}\label{eqn:zigzag}
\xymatrix{
\und{D} ~&~\ar[l] \und{D}_0\ar[r] ~&~\und{D}_1 ~&~ \ar[l] \cdots \ar[r]~&~ \und{D}_{N-1}~&~ \ar[l] \und{D}_N \ar[r] ~&~\und{D}^\prime
}
\end{flalign}
of families of single-disk inclusions. The existence of such zig-zags depends
on the dimension $m$ of the Cartesian space $\bbR^m$:
\begin{itemize}
\item On the $1$-dimensional Cartesian space $\bbR^1$, given two $n$-tuples $\und{D}$ 
and $\und{D}^\prime$ of mutually disjoint disks (i.e.\ intervals), there exist unique 
permutations $\sigma,\sigma^\prime\in\Sigma_n$
such that the permuted tuples $\und{D}\sigma$ and $\und{D}^\prime\sigma^\prime$ are ordered
with respect to the choice of orientation of $\bbR^1$, i.e.\ $D_{\sigma(1)} < D_{\sigma(2)} < \cdots <D_{\sigma(n)}$
and $D^\prime_{\sigma^\prime(1)} < D^\prime_{\sigma^\prime(2)} < \cdots < D^\prime_{\sigma^\prime(n)}$.
The tuples $\und{D}$ and $\und{D}^\prime$ can be connected by a chain of zig-zags \eqref{eqn:zigzag}
if and only if $\sigma=\sigma^\prime$.

\item On the $(m\geq 2)$-dimensional Cartesian space $\bbR^m$, any two $n$-tuples
$\und{D}$ and $\und{D}^\prime$ of mutually disjoint disks can be connected by a chain of zig-zags \eqref{eqn:zigzag}.
\end{itemize}

These observations allow us to identify the algebraic structure that is
determined by the transferred structure maps in \eqref{eqn:strictPFA} and \eqref{eqn:higherarityoperations}.
\begin{propo}\label{prop:strictstructuremaps}
We have:
\begin{itemize}
\item[$(1)$] In $m=1$ dimension,
the structure maps $\mu_{\und{D}}=: \mu_\sigma$ in \eqref{eqn:higherarityoperations} 
depend only on the permutation $\sigma\in\Sigma_n$ that orders
$\und{D}\sigma=(D_{\sigma(1)},\dots,D_{\sigma(n)})$ 
along the orientation. The resulting family of operations $\{\mu_\sigma\,:\, \sigma\in\Sigma_n\,,~n\geq 0\}$ 
defines the structure of an associative and unital algebra on $\H\FFF(\bbR^1)$.

\item[$(2)$] In $m\geq 2$ dimensions, the structure maps 
$\mu_{\und{D}}=: \mu_n$ in \eqref{eqn:higherarityoperations} 
depend only on the length of the tuple of disks $\und{D} = (D_1,\dots,D_n)$. 
The resulting family of operations $\{\mu_n\,:\, n\geq 0\}$ 
defines the structure of an associative, unital and commutative algebra on $\H\FFF(\bbR^m)$.
\end{itemize}
\end{propo}
\begin{proof}
The associative and unital algebra axioms for $\{\mu_\sigma\,:\, \sigma\in\Sigma_n\,,~n\geq 0\}$
follow directly from the prefactorization algebra axioms in Definition \ref{def:PFA},
and so do the associative, unital and commutative algebra axioms for $\{\mu_n\,:\, n\geq 0\}$.
\end{proof}

In order to make some non-trivial statements about the universal first-order Massey product
of a locally constant prefactorization algebra $\FFF\in\Alg_{\PP_{\bbR^m}}$,
we leverage the fact that $\H\FFF$ is strictly locally constant \eqref{eqn:strictPFAisos}
in order to improve the general choice \eqref{eqn:PFASDR} of a strong deformation retract.
(Recall that changing the strong deformation retract does not change the cohomology class 
$\big[\beta^{(2)}\big]\in \H_\Gamma^1(\H \FFF)$ that defines the universal Massey product.)
\begin{lem}\label{lem:PFASDRimproved}
Given any strong deformation retract \eqref{eqn:PFASDR} for a locally constant
prefactorization algebra $\FFF$ on $M=\bbR^m$, then the components
\begin{subequations}\label{eqn:PFASDRimproved}
\begin{flalign}
\widetilde{i}_D\,&:=\, i_{D}\,\H\FFF\big(\iota_D^{\bbR^m}\big)^{-1}\quad,\\
\widetilde{p}_D\,&:=\, p_{\bbR^m}\,\FFF\big(\iota_{D}^{\bbR^m}\big)\quad,\\
\widetilde{h}_D\,&:=\, h_D - i_D\,\H\FFF\big(\iota_D^{\bbR^m}\big)^{-1}\,p_{\bbR^m}\,\FFF\big(\iota_D^{\bbR^m}\big)\,h_D\quad,
\end{flalign}
for all disks $D\in\mathrm{Disk}_{\bbR^m}$, define a strong deformation retract
\begin{equation}
\begin{tikzcd}
\H\FFF(\bbR^m) \ar[r,shift right=-1ex,"\widetilde{i}"] & \ar[l,shift right=-1ex,"\widetilde{p}"] \FFF \ar[loop,out=-30,in=30,distance=30,swap,"\widetilde{h}"]
\end{tikzcd}
\end{equation}
\end{subequations}
to the constant object $\H\FFF(\bbR^m) := 
\big\{\H\FFF(\bbR^m)\in\Ch\,:\,D\in \mathrm{Disk}_{\bbR^m}\big\}\in\Ch^{\mathrm{Disk}_{\bbR^m}}$.
\end{lem}
\begin{proof}
It is a straightforward algebraic check to show that 
the strong deformation retract properties \eqref{eqn:deformationretract2}
for $(\widetilde{i},\widetilde{p},\widetilde{h})$ are inherited 
from the strong deformation retract properties of $(i,p,h)$. 
\end{proof}

Using the improved strong deformation retract \eqref{eqn:PFASDRimproved} results in various simplifications
in the computation of a minimal model for $\FFF$, and hence the universal Massey product.
Before explaining these simplifications, let us note that the associated underlying
strict $(\PP_{\bbR^m})_{\infty^{}}$-algebra structure reads as
\begin{flalign}
\widetilde{\beta}\big(\iota_{\und{D}}^D\big) \,=\, 
\,\widetilde{p}_{D}\,\FFF\big(\iota_{\und{D}}^D\big)\,\widetilde{i}_{\und{D}}
\,=\,p_{\bbR^m}\,\FFF\big(\iota^{\bbR^m}_{\und{D}}\big)\,i_{\und{D}}~\bigg(\bigotimes_{i=1}^n\H\FFF\big(\iota_{D_i}^{\bbR^m}\big)^{-1}\bigg)\,=\,\mu_{\und{D}}\quad,
\end{flalign}
i.e.\ it agrees precisely with the cohomology operations from \eqref{eqn:higherarityoperations}.
For the $1$-cocycle representing the universal first-order Massey product
of $\FFF$, we find
\begin{flalign}
\nn \widetilde{\beta}^{(2)}\big(\iota_{\und{D}}^D,\iota_{\und{D}_i}^{D_i}\big)\,&=\,
\widetilde{p}_{D}~\FFF\big(\iota_{\und{D}}^{D}\big)~\widetilde{h}_{D_i}~\FFF\big(\iota_{\und{D}_i}^{D_i}\big)~\widetilde{i}_{(D_1,\dots,\und{D}_i,\dots, D_n)}\\
\nn \,&=\, p_{\bbR^m}\, \FFF\big(\iota_{\und{D}}^{\bbR^m}\big)~h_{D_i}~\FFF\big(\iota_{\und{D}_i}^{D_i}\big)~\widetilde{i}_{(D_1,\dots,\und{D}_i,\dots, D_n)}\\
\nn &\qquad - p_{\bbR^m}\, \FFF\big(\iota_{\und{D}}^{\bbR^m}\big)\, i_{D_i}\,\H\FFF\big(\iota_{D_i}^{\bbR^m}\big)^{-1}\,p_{\bbR^m}\,\FFF\big(\iota_{D_i}^{\bbR^m}\big)\,h_{D_i}
~\FFF\big(\iota_{\und{D}_i}^{D_i}\big)~\widetilde{i}_{(D_1,\dots,\und{D}_i,\dots, D_n)}\\
\,&=\,\widehat{\beta}^{(2)}\big(\iota_{\und{D}}^{\bbR^m},\iota_{\und{D}_i}^{D_i}\big)
- \mu_{\und{D}}~\widehat{\beta}^{(2)}\big(\iota_{D_i}^{\bbR^m},\iota_{\und{D}_i}^{D_i}\big)\quad,\label{eqn:improvedMassey}
\end{flalign}
where $\widehat{\beta}^{(2)}$ is defined from the
$1$-cocycle $\beta^{(2)}$ associated with the original strong deformation retract \eqref{eqn:universalMasseyPFA}
by the commutative diagram
\begin{flalign}\label{eqn:firstMasseyonRm}
\xymatrix@C=7em@R=3em{
\H\FFF(\bbR^m)^{\otimes (n+n_i-1)} \ar[r]^-{\widehat{\beta}^{(2)}\big(\iota_{\und{D}}^{D}, \iota_{\und{D}_i}^{D_i}\big)} ~&~ \H\FFF(\bbR^m)\\\
\ar[u]_-{\cong}^-{\H\FFF(\iota_{D_1}^{\bbR^m})\otimes\cdots\otimes 
\bigotimes\limits_{j=1}^{n_i}\!\big(\H\FFF(\iota_{D_{ij}}^{\bbR^m})\big)\otimes\cdots\otimes\H\FFF(\iota_{D_n}^{\bbR^m})} \H\FFF\big(D_1,\dots,\und{D}_i,\dots, D_n \big) \ar[r]_-{\beta^{(2)}\big(\iota_{\und{D}}^{D}, \iota_{\und{D}_i}^{D_i}\big)} ~&~
\H\FFF(D)\ar[u]^-{\cong}_-{\H\FFF(\iota_D^{\bbR^m})}
}
\end{flalign}
Note that these are the same canonical identifications that we have used for
defining the cohomology operations $\mu_{\und{D}}$ in \eqref{eqn:higherarityoperations}.
\begin{propo} \label{prop: improved Masseys}
The $1$-cocycle $\widetilde{\beta}^{(2)}\in\Gamma^1\big(\H\FFF\big)$ in \eqref{eqn:improvedMassey}
that is associated with the improved strong deformation retract \eqref{eqn:PFASDRimproved} has the 
following properties:
\begin{itemize}
\item[$(1)$] It is constant in the outgoing disk, i.e.\
\begin{flalign}
\widetilde{\beta}^{(2)}\big(\iota_{\und{D}}^D,\iota_{\und{D}_i}^{D_i}\big)\,=\,
\widetilde{\beta}^{(2)}\big(\iota_{\und{D}}^{\bbR^m},\iota_{\und{D}_i}^{D_i}\big)\quad,
\end{flalign}
for all $D\in\mathrm{Disk}_{\bbR^m}$.

\item[$(2)$] It vanishes 
\begin{flalign}
\widetilde{\beta}^{(2)}\big(\iota_{D}^{D^\prime}, \iota_{\und{D}}^{D}\big)\,=\,0
\end{flalign}
on all trees of the form
\begin{flalign}\label{eqn:forktrees}
t\big(\iota_{D}^{D^\prime}, \iota_{\und{D}}^{D}\big)~~=~~
\parbox{1cm}{\begin{tikzpicture}[cir/.style={circle,draw=black,inner sep=0pt,minimum size=2mm},
        poin/.style={rectangle, inner sep=2pt,minimum size=0mm},scale=0.8, every node/.style={scale=0.8}]
\node[poin] (i)   at (0,0) {};
\node[poin] (v1)  at (0,-1) {{$\iota_{D}^{D^\prime}$}};
\node[poin] (v2)  at (0,-2) {{$\iota^{D}_{\und{D}}$}};
\node[poin] (o21)  at (-0.5,-3) {};
\node[poin] (o22)  at (0,-3) {};
\node[poin] (o23)  at (0.5,-3) {};
\draw[thick] (i) -- (v1);
\draw[thick] (v1) -- (v2);
\draw[thick] (v2) -- (o21);
\draw[thick] (v2) -- (o22);
\draw[thick] (v2) -- (o23);
\end{tikzpicture}}\quad.
\end{flalign}

\item[$(3)$] Its value $\widetilde{\beta}^{(2)}\big(\iota_{\und{D}}^D,\iota_{\und{D}_i}^{D_i}\big)$ 
on any $2$-vertex tree is determined by a linear combination of its values
$\widetilde{\beta}^{(2)}\big(\iota_{(D_1,D_2)}^{D}, \iota_{\und{D}_1}^{D_1}\big)$ on trees of the form
\begin{flalign}\label{eqn:complicatedtrees}
t\big(\iota_{(D_1,D_2)}^{D}, \iota_{\und{D}_1}^{D_1}\big)~~=~~
\parbox{2cm}{\begin{tikzpicture}[cir/.style={circle,draw=black,inner sep=0pt,minimum size=2mm},
        poin/.style={rectangle, inner sep=2pt,minimum size=0mm},scale=0.8, every node/.style={scale=0.8}]
\node[poin] (i)   at (0,0) {};
\node[poin] (v1)  at (0,-1) {{$\iota_{(D_1,D_2)}^{D}$}};
\node[poin] (v2)  at (-0.75,-2) {{$\iota^{D_1}_{\und{D}_1}$}};
\node[poin] (o1)  at (0.75,-2) {};
\node[poin] (o21)  at (-0.25,-3) {};
\node[poin] (o22)  at (-0.75,-3) {};
\node[poin] (o23)  at (-1.25,-3) {};
\draw[thick] (i) -- (v1);
\draw[thick] (v1) -- (v2);
\draw[thick] (v1) -- (o1);
\draw[thick] (v2) -- (o21);
\draw[thick] (v2) -- (o22);
\draw[thick] (v2) -- (o23);
\end{tikzpicture}}\quad,
\end{flalign} 
where $\und{D}_1 = (D_{11},\dots,D_{1n_1})$ is a tuple of arbitrary length $n_1\geq 0$.
\end{itemize}
\end{propo}
\begin{proof}
Item (1) is a direct consequence of the explicit formula \eqref{eqn:improvedMassey}, which is manifestly 
independent of the disk $D$. Item (2) is proven by a short calculation using the same formula
\begin{flalign}
\widetilde{\beta}^{(2)}\big(\iota_{D}^{D^\prime}, \iota_{\und{D}}^{D}\big)
\,=\,\widehat{\beta}^{(2)}\big(\iota_{D}^{\bbR^m},\iota_{\und{D}}^{D}\big)
- \mu_{D}~\widehat{\beta}^{(2)}\big(\iota_{D}^{\bbR^m},\iota_{\und{D}}^{D}\big) \,=\,0\quad,
\end{flalign}
where in the last step we used that $\mu_{D}=\id$ as a consequence of Proposition \ref{prop:strictstructuremaps}.
\sk

To prove item (3), let us first observe that, as a consequence of item (2),
the $1$-cocycle $\widetilde{\beta}^{(2)}$ can only be non-vanishing
on trees \eqref{eqn:universalMasseyPFAtrees} with at least one free edge at the top vertex.
So the statement we have to prove is that its value on any such tree can be determined
as a linear combination of its values on trees of the form \eqref{eqn:complicatedtrees}. Note that the latter have
precisely one free edge at the top vertex that moreover points to the right.
The proof is slightly different in dimensions $m=1$ 
and $m\geq 2$, so we will treat these cases separately.
\sk

\noindent \textit{\underline{Item $(3)$ in dimension $m\geq 2$:}} Using permutation actions, 
it suffices to consider trees of the form
\begin{flalign}\label{eqn:lefttree}
t\big(\iota_{(D_1,\dots,D_n)}^{D}, \iota_{\und{D}_1}^{D_1}\big)~~=~~
\parbox{2.5cm}{\begin{tikzpicture}[cir/.style={circle,draw=black,inner sep=0pt,minimum size=2mm},
        poin/.style={rectangle, inner sep=2pt,minimum size=0mm},scale=0.8, every node/.style={scale=0.8}]
\node[poin] (i)   at (0,0) {};
\node[poin] (v1)  at (0,-1) {{$\iota_{(D_1,\dots,D_n)}^{D}$}};
\node[poin] (v2)  at (-1,-2) {{$\iota^{D_1}_{\und{D}_1}$}};
\node[poin] (o12)  at (0,-2) {};
\node[poin] (o13)  at (0.75,-2) {};
\node[poin] (o14)  at (1.5,-2) {};
\node[poin] (o21)  at (-1.5,-3) {};
\node[poin] (o22)  at (-1,-3) {};
\node[poin] (o23)  at (-0.5,-3) {};
\draw[thick] (i) -- (v1);
\draw[thick] (v1) -- (v2);
\draw[thick] (v1) -- (o12);
\draw[thick] (v1) -- (o13);
\draw[thick] (v1) -- (o14);
\draw[thick] (v2) -- (o21);
\draw[thick] (v2) -- (o22);
\draw[thick] (v2) -- (o23);
\end{tikzpicture}}\quad,
\end{flalign} 
where the bottom vertex connects to the left disk $D_1$. Choosing any disk $D^\prime\subseteq\bbR^m$
such that  $\bigcup_{i=1}^{n-1}D_i \subset D^\prime \subset D\setminus D_n$,
we can factorize at the top vertex and obtain the $3$-vertex tree
\begin{flalign}
t\big(\iota_{(D^\prime,D_n)}^{D}, \iota^{D^\prime}_{(D_1,\dots,D_{n-1})}, \iota_{\und{D}_1}^{D_1}\big)\quad.
\end{flalign}
Applying the cocycle condition \eqref{eqn:genMasseyrelations1} for $\widetilde{\beta}^{(2)}$ 
to this tree gives the identity
\begin{flalign}
\nn \widetilde{\beta}^{(2)}\big(\iota_{(D_1,\dots,D_n)}^{D}, \iota_{\und{D}_1}^{D_1}\big)\,&=\,
\widetilde{\beta}^{(2)}\big(\iota^{D}_{(D^\prime,D_n)}, \iota^{D^\prime}_{(\und{D}_1,\dots,D_{n-1})}\big) 
+ \mu_{(D^\prime,D_n)}~\widetilde{\beta}^{(2)}\big(\iota^{D^\prime}_{(D_1,\dots,D_{n-1})},\iota^{D_1}_{\und{D}_1}\big) \\[3pt]
&\qquad - \widetilde{\beta}^{(2)}\big(\iota^{D}_{(D^\prime,D_n)},\iota^{D^\prime}_{(D_1,\dots,D_{n-1})}\big)~\mu_{\und{D_1}}
\end{flalign}
that expresses the value of  $\widetilde{\beta}^{(2)}$ on the tree \eqref{eqn:lefttree} 
as a linear combination of its values on trees whose top vertex has arity $n-1$ or $2$.
Applying this procedure iteratively allow us to reduce the arity of the top vertex to $2$, which completes the proof.
\sk

\noindent \textit{\underline{Item $(3)$ in dimension $m=1$:}} Using permutation actions, 
it suffices to consider trees \eqref{eqn:universalMasseyPFAtrees} 
for which the $(n+n_i-1)$-tuple $(D_1,\dots,\und{D}_i,\dots,D_n)$ is either ordered or reverse-ordered 
with respect to the orientation of $\bbR^1$. Choosing any disk (i.e.\ interval) $D^\prime\subseteq\bbR^1$
such that  $\bigcup_{j=i}^{n}D_j \subset D^\prime \subset D\setminus \bigcup_{j=1}^{i-1}D_j$,
we can factorize at the top vertex and obtain the $3$-vertex tree
\begin{flalign}
t\big(\iota_{(D_1,\dots,D_{i-1},D^\prime)}^{D}, \iota^{D^\prime}_{(D_i,\dots,D_{n})}, \iota_{\und{D}_i}^{D_i}\big)\quad.
\end{flalign}
Applying the cocycle condition \eqref{eqn:genMasseyrelations1} for $\widetilde{\beta}^{(2)}$ 
to this tree gives the identity
\begin{flalign}
\nn \widetilde{\beta}^{(2)}\big(\iota_{(D_1,\dots,D_n)}^{D}, \iota_{\und{D}_i}^{D_i}\big)\,&=\,
\widetilde{\beta}^{(2)}\big(\iota^{D}_{(D_1,\dots,D_{i-1},D^\prime)}, \iota^{D^\prime}_{(\und{D}_i,\dots,D_{n})}\big) 
+ \mu_{(D_1,\dots,D_{i-1},D^\prime)}~\widetilde{\beta}^{(2)}\big(\iota^{D^\prime}_{(D_i,\dots,D_{n})}, \iota_{\und{D}_i}^{D_i}\big) \\[3pt]
&\qquad - \widetilde{\beta}^{(2)}\big(\iota_{(D_1,\dots,D_{i-1},D^\prime)}^{D}, \iota^{D^\prime}_{(D_i,\dots,D_{n})}\big)~\mu_{\und{D}_i}
\end{flalign}
that expresses the value of  $\widetilde{\beta}^{(2)}$ on the tree \eqref{eqn:universalMasseyPFAtrees}
as a linear combination of its values on trees whose top vertex has free edges only to the left or to the right.
If the free edges are to the left, we use again permutation actions to reverse the order and bring all the
free edges to the right. Using now the same argument as in dimension $m\geq 2$,
we can reduce iteratively the arity of the top vertex to arrive at arity $2$, which completes the proof.
\end{proof}

\begin{rem}\label{rem:computingMassey}
The relevance of this proposition is that it allows us to simplify
the computation of the universal first-order Massey product
$\big[\widetilde{\beta}^{(2)}\big]\in\H^1_\Gamma\big(\H\FFF\big)$
by working with the improved strong deformation retract \eqref{eqn:PFASDRimproved}.
Item (2) states that the value of the representative $1$-cocycle $\widetilde{\beta}^{(2)}$
is zero on all trees of the form \eqref{eqn:forktrees}, hence one does 
not have to consider these trees. 
Item (3) shows that all operations $\widetilde{\beta}^{(2)}\big(\iota_{\und{D}}^D,\iota_{\und{D}_i}^{D_i}\big)$
can be expressed using 
trees of the form \eqref{eqn:complicatedtrees}, i.e.\ with a single free edge on 
the top vertex that points to the right, so any explicit computation 
should start with understanding the values of $\widetilde{\beta}^{(2)}$ on such trees. 
In the case one finds that these values are all zero, then item (3) implies that
the cohomology class $\big[\widetilde{\beta}^{(2)}\big]=0$ vanishes, i.e.\ 
the universal first-order Massey product is trivial. It is important
to stress that the opposite conclusion is in general not true:
If the value of $\widetilde{\beta}^{(2)}$ is non-zero on some of the trees
\eqref{eqn:complicatedtrees}, it still can happen that $\big[\widetilde{\beta}^{(2)}\big]=0$
if there exists a coboundary $\partial_\Gamma\chi$ that transforms
these components to zero via \eqref{eqn:gaugetransformationPFA}. 
In order to preserve the constancy property in item (1),
the $0$-cochain $\chi\in\Gamma^0\big(\H\FFF\big)$ must satisfy the condition
\begin{flalign}\label{eqn:improvedgauge1}
\chi\big(\iota^{\bbR^m}_{(D_1,\dots,\und{D}_i,\dots,D_n)}\big) -
\chi\big(\iota^{D}_{(D_1,\dots,\und{D}_i,\dots,D_n)}\big) 
\,=\,
\chi\big(\iota^{\bbR^m}_{\und{D}}\big)\,\mu_{\und{D}_i} -
\chi\big(\iota^{D}_{\und{D}}\big)\,\mu_{\und{D}_i}\quad,
\end{flalign}
for all trees \eqref{eqn:universalMasseyPFAtrees},
and to preserve the vanishing property in item (2), it must satisfy
\begin{flalign}\label{eqn:improvedgauge2}
\chi\big(\iota^{D^\prime}_{\und{D}}\big) - \chi\big(\iota^{D}_{\und{D}}\big) \,=\,\chi\big(\iota_D^{D^\prime}\big)\,\mu_{\und{D}}\quad,
\end{flalign}
for all trees of the form \eqref{eqn:forktrees}.
Using the identity \eqref{eqn:improvedgauge2} twice, with $D^\prime = \bbR^m$ and 
then also with $\und{D}$ replaced by $(D_1,\dots,\und{D}_i,\dots,D_n)$ one 
obtains that \eqref{eqn:improvedgauge2} implies \eqref{eqn:improvedgauge1}.
\end{rem}

\subsection{\label{subsec:2dinvariant}A simple binary invariant in $m=2$ dimensions}
In this subsection we consider a locally constant
prefactorization algebra $\FFF\in \Alg_{\PP_{\bbR^2}}$ on the $2$-dimensional 
Cartesian space $\bbR^2$. We construct a binary invariant 
$L \in \big[\H\FFF(\bbR^2)^{\otimes 2},\H\FFF(\bbR^2)\big]^{-1}$
of $\FFF$ which, when non-trivial $L\neq 0$, implies non-triviality of the 
universal first-order Massey product $\big[\widetilde{\beta}^{(2)}\big]\neq 0$.
We will show that the invariant $L$ has pleasant algebraic properties,
namely that it defines a degree $-1$ Poisson bracket (i.e.\ a $\mathbb{P}_2$-algebra structure) 
on the associative, unital and commutative algebra $\H\FFF(\bbR^2)$ from Proposition \ref{prop:strictstructuremaps}.
\begin{rem}\label{rem:PFAEmPm}
We believe, but do not know how to prove, that our invariant $L$ captures in dimension $m=2$ 
the higher algebraic structure that is expected by the following more abstract reasoning: By a result of Lurie 
\cite[Theorem 5.4.5.9]{HigherAlgebra}, see also \cite{AyalaFrancis} and \cite{CFM}, it is known that locally constant 
prefactorization algebras on $\bbR^m$ are equivalent to $\mathbb{E}_m$-algebras.
Furthermore, the $\mathbb{E}_m$-operads are formal in dimension $m\geq 2$ by \cite{Tamarkin, Fresse},
hence $\mathbb{E}_m$-algebras can be described equivalently in terms of $\H\mathbb{E}_m\simeq \mathbb{P}_m$-algebras
for $m\geq 2$. The degree $-1$ Poisson bracket that we find 
should be related to this $\mathbb{P}_2$-algebra structure.
We also expect that there exists a generalization of our invariant to $m\geq 3$ dimensions
which detects the degree $1-m$ Poisson bracket of the $\mathbb{P}_m$-operad.
However, as seen by a simple degree count, this requires working with
the more involved higher-order Massey products \cite{Dimitrova,UniversalMassey3}, which 
lies beyond the scope of our paper. We hope to come back to this issue in a future work.
\end{rem}

To define the invariant $L \in \big[\H\FFF(\bbR^2)^{\otimes 2},\H\FFF(\bbR^2)\big]^{-1}$,
let us consider any configuration of disks $D_1, D_2, D_u, D_d, \widetilde{D} \in \mathrm{Disk}_{\bbR^2}$
in some $D \in \mathrm{Disk}_{\bbR^2}$, with the property that 
$D_1\sqcup D_2 \subset D_u$, $D_1\sqcup D_2 \subset D_d$, 
$D_u \cap \widetilde{D} = D_d \cap \widetilde{D} = \emptyset$ 
and $D_u \cup D_d$ forming an annulus around $\widetilde{D}$, as depicted in the following picture:
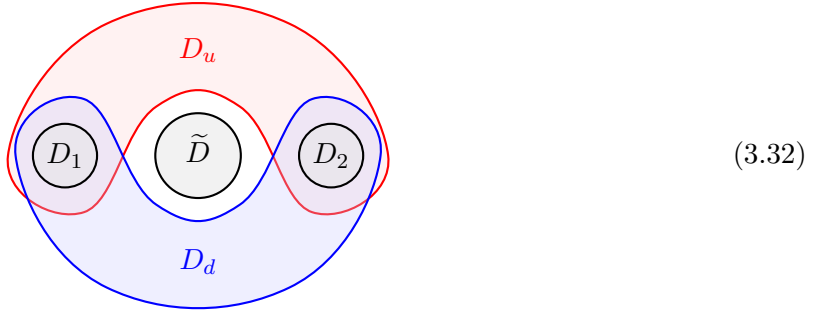
\begin{flalign}\label{eqn:Lpicture}
\parbox{5cm}{\begin{tikzpicture}
\draw [red, thick, fill=red!10, fill opacity=0.5] plot [smooth cycle, tension=.75] coordinates {(-2.5,0) (-1.5,-.75) (-.55,.65) (.55,.65) (1.5,-.75) (2.5,0) (1.2,1.75) (-1.2,1.75)};
\node[red] at (0,1.4) {$D_u$};
\draw [blue, thick, fill=blue!20, fill opacity=0.3] plot [smooth cycle, tension=.75] coordinates {(-2.4,0) (-1.5,.75) (-.55,-.65) (.55,-.65) (1.5,.75) (2.4,0) (1.2,-1.75) (-1.2,-1.75)};
\node[blue] at (0,-1.4) {$D_d$};
\draw[thick, fill=black!10, fill opacity=0.5] (0,0) circle (16pt);
\node at (0,0.1) {$\widetilde{D}$};
\draw[thick, fill=black!10, fill opacity=0.5] (-1.75,0) circle (12pt);
\node at (-1.75,0) {$D_1$};
\draw[thick, fill=black!10, fill opacity=0.5] (1.75,0) circle (12pt);
\node at (1.75,0) {$D_2$};
\end{tikzpicture}
}
\end{flalign}
To this choice we assign the linear combination
\begin{flalign}\label{eqn:Lcombination}
L \,:=\, \widetilde{\beta}^{(2)}\big(\iota_{(D_u, \widetilde{D})}^{D}, \iota_{D_1}^{D_u}\big) 
- \widetilde{\beta}^{(2)}\big(\iota_{(D_u,\widetilde{D})}^{D}, \iota_{D_2}^{D_u}\big) 
+ \widetilde{\beta}^{(2)}\big(\iota_{(D_d,\widetilde{D})}^{D}, \iota_{D_2}^{D_d}\big)
- \widetilde{\beta}^{(2)}\big(\iota_{(D_d,\widetilde{D})}^{D}, \iota_{D_1}^{D_d}\big) 
\end{flalign}
of degree $-1$ linear maps $\H\FFF(\bbR^2)^{\otimes 2}\to \H\FFF(\bbR^2)$. 
Informally, one should think of $L$ as being defined by moving the disk $D_1$
clockwise once around $\widetilde{D}$.
\begin{propo}\label{prop:Linvariant}
$L$ is invariant under the transformations $\widetilde{\beta}^{(2)}\to 
\widetilde{\beta}^{(2)} +\partial_\Gamma\chi$, for all 
$0$-cochains $\chi\in \Gamma^0(\H\FFF)$ that satisfy the conditions in Remark \ref{rem:computingMassey}.
Hence, $L$ depends only on the cohomology class 
$\big[\widetilde{\beta}^{(2)}\big]\in \H^1_\Gamma(\H\FFF)$ and it 
can be computed using any representative $1$-cocycle 
that satisfies the properties from Proposition \ref{prop: improved Masseys}.
Furthermore, if $L\neq 0$ is non-trivial, then the cohomology class $\big[\widetilde{\beta}^{(2)}\big]\neq 0$
is non-trivial too.
\end{propo}
\begin{proof}
We have to show that 
\begin{flalign}\label{eqn:Lpuregauge}
\partial_\Gamma\chi\big(\iota_{(D_u, \widetilde{D})}^{D}, \iota_{D_1}^{D_u}\big) 
- \partial_\Gamma\chi\big(\iota_{(D_u,\widetilde{D})}^{D}, \iota_{D_2}^{D_u}\big) 
+ \partial_\Gamma\chi\big(\iota_{(D_d,\widetilde{D})}^{D}, \iota_{D_2}^{D_d}\big)
- \partial_\Gamma\chi\big(\iota_{(D_d,\widetilde{D})}^{D}, \iota_{D_1}^{D_d}\big)
\,=\,0
\end{flalign}
for every $0$-cochain $\chi\in \Gamma^0(\H\FFF)$ that satisfies 
the conditions in Remark \ref{rem:computingMassey}.
Using the explicit formula for $\partial_\Gamma\chi$ from \eqref{eqn:gaugetransformationPFA},
this simplifies to 
\begin{flalign}
\mu_2\,\Big(\Big(
\chi(\iota_{D_1}^{D_u})
-\chi(\iota_{D_2}^{D_u})
+\chi(\iota_{D_2}^{D_d})
-\chi(\iota_{D_1}^{D_d})
\Big)\otimes\id
\Big) \,=\,0\quad,
\end{flalign}
where $\mu_2$ denotes the binary commutative multiplication from Proposition \ref{prop:strictstructuremaps}.
The latter holds true by using the conditions \eqref{eqn:improvedgauge2} on $\chi$ for the inclusions
$D_u\subset D$ and $D_d\subset D$.
\sk

To prove the last claim, note that $\big[\widetilde{\beta}^{(2)}\big]= 0$ implies 
$\widetilde{\beta}^{(2)} = \partial_\Gamma\chi$, hence $L=0$ by \eqref{eqn:Lpuregauge}.
\end{proof}

\begin{lem}\label{lem:Lindependentdisks}
$L$ is independent of the choice of disk configuration as in \eqref{eqn:Lpicture}.
\end{lem}
\begin{proof}
It suffices to show that $L$ does not change when replacing any of 
the individual disks $D_1,D_2,D_u,D_d,\widetilde{D}, D$ by a bigger one such that the resulting disk configuration
is still of the form \eqref{eqn:Lpicture}.
Independence of the choice of big disk $D\subseteq \bbR^2$ is a direct consequence of 
Proposition \ref{prop: improved Masseys} (1). Next, let us consider $D_1\subset D_1^\prime\subset D_u\cap D_d$.
The difference $L^\prime - L$ of \eqref{eqn:Lcombination} evaluated for the two different disk configurations then reads as
\begin{flalign}
\nn L^\prime - L \,&=\,
\widetilde{\beta}^{(2)}\big(\iota_{(D_u, \widetilde{D})}^{D}, \iota_{D^\prime_1}^{D_u}\big) 
- \widetilde{\beta}^{(2)}\big(\iota_{(D_d,\widetilde{D})}^{D}, \iota_{D_1^\prime}^{D_d}\big) 
-\widetilde{\beta}^{(2)}\big(\iota_{(D_u, \widetilde{D})}^{D}, \iota_{D_1}^{D_u}\big) 
+ \widetilde{\beta}^{(2)}\big(\iota_{(D_d,\widetilde{D})}^{D}, \iota_{D_1}^{D_d}\big) \\
\,&=\,- \widetilde{\beta}^{(2)}\big(\iota_{(D_1^\prime,\widetilde{D})}^{D}, \iota_{D_1}^{D_1^\prime}\big)
+ \widetilde{\beta}^{(2)}\big(\iota_{(D_1^\prime,\widetilde{D})}^{D}, \iota_{D_1}^{D_1^\prime}\big)\,=\,0\quad,
\end{flalign}
where in the second step we used the cocycle conditions \eqref{eqn:genMasseyrelations1} for the trees
$t\big(\iota^{D}_{(D_u,\widetilde{D})},\iota^{D_u}_{D_1^\prime},\iota^{D_1^\prime}_{D_1}\big)$ and
$t\big(\iota^{D}_{(D_d,\widetilde{D})},\iota^{D_d}_{D_1^\prime},\iota^{D_1^\prime}_{D_1}\big)$.
Independence under the replacements 
$D_2\subset D_2^\prime\subset D_u\cap D_d$, $D_u\subset D_u^\prime\subset D\setminus \widetilde{D}$ 
and $D_d\subset D_d^\prime\subset D\setminus \widetilde{D}$ is shown similarly. For 
$\widetilde{D}\subset \widetilde{D}^\prime\subset D\setminus (D_u\cup D_d)$,
instead of \eqref{eqn:genMasseyrelations1} one uses the cocycle conditions \eqref{eqn:genMasseyrelations2}.
\end{proof}

\begin{theo}\label{theo:LPoisson}
The map $L: \H\FFF(\bbR^2)\otimes \H\FFF(\bbR^2)\to \H\FFF(\bbR^2)$ defines a degree $-1$ Poisson bracket 
(i.e.\ a $\mathbb{P}_2$-algebra structure)
on the associative, unital and commutative algebra $\H\FFF(\bbR^2)$ from Proposition \ref{prop:strictstructuremaps}.
Explicitly, it satisfies:
\begin{itemize}
\item[(1)] \textit{Symmetry:}\footnote{Recall that oddly shifted Poisson brackets 
are symmetric and not antisymmetric, see e.g.\ \cite{Fresse}.} 
\begin{flalign}
L\,\tau \,=\, L\quad,
\end{flalign}
where $\tau$ denotes the symmetric braiding \eqref{eqn:braiding}.

\item[(2)] \textit{Derivation property:}
\begin{flalign} \label{eqn:derivation}
L\,(\id\otimes \mu_2)\,=\, \mu_2\,(L\otimes \id)+ \mu_2\,(\id\otimes L)\,(\tau\otimes\id) \quad.
\end{flalign}

\item[(3)] \textit{Jacobi identity:}
\begin{flalign} \label{eqn:Jacobi}
L\,(\id\otimes L) + L\,(\id\otimes L) \,\tau_{(123)} + L\,(\id\otimes L) \,\tau_{(132)} \,=\, 0 \quad,
\end{flalign}
where $\tau_{(123)}$ and $\tau_{(132)}$ denote the actions of the cyclic permutations $(123)$ and $(132)$ in $\Sigma_3$, respectively, on $\H\FFF(\bbR^2)^{\otimes 3}$ via the symmetric braiding \eqref{eqn:braiding}.
\end{itemize}
\end{theo}
\begin{proof}
Item (1) can be proven by considering the following configuration of disks
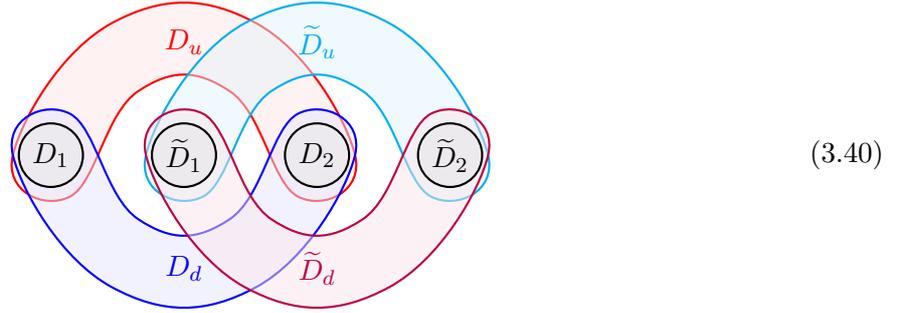
\begin{flalign}\label{eqn:Lsymmetrypicture}
\parbox{7cm}{\begin{tikzpicture}
\draw [red, thick, fill=red!10, fill opacity=0.5] plot [smooth cycle, tension=.75] coordinates {(-2.25,0) (-1.5,-.55) (-.6,.85) (.6,.85) (1.5,-.55) (2.25,0) (0.95,1.75) (-0.95,1.75)};
\node[red] at (0,1.5) {$D_u$};
\draw [blue, thick, fill=blue!10, fill opacity=0.5] plot [smooth cycle, tension=.75] coordinates {(-2.25,0) (-1.5,.55) (-.6,-.85) (.6,-.85) (1.5,.55) (2.25,0) (0.95,-1.75) (-0.95,-1.75)};
\node[blue] at (0,-1.5) {$D_d$};
\draw [cyan, thick, fill=cyan!10, fill opacity=0.5] plot [smooth cycle, tension=.75] coordinates {(-0.5,0) (0.25,-.55) (1.15,.85) (2.35,.85) (3.25,-.55) (4,0) (2.7,1.75) (0.8,1.75)};
\node[cyan] at (1.75,1.5) {$\widetilde{D}_u$};
\draw [purple, thick, fill=purple!10, fill opacity=0.5] plot [smooth cycle, tension=.75] coordinates {(-0.5,0) (0.25,.55) (1.15,-.85) (2.35,-.85) (3.25,.55) (4,0) (2.7,-1.75) (0.8,-1.75)};
\node[purple] at (1.75,-1.5) {$\widetilde{D}_d$};
\draw[thick, fill=black!10, fill opacity=0.5] (-1.75,0) circle (12pt);
\node at (-1.75,0) {$D_1$};
\draw[thick, fill=black!10, fill opacity=0.5] (1.75,0) circle (12pt);
\node at (1.75,0) {$D_2$};
\draw[thick, fill=black!10, fill opacity=0.5] (0,0) circle (12pt);
\node at (0,0) {$\widetilde{D}_1$};
\draw[thick, fill=black!10, fill opacity=0.5] (3.5,0) circle (12pt);
\node at (3.5,0) {$\widetilde{D}_2$};
\end{tikzpicture}
}
\end{flalign}
and suitably combining the cocycle conditions \eqref{eqn:genMasseyrelations2}
for the trees 
\begin{flalign}
t\big(\iota^D_{(D_u,\widetilde{D}_d)},\iota^{D_u}_{D_i},\iota^{\widetilde{D}_d}_{\widetilde{D}_j}\big)
\qquad\text{and}\qquad
t\big(\iota^D_{(D_d,\widetilde{D}_u)},\iota^{D_d}_{D_i},\iota^{\widetilde{D}_u}_{\widetilde{D}_j}\big)\quad,
\end{flalign}
for all combinations of $i,j\in\{1,2\}$.
To arrive at the result $L\,\tau = L$, one also has to use that $L$ is insensitive
to the choice of disk configuration and that the invariant \eqref{eqn:Lcombination} is trivial
whenever $\widetilde{D}$ lies \textit{outside} of the annulus formed by $D_u\cup D_d$.
Both of these claims are a consequence of Lemma \ref{lem:Lindependentdisks} and its proof.
\sk

To prove item (2), we use the following configuration of disks
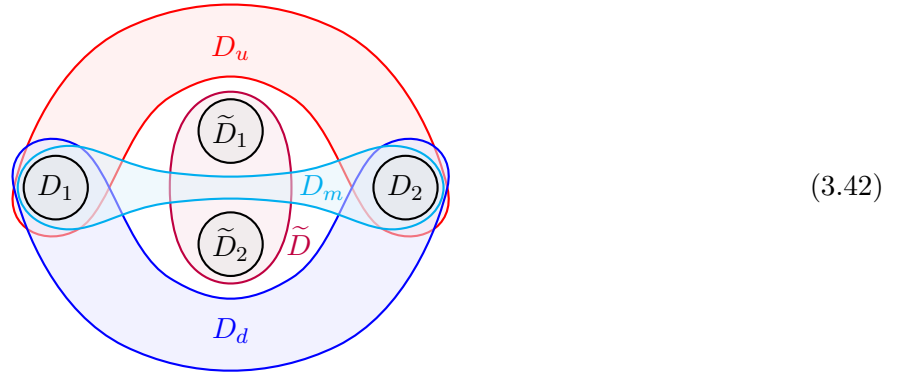
\begin{flalign}
\parbox{7cm}{\begin{tikzpicture}
\draw [purple, thick, fill=purple!10, fill opacity=0.5] plot [smooth cycle, tension=.75] coordinates {(-.45,1.1) (.45,1.1) (.8,0) (.45,-1.1) (-.45,-1.1) (-.8,0)};
\node[purple] at (.9,-.7) {$\widetilde{D}$};
\draw [red, thick, fill=red!10, fill opacity=0.5] plot [smooth cycle, tension=.75] coordinates {(-2.85,0) (-2,-.55) (-.8,1.2) (.8,1.2) (2,-.55) (2.85,0) (1.4,2.1) (-1.4,2.1)};
\node[red] at (0,1.85) {$D_u$};
\draw [blue, thick, fill=blue!10, fill opacity=0.5] plot [smooth cycle, tension=.75] coordinates {(-2.85,0) (-2,.55) (-.8,-1.2) (.8,-1.2) (2,.55) (2.85,0) (1.4,-2.1) (-1.4,-2.1)};
\node[blue] at (0,-1.9) {$D_d$};
\draw [cyan, thick, fill=cyan!10, fill opacity=0.5] plot [smooth cycle, tension=.75] coordinates {(-2.8,0) (-2.2,.55) (-.8,.2) (.8,.2) (2.2,.55) (2.8,0) (2.2,-.55) (.8,-.2) (-.8,-.2) (-2.2,-.55)};
\node[cyan] at (1.2,0) {$D_m$};
\draw[thick, fill=black!10, fill opacity=0.5] (-2.3,0) circle (12pt);
\node at (-2.3,0) {$D_1$};
\draw[thick, fill=black!10, fill opacity=0.5] (2.3,0) circle (12pt);
\node at (2.3,0) {$D_2$};
\draw[thick, fill=black!10, fill opacity=0.5] (0,.75) circle (12pt);
\node at (0,.75) {$\widetilde{D}_1$};
\draw[thick, fill=black!10, fill opacity=0.5] (0,-.75) circle (12pt);
\node at (0,-.75) {$\widetilde{D}_2$};
\end{tikzpicture}
}
\end{flalign}
and the fact that $L$ is insensitive to the choice of disk configuration by Lemma \ref{lem:Lindependentdisks}. 
To evaluate the term $L\,(\id\otimes \mu_2)$ on the left-hand side of \eqref{eqn:derivation},
we use a suitable combination of the cocycle conditions \eqref{eqn:genMasseyrelations2}
for the trees
\begin{flalign}
t\big(\iota^D_{(D_u,\widetilde{D})},\iota^{D_u}_{D_i},\iota^{\widetilde{D}}_{(\widetilde{D}_1,\widetilde{D}_2)}\big)
\qquad\text{and}\qquad
t\big(\iota^D_{(D_d,\widetilde{D})},\iota^{D_d}_{D_i},\iota^{\widetilde{D}}_{(\widetilde{D}_1,\widetilde{D}_2)}\big)\quad,
\end{flalign}
with $i \in \{1, 2\}$.
For the first term $\mu_2\,(L\otimes \id)$ on the right-hand side of \eqref{eqn:derivation}, 
we use instead a suitable combination of the cocycle conditions \eqref{eqn:genMasseyrelations1} for the trees
\begin{flalign}
t\big(\iota^D_{(\widehat{D}_u,\widetilde{D}_2)},\iota^{\widehat{D}_u}_{(D_u, \widetilde{D}_1)}, \iota^{D_u}_{D_i} \big)
\qquad\text{and}\qquad
t\big(\iota^D_{(\widehat{D}_u,\widetilde{D}_2)},\iota^{\widehat{D}_u}_{(D_m, \widetilde{D}_1)}, \iota^{D_m}_{D_i} \big)\quad,
\end{flalign}
where $\widehat{D}_u\subset \bbR^2$ is a disk such that $D_u, D_m \subset \widehat{D}_u$ and 
$\widehat{D}_u \cap \widetilde{D}_2 = \emptyset$, as in the following picture:
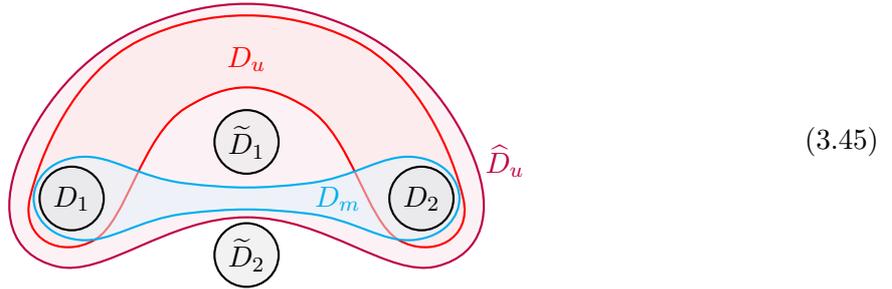
\begin{flalign}
\parbox{7cm}{\begin{tikzpicture}
\draw [purple, thick, fill=purple!10, fill opacity=0.5] plot [smooth cycle, tension=.75] coordinates {(-3,.5) (-1.3,2.3) (1.3,2.3) (3,.5) (2.5,-.9) (0,-.25) (-2.5,-.9)};
\node[purple] at (3.4,0.5) {$\widehat{D}_u$};
\draw [red, thick, fill=red!10, fill opacity=0.5] plot [smooth cycle, tension=.75] coordinates {(-2.85,0) (-2,-.55) (-.8,1.2) (.8,1.2) (2,-.55) (2.85,0) (1.4,2.1) (-1.4,2.1)};
\node[red] at (0,1.85) {$D_u$};
\draw [cyan, thick, fill=cyan!10, fill opacity=0.5] plot [smooth cycle, tension=.75] coordinates {(-2.8,0) (-2.2,.55) (-.8,.2) (.8,.2) (2.2,.55) (2.8,0) (2.2,-.55) (.8,-.2) (-.8,-.2) (-2.2,-.55)};
\node[cyan] at (1.2,0) {$D_m$};
\draw[thick, fill=black!10, fill opacity=0.5] (-2.3,0) circle (12pt);
\node at (-2.3,0) {$D_1$};
\draw[thick, fill=black!10, fill opacity=0.5] (2.3,0) circle (12pt);
\node at (2.3,0) {$D_2$};
\draw[thick, fill=black!10, fill opacity=0.5] (0,.75) circle (12pt);
\node at (0,.75) {$\widetilde{D}_1$};
\draw[thick, fill=black!10, fill opacity=0.5] (0,-.75) circle (12pt);
\node at (0,-.75) {$\widetilde{D}_2$};
\end{tikzpicture}
}
\end{flalign}
Similarly, to compute the second term $\mu_2\,(\id\otimes L) (\tau \otimes \id)$ 
on the right-hand side of \eqref{eqn:derivation}, we use a suitable combination of 
the cocycle conditions \eqref{eqn:genMasseyrelations1} for the trees
\begin{flalign}
t\big(\iota^D_{(\widetilde{D}_1, \widehat{D}_d)},\iota^{\widehat{D}_d}_{(D_m, \widetilde{D}_2)}, \iota^{D_m}_{D_i} \big)
\qquad\text{and}\qquad
t\big(\iota^D_{(\widetilde{D}_1, \widehat{D}_d)},\iota^{\widehat{D}_d}_{(D_d, \widetilde{D}_2)}, \iota^{D_d}_{D_i} \big)\quad,
\end{flalign}
where $\widehat{D}_d \subset \bbR^2$ is a disk such that 
$D_m, D_d \subset \widehat{D}_d$ and $\widehat{D}_d \cap \widetilde{D}_1 = \emptyset$, 
as in the following picture:
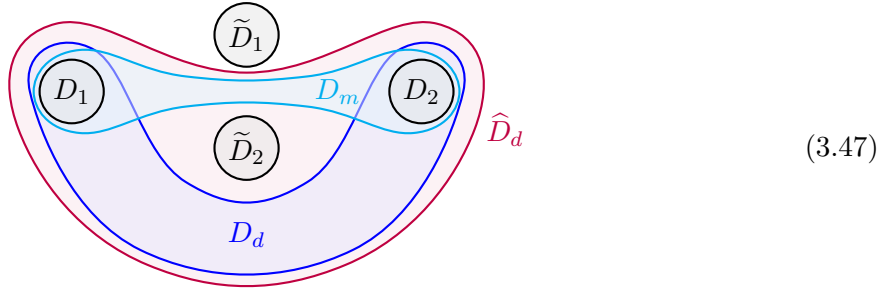
\begin{flalign}
\parbox{7cm}{\begin{tikzpicture}
\draw [purple, thick, fill=purple!10, fill opacity=0.5] plot [smooth cycle, tension=.75] coordinates {(-3,-.5) (-1.3,-2.3) (1.3,-2.3) (3,-.5) (2.5,.9) (0,.25) (-2.5,.9)};
\node[purple] at (3.4,-0.5) {$\widehat{D}_d$};
\draw [blue, thick, fill=blue!10, fill opacity=0.5] plot [smooth cycle, tension=.75] coordinates {(-2.85,0) (-2,.55) (-.8,-1.2) (.8,-1.2) (2,.55) (2.85,0) (1.4,-2.1) (-1.4,-2.1)};
\node[blue] at (0,-1.9) {$D_d$};
\draw [cyan, thick, fill=cyan!10, fill opacity=0.5] plot [smooth cycle, tension=.75] coordinates {(-2.8,0) (-2.2,.55) (-.8,.2) (.8,.2) (2.2,.55) (2.8,0) (2.2,-.55) (.8,-.2) (-.8,-.2) (-2.2,-.55)};
\node[cyan] at (1.2,0) {$D_m$};
\draw[thick, fill=black!10, fill opacity=0.5] (-2.3,0) circle (12pt);
\node at (-2.3,0) {$D_1$};
\draw[thick, fill=black!10, fill opacity=0.5] (2.3,0) circle (12pt);
\node at (2.3,0) {$D_2$};
\draw[thick, fill=black!10, fill opacity=0.5] (0,.75) circle (12pt);
\node at (0,.75) {$\widetilde{D}_1$};
\draw[thick, fill=black!10, fill opacity=0.5] (0,-.75) circle (12pt);
\node at (0,-.75) {$\widetilde{D}_2$};
\end{tikzpicture}
}
\end{flalign}
Combining all the above, the identity \eqref{eqn:derivation} follows immediately.
\sk

Finally, in order to prove item (3), 
it will be convenient to consider a configuration 
of input disks $D_1, D_2, D_1^\prime, D_2^\prime, \widetilde{D}_1, \widetilde{D}_2 \subset \bbR^2$ 
which is symmetric under rotation by $\pi/3$. We will also need four types of disks $\widetilde{D}_u,
D_m^\prime, D_d, D_p \subset \bbR^2$, as depicted in the following picture:
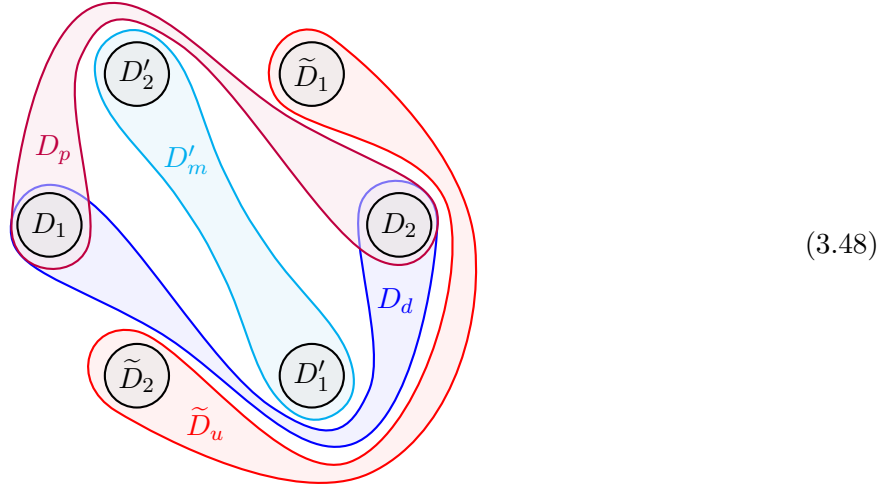
\begin{flalign}\label{eqn:Ljacobipicture1}
\parbox{7cm}{\begin{tikzpicture}
\draw [red, thick, fill=red!10, rotate=-120, fill opacity=0.5] plot [smooth cycle, tension=.75] coordinates {(-2.85,0) (-1.8,-.25) (-1.6,2.5)(1.9,3) (1.8,-.3) (2.85,0) (1.9,3.3) (-1.4,3)};
\node[red] at ([rotate=-120] 2.4,1.1) {$\widetilde{D}_u$};
\draw [cyan, thick, fill=cyan!10, rotate=120, fill opacity=0.5] plot [smooth cycle, tension=.75] coordinates {(-2.9,0) (-2.3,.55) (-.6,.35) (.6,.35) (2.3,.55) (2.9,0) (2.3,-.55) (.6,-.35) (-.6,-.35) (-2.3,-.55)};
\node[cyan] at ([rotate=120] 1,0) {$D_m^\prime$};
\draw [blue, thick, fill=blue!10, fill opacity=0.5] plot [smooth cycle, tension=.75] coordinates {(-1.9,.4) (.6,-2.4) (1.9,-2.2) (1.8,.3) (2.8,0) (1.7,-2.9) (-.9,-1.4) (-2.75,-.3)};
\node[blue] at (2.25,-1) {$D_d$};
\draw [purple, thick, fill=purple!10, fill opacity=0.5] plot [smooth cycle, tension=.75] coordinates {(1.9,-.4) (-.6,2.4) (-1.9,2.2) (-1.8,-.3) (-2.8,0) (-1.7,2.9) (.9,1.4) (2.75,.3)};
\node[purple] at (-2.25,1) {$D_p$};
\draw[thick, fill=black!10, fill opacity=0.5] (2.3, 0) circle (12pt);
\node at (2.3, 0) {$D_2$};
\draw[thick, fill=black!10, fill opacity=0.5] (1.15, 2) circle (12pt);
\node at (1.15, 2) {$\widetilde{D}_1$};
\draw[thick, fill=black!10, fill opacity=0.5] (-1.15, 2) circle (12pt);
\node at (-1.15, 2) {$D_2^\prime$};
\draw[thick, fill=black!10, fill opacity=0.5] (-2.3, 0) circle (12pt);
\node at (-2.3, 0) {$D_1$};
\draw[thick, fill=black!10, fill opacity=0.5] (-1.15, -2) circle (12pt);
\node at (-1.15, -2) {$\widetilde{D}_2$};
\draw[thick, fill=black!10, fill opacity=0.5] (1.15, -2) circle (12pt);
\node at (1.15, -2) {$D_1^\prime$};
\end{tikzpicture}
}
\end{flalign}
together with their rotations by $2\pi/3$, which we call $D_u, \widetilde{D}_m, D_d^\prime,
D_p^\prime \subset \bbR^2$, and their rotations by $- 2\pi/3$, which we call $D_u^\prime, D_m,
\widetilde{D}_d, \widetilde{D}_p \subset \bbR^2$. 
We require the pairwise disjointness conditions 
\begin{subequations}
\begin{alignat}{3}
D_m^\prime \cap \widetilde{D}_u = \emptyset~~, \quad
& \quad D_m^\prime \cap D_d = \emptyset~~, \quad
& \quad\widetilde{D}_u \cap D_d = \emptyset~~, 
\\
D_m^\prime \cap D_u = \emptyset~~, \quad
& \quad D_m^\prime \cap D_p = \emptyset~~, \quad
& \quad \widetilde{D}_u \cap D_p = \emptyset~~,
\end{alignat}
\end{subequations}
suggested by the picture \eqref{eqn:Ljacobipicture1} 
(note that $D_u$ is the rotation of $\widetilde{D}_u$ by $2 \pi /3$),
and similarly for the rotations of these disks by $\pm 2 \pi/3$.
\sk

Since $L$ is insensitive to the choice of disk configuration by Lemma \ref{lem:Lindependentdisks}, 
the first term on the left-hand side of \eqref{eqn:Jacobi} can be expressed 
using the following configuration of disks:
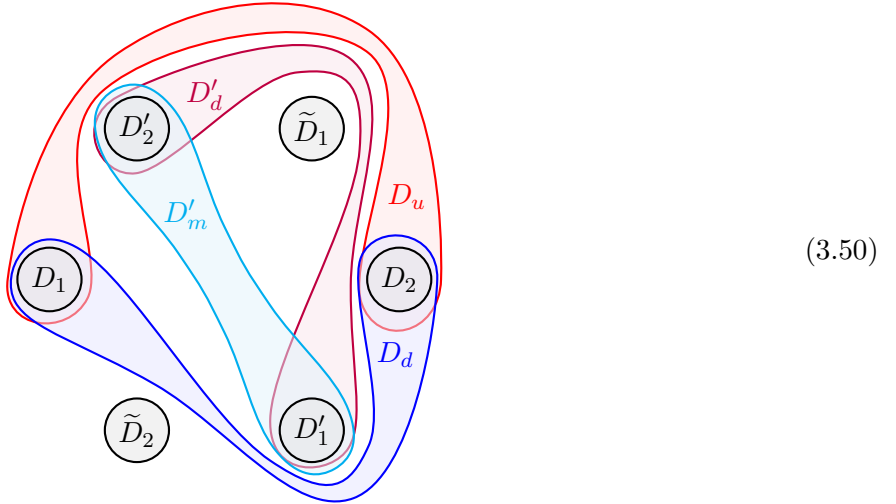
\begin{flalign}\label{eqn:Ljacobipicture2}
\parbox{7cm}{\begin{tikzpicture}
\draw [purple, thick, fill=purple!10, rotate=120, fill opacity=0.5] plot [smooth cycle, tension=.75] coordinates {(-1.9,.4) (.6,-2.4) (1.9,-2.2) (1.8,.3) (2.8,0) (1.7,-2.9) (-.9,-1.4) (-2.75,-.3)};
\node[purple] at ([rotate=120] 2.25,-1) {$D_d^\prime$};
\draw [cyan, thick, fill=cyan!10, rotate=120, fill opacity=0.5] plot [smooth cycle, tension=.75] coordinates {(-2.9,0) (-2.3,.55) (-.6,.35) (.6,.35) (2.3,.55) (2.9,0) (2.3,-.55) (.6,-.35) (-.6,-.35) (-2.3,-.55)};
\node[cyan] at ([rotate=120] 1,0) {$D_m^\prime$};
\draw [red, thick, fill=red!10, fill opacity=0.5] plot [smooth cycle, tension=.75] coordinates {(-2.85,0) (-1.8,-.25) (-1.6,2.5)(1.9,3) (1.8,-.3) (2.85,0) (1.9,3.3) (-1.4,3)};
\node[red] at (2.4,1.1) {$D_u$};
\draw [blue, thick, fill=blue!10, fill opacity=0.5] plot [smooth cycle, tension=.75] coordinates {(-1.9,.4) (.6,-2.4) (1.9,-2.2) (1.8,.3) (2.8,0) (1.7,-2.9) (-.9,-1.4) (-2.75,-.3)};
\node[blue] at (2.25,-1) {$D_d$};
\draw[thick, fill=black!10, fill opacity=0.5] (2.3, 0) circle (12pt);
\node at (2.3, 0) {$D_2$};
\draw[thick, fill=black!10, fill opacity=0.5] (1.15, 2) circle (12pt);
\node at (1.15, 2) {$\widetilde{D}_1$};
\draw[thick, fill=black!10, fill opacity=0.5] (-1.15, 2) circle (12pt);
\node at (-1.15, 2) {$D_2^\prime$};
\draw[thick, fill=black!10, fill opacity=0.5] (-2.3, 0) circle (12pt);
\node at (-2.3, 0) {$D_1$};
\draw[thick, fill=black!10, fill opacity=0.5] (-1.15, -2) circle (12pt);
\node at (-1.15, -2) {$\widetilde{D}_2$};
\draw[thick, fill=black!10, fill opacity=0.5] (1.15, -2) circle (12pt);
\node at (1.15, -2) {$D_1^\prime$};
\end{tikzpicture}
}
\end{flalign}
Explicitly, 
\begin{flalign}\label{L(idL0)}
L \, (\id \otimes L) \,=\, 
\sum_{i,j,k=1}^2 \sum_{\substack{\alpha \in \{ u,d\}\\ \beta \in \{m,d\}}} 
- (-1)^{i+j+k+\vert \alpha \vert+\vert \beta \vert} \,
\widetilde{\beta}^{(2)} \big( \iota_{D_\alpha , \widehat{D}}^{D} , \iota_{D_i}^{D_\alpha} \big) \,
 \widetilde{\beta}^{(2)}  \big( \iota_{D_\beta^\prime , \widetilde{D}_k}^{\widehat{D}} , \iota_{D_j^\prime}^{D_\beta^\prime }  \big)\quad,
\end{flalign}
where $\widehat{D} \subset \bbR^2$ is a disk such that 
$D_d^\prime \cup D_m^\prime \subset \widehat{D}$, 
$\widehat{D} \cap D_u = \emptyset$ and 
$\widehat{D} \cap D_d = \emptyset$, 
and we define $|u| := 0$, $|m| := 0$ and $|d| := 1$. 
Note that the sum over $k = 2$ is not a part of the usual expression of $L \, (\id \otimes L)$, 
but it is zero by Lemma \ref{lem:Lindependentdisks} 
with an argument similar to the one used in the proof of item (1). 
\sk

We now use a suitable combination of the Maurer-Cartan equation \eqref{eqn:MCtwisting} 
evaluated on the $4$-vertex trees 
\begin{flalign}
t \big(\iota^D_{(D_\alpha,\widehat{D})}, 
\iota^{D_\alpha}_{D_i}, 
\iota^{\widehat{D}}_{(D_\beta^\prime, \widetilde{D}_k)}, 
\iota_{D_j^\prime}^{D_\beta^\prime} \big) \quad,
\end{flalign}
with $\alpha \in \{ u, d\}$, $\beta \in \{m, d \}$ and $i, j, k \in \{1, 2\}$. 
We then find that
\begin{flalign}\label{L(idL)}
L \, (\id \otimes L) 
\,=\, 
\sum_{i,j,k=1}^2 
\sum_{\substack{\alpha \in \{ u,d\}\\ \beta \in \{m,d\}}} 
- (-1)^{i+j+k+\vert \alpha \vert+\vert \beta \vert} \,
\widetilde{\beta}\big(\iota_{(D_\alpha, D_\beta^\prime, \widetilde{D}_k)}^{D}, 
\iota_{D_i}^{D_\alpha}, \iota_{D_j^\prime}^{D_\beta^\prime}\big) \quad .
\end{flalign}
The other two terms $L \, (\id \otimes L) \tau_{(123)}$ and $L \, (\id \otimes L) \tau_{(132)}$
on the left-hand side of \eqref{eqn:Jacobi} are given by similar expressions involving
the disks $D_u^\prime, D_d^\prime, \widetilde{D}_m, \widetilde{D}_d$ and
$\widetilde{D}_u, \widetilde{D}_d, D_m, D_d$, respectively.
To this expression of the left-hand side of \eqref{eqn:Jacobi}
we now apply the Maurer-Cartan equation \eqref{eqn:MCtwisting} evaluated 
on the $4$-vertex trees 
\begin{alignat}{2}
\nn &t\big(\iota^D_{(D_u,D_\alpha^\prime, \widetilde{D}_m)}, \iota^{D_u}_{D_i}, \iota^{D_\alpha^\prime}_{D_j^\prime}, \iota^{\widetilde{D}_m}_{\widetilde{D}_k} \big)
\qquad, &\qquad
&t\big(\iota^D_{(D_\alpha,D_m^\prime, \widetilde{D}_u)}, \iota^{D_\alpha}_{D_i}, \iota^{D_m^\prime}_{D_j^\prime}, \iota^{\widetilde{D}_u}_{\widetilde{D}_k} \big) \quad,\\
&t\big(\iota^D_{(D_m,D_u^\prime, \widetilde{D}_\alpha)}, \iota^{D_m}_{D_i}, \iota^{D_u^\prime}_{D_j^\prime}, \iota^{\widetilde{D}_\alpha}_{\widetilde{D}_k} \big)
\qquad\text{and}&\qquad
&t\big(\iota^D_{(D_\alpha,D_\alpha^\prime, \widetilde{D}_\alpha)}, \iota^{D_\alpha}_{D_i}, \iota^{D_\alpha^\prime}_{D_j^\prime}, \iota^{\widetilde{D}_\alpha}_{\widetilde{D}_k} \big) \quad,
\end{alignat}
with $i, j, k \in \{1, 2\}$, 
first with $\alpha = d$ and then with $\alpha = p$. 
One then finds that
\begin{flalign}\label{eqn:Jacobi L to K}
L (\id \otimes L) + L (\id \otimes L) \tau_{(123)} + L (\id \otimes L) \tau_{(132)}
\,=\, K + K \tau_{(123)} + K \tau_{(132)} \quad .
\end{flalign}
Here, $K$ is the same expression as on the right-hand side of \eqref{L(idL)}, but with the disks $D_d$
and $D^\prime_d$ replaced by $D_p$ and $D_p^\prime$, respectively, i.e.\ with the disks
in picture \eqref{eqn:Ljacobipicture2} used in the expression \eqref{L(idL)} for $L (\id \otimes L)$
replaced by the disks in the following picture:
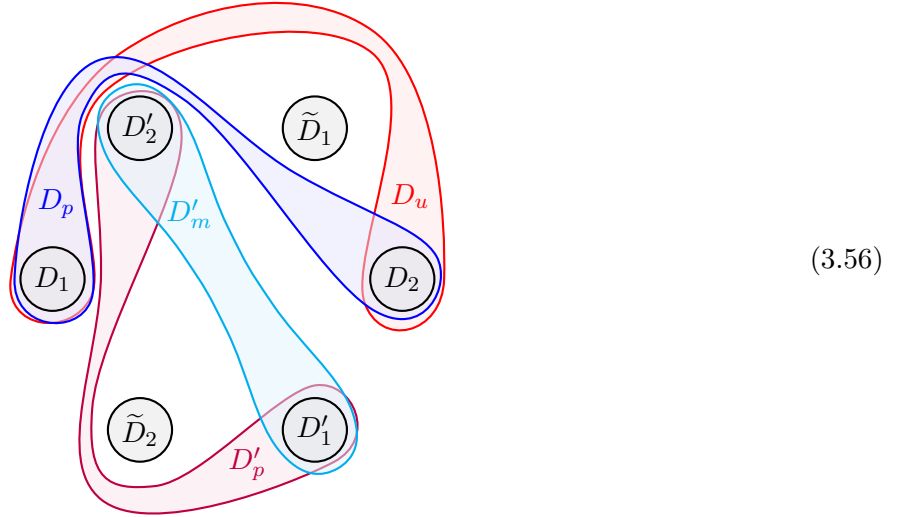
\begin{flalign} \label{alternative disks}
\parbox{7cm}{\begin{tikzpicture}
\draw [purple, thick, fill=purple!10, rotate=120, fill opacity=0.5] plot [smooth cycle, tension=.75] coordinates {(1.9,-.4) (-.6,2.4) (-1.9,2.2) (-1.8,-.3) (-2.8,0) (-1.7,2.9) (.9,1.4) (2.75,.3)};
\node[purple] at ([rotate=120] -2.25,1) {$D_p^\prime$};
\draw [cyan, thick, fill=cyan!10, rotate=120, fill opacity=0.5] plot [smooth cycle, tension=.75] coordinates {(-2.9,0) (-2.3,.55) (-.6,.35) (.6,.35) (2.3,.55) (2.9,0) (2.3,-.55) (.6,-.35) (-.6,-.35) (-2.3,-.55)};
\node[cyan] at ([rotate=120] 1,0) {$D_m^\prime$};
\draw [red, thick, fill=red!10, fill opacity=0.5] plot [smooth cycle, tension=.75] coordinates {(-2.85,0) (-1.8,-.25) (-1.6,2.5)(1.9,3) (1.8,-.3) (2.85,0) (1.9,3.3) (-1.4,3)};
\node[red] at (2.4,1.1) {$D_u$};
\draw [blue, thick, fill=blue!10, fill opacity=0.5] plot [smooth cycle, tension=.75] coordinates {(1.9,-.4) (-.6,2.4) (-1.9,2.2) (-1.8,-.3) (-2.8,0) (-1.7,2.9) (.9,1.4) (2.75,.3)};
\node[blue] at (-2.25,1) {$D_p$};
\draw[thick, fill=black!10, fill opacity=0.5] (2.3, 0) circle (12pt);
\node at (2.3, 0) {$D_2$};
\draw[thick, fill=black!10, fill opacity=0.5] (1.15, 2) circle (12pt);
\node at (1.15, 2) {$\widetilde{D}_1$};
\draw[thick, fill=black!10, fill opacity=0.5] (-1.15, 2) circle (12pt);
\node at (-1.15, 2) {$D_2^\prime$};
\draw[thick, fill=black!10, fill opacity=0.5] (-2.3, 0) circle (12pt);
\node at (-2.3, 0) {$D_1$};
\draw[thick, fill=black!10, fill opacity=0.5] (-1.15, -2) circle (12pt);
\node at (-1.15, -2) {$\widetilde{D}_2$};
\draw[thick, fill=black!10, fill opacity=0.5] (1.15, -2) circle (12pt);
\node at (1.15, -2) {$D_1^\prime$};
\end{tikzpicture}
}
\end{flalign}
Explicitly, we set
\begin{flalign}\label{K def}
K \, := \, 
\sum_{i,j,k=1}^2
\sum_{\substack{\alpha \in \{ u,p\}\\ \beta \in \{m,p\}}} 
- (-1)^{i+j+k+\vert \alpha \vert+\vert \beta \vert} \,
\widetilde{\beta} \big(\iota_{(D_\alpha, D_\beta^\prime, \widetilde{D}_k)}^{D}, \iota_{D_i}^{D_\alpha}, \iota_{D_j^\prime}^{D_\beta^\prime}\big)\quad,
\end{flalign}
where as before we have $|u| =0$, $|m| = 0$ and we now also set $|p| := 1$. 
The terms $K \tau_{(123)}$ and $K \tau_{(132)}$
on the right-hand side of \eqref{eqn:Jacobi L to K}
are given by similar expressions to \eqref{K def}
involving the disks $D_u^\prime, D_p^\prime, \widetilde{D}_m, \widetilde{D}_p$ and
$\widetilde{D}_u, \widetilde{D}_p, D_m, D_p$, respectively.
\sk

The result now follows by noting that $K=0$. 
Indeed, by the same argument as in the proof of Lemma \ref{lem:Lindependentdisks}, 
one shows that $K$ is independent of the choice of disk configurations as in \eqref{alternative disks}. 
Hence, the terms in the sum in \eqref{K def} with $k=1$ vanish by deforming $D_p^\prime$ into 
$D_m^\prime$ and the terms in the sum in \eqref{K def} with $k=2$ vanish by deforming $D_p$ into $D_u$.
\end{proof}


\section{\label{sec:examples}Factorization envelopes}
Throughout this section, we fix the field $\bbK$ to be either the
real numbers $\bbR$ or the complex numbers $\bbC$.
Let us recall from \cite[Chapter 3.6]{CostelloGwilliam} 
that associated with any local dg-Lie algebra
is a (pre)factorization algebra, called its factorization envelope,
that is formed by taking Chevalley-Eilenberg chains.
In this section we study the case of $\g^{\bbR^m} := \g \otimes \Omega_{\bbR^m}^\bullet$,
where $\g$ is any finite-dimensional Lie algebra over $\bbK$
and $\Omega_{\bbR^m}^\bullet$ is the sheaf of de Rham complexes on $\bbR^m$.
The associated factorization envelope is then given by
\begin{flalign}
\mathfrak{U} \g^{\bbR^m} \,:=\,  \mathrm{CE}_\bullet\big(\g^{\bbR^m}\big)\,:=\,
\Big(\Sym\big( \g_{\cc}^{\bbR^m}[1] \big) , \dd_{\dR[1]} + \dd_{\CE}\Big)\,\in \Alg_{\PP_{\bbR^m}}\quad,
\end{flalign}
where $\g_{\cc}^{\bbR^m}[1] := \g\otimes\Omega^\bullet_{\bbR^m,\cc}[1]$ denotes the cosheaf
of ($1$-shifted) compactly supported sections and the Chevalley-Eilenberg differential $\dd_{\CE}$ 
will be described explicitly in \eqref{dCE def} below. This prefactorization algebra
is locally constant.

\subsection{Strong deformation retract}
Recall that the input for our constructions in Section \ref{sec:PFA} is a 
strong deformation retract \eqref{eqn:PFASDR} between the $\mathrm{Disk}_{\bbR^m}$-colored object
$\FFF \in \Ch^{\mathrm{Disk}_{\bbR^m}}$ and its cohomology $\H\FFF \in \Ch^{\mathrm{Disk}_{\bbR^m}}$. 
To construct such datum for the factorization envelope $\FFF = \mathfrak{U} \g^{\bbR^m}$, 
we shall start from a strong deformation retract between the compactly supported 
de Rham complexes $\Omega^\bullet_{\bbR^m,\cc} \in \Ch^{\mathrm{Disk}_{\bbR^m}}$ and 
their cohomologies $\H \Omega^\bullet_{\bbR^m,\cc}$. Since $\H \Omega^\bullet_{\cc}(D) \cong \bbK[-m]$,
for every disk $D \in \mathrm{Disk}_{\bbR^m}$, this is equivalent to a
strong deformation retract between $\Omega^\bullet_{\bbR^m,\cc} \in \Ch^{\mathrm{Disk}_{\bbR^m}}$ 
and the constant $\mathrm{Disk}_{\bbR^m}$-colored object 
$\bbK[-m] := \{ \bbK[-m] \in \Ch : D \in \mathrm{Disk}_{\bbR^m} \} \in \Ch^{\mathrm{Disk}_{\bbR^m}}$. 
\sk

We require a sufficiently explicit model for this strong deformation retract, 
which we will build in Lemma \ref{lem: deRhamSDR} by making the following choices:
For each disk $D \in \mathrm{Disk}_{\bbR^m}$, we choose a representative 
$\omega_D \in \Omega^m_{\cc}(D)$ of the non-trivial 
cohomology class $1\in \bbK\cong \H^m \Omega^\bullet_{\cc}(D)$, i.e.\ $\int_D \omega_D = 1$,
which defines a decomposition
\begin{subequations} \label{Omega decomp}
\begin{flalign} \label{Omega m decomp}
\Omega^m_{\cc}(D) \, = \, \bbK \, \omega_D \, \oplus \, \dd_{\rm dR} \Omega_{\cc}^{m-1}(D) \quad.
\end{flalign}
For each $r \in \{1, \ldots, m-1\}$, we choose a complement $K^r(D)$ of 
$\dd_{\rm dR} \Omega_{\cc}^{r-1}(D)\subset \Omega^r_{\cc}(D)$,
which defines decompositions
\begin{flalign} \label{Omega r decomp}
\Omega^r_{\cc}(D) \, = \, K^r(D) \, \oplus \, \dd_{\rm dR} \Omega_{\cc}^{r-1}(D)\quad,
\end{flalign}
\end{subequations}
for $r \in \{1, \ldots, m-1\}$. Finally, we set $K^0(D) := \Omega^0_{\cc}(D)$. 
Comparing with \eqref{Omega m decomp}, it is convenient 
to denote $K^m(D) := \bbK \, \omega_D$. Since $\H^{r-1} \Omega^\bullet_{\cc}(D) = 0$, 
for all $r \in \{ 1, \ldots, m \}$, it follows that the differential defines an isomorphism $\dd_{\rm dR} : K^{r-1}(D) \CongTo
\dd_{\rm dR} \Omega_{\cc}^{r-1}(D)$. We denote its inverse by
\begin{flalign} \label{hD iso}
h^r_D : \dd_{\rm dR} \Omega_{\cc}^{r-1}(D) \overset{\cong}\longrightarrow K^{r-1}(D)\quad.
\end{flalign}
Extending by zero to the complements $K^r(D)$ in \eqref{Omega decomp}, 
we obtain from this a family of linear maps $h^r_D : \Omega^r_{\cc}(D) \to K^{r-1}(D)\subset \Omega^{r-1}_{\cc}(D)$,
for $r \in \{ 1, \ldots, m \}$. It will be convenient to denote by $h^0_D : \Omega^0_{\cc}(D) \to 0$
the zero map.
\begin{lem} \label{lem: deRhamSDR}
For each disk $D \in \mathrm{Disk}_{\bbR^m}$, let us define
\begin{flalign} \label{def pD iD hD}
p_D(\omega) \,:=\, \int_D \omega \quad, \qquad 
i_D(k) \,:=\, k \, \omega_D \quad, \qquad 
h_D(\omega) \, := \, - h^r_D(\omega)\quad,
\end{flalign}
for all $\omega \in \Omega^r_{\cc}(D)$ with $r \in \{ 0, \ldots, m \}$ and $k \in \bbK[-m]$.
This defines a strong deformation retract
\begin{equation}\label{eqn:deRhamSDR}
\begin{tikzcd}
\bbK[-m] \ar[r,shift right=-1ex,"i"] & \ar[l,shift right=-1ex,"p"] \Omega^\bullet_{\bbR^m,\cc} \ar[loop,out=-25,in=25,distance=30,swap,"h"]
\end{tikzcd}
\end{equation}
in the category $\Ch^{\mathrm{Disk}_{\bbR^m}}$.
\end{lem}
\begin{proof}
From the decompositions \eqref{Omega decomp} and the definition of 
$h^r_D : \Omega^r_{\cc}(D) \to K^{r-1}(D)$, it follows that
\begin{subequations} \label{dh ip id}
\begin{flalign} \label{dh ip id a}
\omega \,=\, \omega_D \int_D \omega + \dd_{\rm dR} h^m_D(\omega) \quad,
\end{flalign}
for every $\omega \in \Omega^m_{\cc}(D)$, and
\begin{flalign} \label{dh ip id b}
\omega \,=\, h^{r+1}_D(\dd_{\rm dR} \omega) + \dd_{\rm dR} h^r_D(\omega) \quad,
\end{flalign}
\end{subequations}
for every $\omega \in \Omega^r_{\cc}(D)$ with $r \in \{ 0, \ldots, m-1 \}$. 
For $r=0$ this reads $\omega = h^1_D(\dd_{\rm dR} \omega)$. 
The identity $\partial h_D = i_D \, p_D - \id$ from \eqref{eqn:deformationretract2} 
then follows from \eqref{dh ip id}, noting the minus sign in the definition of $h_D$. 
All other identities in \eqref{eqn:deformationretract2} immediately follow from the definitions.
\end{proof}
\begin{rem} \label{hD in 1d}
In dimension $m=1$, we have an explicit formula for $h^1_D$, which is 
the only non-trivial component of the homotopy $h_D$ in this case. This formula is given by
\begin{flalign} \label{hD 1 dim}
h^1_D(\omega) \,=\, \int \bigg( \omega - \omega_D \int_D \omega \bigg) \quad,
\end{flalign}
for all $\omega \in \Omega^1_{\cc}(D)$, where $\int :\Omega^1_\cc(D)\to \Omega^0(D)$ 
denotes the following indefinite integral: Writing the disk (i.e.\ interval) 
as $D = (a,b) \subseteq \bbR$, we define $\big(\int\omega\big)(x) := \int_a^x\omega$
for all $x\in D$. Note that if $\int_D\omega = 0$, 
$\int \omega\in  \Omega^0(D)$ has compact support. 
So $h^1_D(\omega)$ is compactly supported. 
\end{rem}

Applying the tensor product $\g \otimes (-) : \Ch \to \Ch$ with the Lie algebra $\g$ 
and the shift functor $[1] : \Ch \to \Ch\,,~V \mapsto V[1]$ to 
\eqref{eqn:deRhamSDR}, we obtain a strong deformation retract
\begin{equation}\label{gdeRhamSDR}
\begin{tikzcd}
\g[1-m] \ar[r,shift right=-1ex,"{i[1]}"] & \ar[l,shift right=-1ex,"{p[1]}"] \g_\cc^{\bbR^m}[1] \ar[loop,out=-25,in=25,distance=30,swap,"{h[1]}"]
\end{tikzcd}\quad,
\end{equation}
where $\g_\cc^{\bbR^m} := \g \otimes \Omega^\bullet_{\bbR^m,\cc} \in \Ch^{\mathrm{Disk}_{\bbR^m}}$
and $\g[1-m] \in \Ch^{\mathrm{Disk}_{\bbR^m}}$ denotes 
the constant $\mathrm{Disk}_{\bbR^m}$-colored object given by $\g[1-m]\cong \g \otimes \bbK[1-m]$ 
on all $D \in \mathrm{Disk}_{\bbR^m}$. 
The components $p_D[1]$, $i_D[1]$ and $h_D[1]$ of the maps in \eqref{gdeRhamSDR} are given explicitly by
\begin{flalign} \label{def pD iD hD shift}
p_D[1]\big(\X \otimes \omega\big) \,=\, \X  \int_D \omega \quad,\quad
i_D[1]\big(\X\big) \,=\,  \X \otimes \omega_D \quad,\quad
h_D[1]\big(\X \otimes \omega\big) \,=\, \X \otimes h^r_D(\omega)\quad, 
\end{flalign}
for all $\omega \in \Omega^r_{\cc}(D)$ with $r \in \{ 0, \ldots, m \}$ 
and $\X \in \g$, where the sign change in the expression for $h_D[1]$ 
compared to \eqref{def pD iD hD} comes from the fact that $h_D$ has degree $-1$ whereas 
$p_D$ and $i_D$ both have degree $0$.
Applying the symmetric algebra functor $\Sym : \Ch \to \Ch$ to \eqref{gdeRhamSDR}, 
we obtain a strong deformation retract
\begin{equation}\label{SymdeRhamSDR}
\begin{tikzcd}[column sep=10mm]
\Sym\big( \g[1-m] \big) \ar[r,shift right=-1ex,"i"] & \ar[l,shift right=-1ex,"p"] \Sym\big( \g_{\cc}^{\bbR^m}[1] \big) \ar[loop,out=-25,in=25,distance=30,swap,"h"]
\end{tikzcd}
\end{equation}
in $\Ch^{\mathrm{Disk}_{\bbR^m}}$, where by abuse of notation we denote the maps simply by 
$p$, $i$ and $h$. Their components are defined in $\Sym$-degree 
$n \in \bbZ_{\geq 0}$ by 
\begin{subequations}\label{SymhD explicit}
\begin{flalign}
p^n_D \,:=\, p_D[1]^{\otimes n}\quad,\qquad i^n_D \,:=\, i_D[1]^{\otimes n}\quad,
\end{flalign}
and the symmetrization
\begin{flalign}
h^n_D \,:=\, \frac{1}{n!} \sum_{j=1}^n \sum_{\sigma \in \Sigma_n} \tau_{\sigma} \,
\Big( \id^{\otimes (j-1)} \otimes h_D[1] \otimes \big(i_D[1]\, p_D[1]\big)^{\otimes (n-j)} \Big)\, \tau_{\sigma^{-1}} \quad,
\end{flalign}
\end{subequations}
where $\tau_{\sigma}$ denotes the action of the permutation $\sigma\in\Sigma_n$ 
via the symmetric braiding \eqref{eqn:braiding}.
The differential on $\Sym\big( \g_{\cc}^{\bbR^m}[1] \big)$ is $\dd_{\rm dR[1]}$, 
i.e.\ the differential on $\g_{\cc}^{\bbR^m}[1]$ extended by the Leibniz rule. 
\sk

The factorization envelope $\mathfrak{U} \g^{\bbR^m} \in \Alg_{\PP_{\bbR^m}}$ is 
constructed by deforming the differential $\dd_{\rm dR[1]}$ of $\Sym\big( \g_{\cc}^{\bbR^m}[1] \big)$ 
along the Chevalley-Eilenberg differential $\dd_\CE$. The latter is defined as follows:
For every disk $D \in \mathrm{Disk}_{\bbR^m}$, the \textit{unshifted} cochain complex 
$\g_{\cc}^{\bbR^m}(D)  = \g\otimes \Omega^\bullet_{\cc}(D) \in \Ch$ 
can be endowed with a dg-Lie algebra structure $[-,-] : \g_{\cc}^{\bbR^m}(D) \otimes \g_{\cc}^{\bbR^m}(D) 
\rightarrow \g_{\cc}^{\bbR^m}(D)$ by setting 
$[\X \otimes \omega, \Y \otimes \eta] := [\X, \Y] \otimes \omega \wedge \eta$, 
for all $\X, \Y \in \g$ and $\omega, \eta \in \Omega^\bullet_{\cc}(D)$. This induces a 
degree $1$ map
\begin{flalign} \label{dCE on gR I}
\dd_\CE \,:\, \Sym\big( \g_{\cc}^{\bbR^m}[1] \big)\, \longrightarrow \,\Sym\big( \g_{\cc}^{\bbR^m}[1] \big)
\end{flalign}
in $\Ch^{\mathrm{Disk}_{\bbR^m}}$, whose component at $D \in \mathrm{Disk}_{\bbR^m}$ is given by
\begin{flalign} \label{dCE def}
\dd_\CE \bigg( \prod_{a=1}^n \big(\X_a \otimes \omega_a\big) \bigg) \,:=\,
\sum_{\substack{i,j=1\\ i < j}}^n\, (-1)^{n^{\leftarrow}_{i,j}} \,(-1)^{\vert\omega_i\vert_{\dR}} \,
\big( [\X_i, \X_j] \otimes \omega_i \wedge \omega_j \big) ~
\prod_{\substack{a=1\\ a \neq i, j}}^n \big(\X_a \otimes \omega_a\big) \quad,
\end{flalign}
for all $\X_a\in \g$ and $\omega_a \in \Omega^\bullet_{\cc}(D)$ 
with $a \in \{ 1, \ldots, n \}$. The products in $\Sym(\g_{\cc}^{\bbR^m}[1])$ are
ordered from left to right and $(-1)^{n^{\leftarrow}_{i,j}}$ denotes
the Koszul sign that arises from bringing $\X_i \otimes \omega_i$ and $\X_j \otimes \omega_j$ 
to the front of the product.
This map satisfies $\dd_{\dR[1]}\, \dd_\CE + \dd_\CE \,\dd_{\dR[1]} = 0$ 
and $\dd_\CE{}^2 = 0$, hence it defines a perturbation of the 
$\mathrm{Disk}_{\bbR^m}$-colored object $\Sym\big( \g_{\cc}^{\bbR^m}[1] \big) \in \Ch^{\mathrm{Disk}_{\bbR^m}}$.
The underlying $\mathrm{Disk}_{\bbR^m}$-colored object of the factorization envelope is then defined as
\begin{flalign}
\mathfrak{U} \g^{\bbR^m} \,:=\,  \Big(\Sym\big( \g_{\cc}^{\bbR^m}[1] \big) , \dd_{\dR[1]} + \dd_{\CE}\Big)\,\in \Ch^{\mathrm{Disk}_{\bbR^m}}\quad.
\end{flalign}
\begin{rem} \label{rem: dCE action}
A trivial yet important observation is that the only non-trivial contributions 
to \eqref{dCE def} come from pairs of factors $\X_i \otimes \omega_i$ 
and $\X_j \otimes \omega_j$ for which the de Rham degrees satisfy
$\vert\omega_i\vert_{\dR} + \vert\omega_j\vert_{\dR} \leq m$.
\end{rem}

Note that the perturbation $\dd_\CE$ is small since it lowers the $\Sym$-degree by one
and the latter is bounded from below by $0$.
Hence, we can apply the homological perturbation lemma \cite{Cra04} and obtain
\begin{propo} \label{prop: fac env SDR}
There exists a perturbed strong deformation retract
\begin{equation} \label{SDR for Ug}
\begin{tikzcd}
\Sym\big( \g[1-m] \big) \arrow[r, "i", shift right=-1ex] & \mathfrak{U} \g^{\bbR^m} \arrow[l, "\widehat{p}", shift right=-1ex] \arrow["\widehat{h}"', loop, out=-25,in=25, distance=30]
\end{tikzcd}
\end{equation}
in the category $\Ch^{\mathrm{Disk}_{\bbR^m}}$
with $i$ given by \eqref{SymdeRhamSDR} and the components of $\widehat{p}$ and $\widehat{h}$ 
at any $D \in \mathrm{Disk}_{\bbR^m}$ given in terms of those of $p$ and $h$ in \eqref{SymdeRhamSDR} by
\begin{flalign}
\widehat{p}_D \,:=\, \sum_{j \geq 0} p_D \,(\dd_\CE \, h_D)^j \quad, \qquad
\widehat{h}_D \,:=\, \sum_{j \geq 0} h_D \,(\dd_\CE \, h_D)^j \quad.
\end{flalign}
\end{propo}
\begin{proof}
This follows immediately by applying the homological perturbation lemma \cite{Cra04}
to the strong deformation retract \eqref{SymdeRhamSDR} and the small perturbation $\dd_\CE$.
The reason why $i$ does not get deformed is as follows:
For any disk $D \in \mathrm{Disk}_{\bbR^m}$, the image of $i_D^n$ lies in the 
subcomplex $\Sym^n\big( \g \otimes \Omega^m_{\cc}(D)[1] \big)$ 
of $\Sym^n \big(\g_\cc^{\bbR^m}(D)[1]\big)$ that is generated by top forms. Hence,
$\dd_\CE \, i_D = 0$ by Remark \ref{rem: dCE action}, which implies 
that $\widehat{i}_D = i_D$. By the same argument, one sees that 
the (trivial) differential on $\Sym(\g[1-m])$ is also not deformed.
\end{proof}

\subsection{\label{sec: structure maps}Underlying strict prefactorization algebra structure}
The prefactorization algebra structure of the factorization envelope $\mathfrak{U}\g^{\bbR^m} \in \Alg_{\PP_{\bbR^m}}$
can be described as follows. First, we note that $\Sym\big( \g_{\cc}^{\bbR^m}[1] \big) \in \Alg_{\PP_{\bbR^m}}$
carries a canonical prefactorization algebra structure that is defined
from the pre-cosheaf structure of $\g_{\cc}^{\bbR^m}[1]$, i.e.\ extension by zero of 
compactly supported differential forms, and the commutative product of the symmetric algebra.
The Chevalley-Eilenberg differential \eqref{dCE def} is a degree $1$ 
derivation of this $\PP_{\bbR^m}$-algebra structure, i.e.\
\begin{flalign}
(\dd_\CE)_D^{~} ~ \Sym\big( \g_{\cc}^{\bbR^m}[1] \big)\big( \iota^D_{\und{D}} \big) \,=\, 
\Sym\big( \g_{\cc}^{\bbR^m}[1] \big)\big( \iota^D_{\und{D}} \big) ~ \sum_{i=1}^n (\dd_\CE)_{D_i}^{~} \quad,
\end{flalign}
for all operations $\iota_{\und{D}}^D$ in the operad $\PP_{\bbR^m}$. Hence,
we obtain an induced $\PP_{\bbR^m}$-algebra structure on the deformation
$\mathfrak{U}\g^{\bbR^m} = \big(\Sym( \g_{\cc}^{\bbR^m}[1]) , \dd_{\dR[1]} + \dd_{\CE}\big)
\in \Alg_{\PP_{\bbR^m}}$. As discussed in Subsection \ref{sec: Masseys for PFA}, this 
$\PP_{\bbR^m}$-algebra structure transfers along the strong deformation retract \eqref{SDR for Ug} 
to a $(\PP_{\bbR^m})_\infty$-algebra structure on the cohomology $\H\mathfrak{U}\g^{\bbR^m} 
= \Sym(\g[1-m]) \in \Ch^{\mathrm{Disk}_{\bbR^m}}$, which is called a minimal model for 
$\mathfrak{U}\g^{\bbR^m}\in \Alg_{\PP_{\bbR^m}}$. The aim of this subsection 
is to compute the underlying \textit{strict} $(\PP_{\bbR^m})_\infty$-algebra structure
of this minimal model, or equivalently the structure maps from
Proposition \ref{prop:strictstructuremaps}. The presence or absence of 
higher structures in the minimal model will be studied in the next subsection.
\sk

To match our notation for the perturbed strong deformation retract \eqref{SDR for Ug},
we will denote the structure maps from Proposition \ref{prop:strictstructuremaps}
by $\widehat{\mu}_{\und{D}} : \Sym(\g[1-m])^{\otimes n}\to \Sym(\g[1-m])$.
Explicitly, they are given by
\begin{flalign} \label{structure map muD}
\widehat{\mu}_{\und{D}} \,=\, \widehat{p}_{\bbR^m} ~ \mathfrak{U}\g^{\bbR^m}\big(\iota_{\und{D}}^{\bbR^m}\big) ~ 
i_{\und{D}} \,=\, \sum_{j \geq 0} p_{\bbR^m} \, (\dd_\CE \, h_{\bbR^m})^j ~ 
\mathfrak{U}\g^{\bbR^m}\big(\iota_{\und{D}}^{\bbR^m}\big) ~ i_{\und{D}} \quad,
\end{flalign}
where we recall that $\und{D} = (D_1,\dots,D_n)$ is a tuple of
mutually disjoint disks. 
Note that there is a stark difference between the cases of dimension $m=1$ and $m\geq 2$, 
stemming from the trivial observation made in Remark \ref{rem: dCE action}.
Indeed, since the image of $\mathfrak{U}\g^{\bbR^m}\big(\iota_{\und{D}}^{\bbR^m}\big) \, i_{\und{D}}$ 
consists of products of top forms in $\Sym(\g_{\cc}^{\bbR^m}[1](\bbR^m))$ and since
$h_{\bbR^m}$ only lowers the total degree of these forms by $1$, we find the following two cases:
\begin{description}
\item[Dimension $m=1$:] All terms in the sum over $j \geq 0$ in \eqref{structure map muD} can, 
in principle, contribute. This implies that the structure map
$\widehat{\mu}_{\und{D}}$ is a deformation of the standard $n$-ary commutative 
product on $\Sym(\g[1-m])$. (The latter agrees with the $j=0$ term in \eqref{structure map muD}.)

\item[Dimension $m \geq 2$:] Only the $j=0$ term in \eqref{structure map muD} can be 
non-trivial. This implies that the structure map $\widehat{\mu}_{\und{D}}$ is simply the
standard $n$-ary commutative product on $\Sym(\g[1-m])$.
\end{description}

As a consequence of this observation, we can restrict our attention
in the remainder of this subsection to the case of dimension $m=1$.
Since by Proposition \ref{prop:strictstructuremaps} the structure maps
\eqref{structure map muD} define an associative and unital algebra structure on 
$\Sym(\g)$, it suffices to determine the binary multiplication
\begin{flalign} \label{deformed product Sg}
\star \, := \, \widehat{\mu}_{(D, D^\prime)} \,:\, \Sym(\g) \otimes \Sym(\g) \,\longrightarrow\, \Sym(\g)\quad,
\end{flalign}
for any pair of disjoint intervals $D, D^\prime \in \mathrm{Disk}_\bbR$ 
with $D$ to the left of $D^\prime$ in $\bbR$, i.e.\ $D<D^\prime$.
(Note that the unit $\widehat{\mu}_{\varnothing}=1\in \Sym(\g)$ is the standard one since
the homotopy in \eqref{structure map muD} acts trivially in $\Sym$-degree $0$.)
To describe the multiplication  \eqref{deformed product Sg} on arbitrary elements 
of $\Sym(\g)$, it suffices to compute the product $\X^n \star \Y \in \Sym(\g)$ 
for arbitrary $\X, \Y \in \g$ and $n \in \bbZ_{\geq 1}$. Indeed, the product of 
$\Y$ with an arbitrary monomial $\prod_{i=1}^n \X_i \in \Sym(\g)$ 
of degree $n$ is obtained from this by polarization
\begin{flalign}
\bigg( \prod_{i=1}^n \X_i \bigg) \star \Y \,=\, \frac{1}{n!} \,\frac{\partial^n}{\partial t_1 \cdots \partial t_n} \bigg\vert_{t_1, \ldots, t_n = 0} \bigg[ \bigg( \sum_{i=1}^n t_i \,\X_i \bigg)^n \star \Y \bigg] \,\in\, \Sym(\g) \quad,
\end{flalign}
where $t_i\in\bbK$ are parameters for $i \in \{1, \ldots, n\}$,
and the product $\big(\prod_{i=1}^n \X_i\big)\star \big(\prod_{k=1}^m \Y_k\big)$
of two arbitrary monomials is obtained by induction on $m$ using the associativity 
of \eqref{deformed product Sg}. Explicitly, it follows from \eqref{structure map muD} 
that we can write $\big( \prod_{k=1}^{m-1} \Y_k \big) \star \Y_m = \prod_{k=1}^m \Y_k + A$ 
for some $A \in \Sym^{\leq m-1}(\g)$ of $\Sym$-degree at most $m-1$ (the first term $\prod_{k=1}^m \Y_k$ 
corresponds to the $j=0$ term in the sum in \eqref{structure map muD} and $A$ corresponds 
to all other terms $j > 0$). We then have
\begin{flalign}
\bigg(\prod_{i=1}^n \X_i\bigg)\star \bigg(\prod_{k=1}^m \Y_k\bigg) = \Bigg( \bigg(\prod_{i=1}^n \X_i\bigg) \star \bigg( \prod_{k=1}^{m-1} \Y_k \bigg) \Bigg) \star \Y_m - \bigg(\prod_{i=1}^n \X_i\bigg)\star A\quad,
\end{flalign}
where both terms on the right-hand side involve $\star$-products 
in which the second factor has $\Sym$-degree at most $m-1$.
\sk

We will now compute $\X^n \star \Y \in \Sym(\g)$ explicitly. First, let us note that
\begin{flalign}
i_{(D,D^\prime)}( \X^n \otimes \Y) \,=\, 
(\X \otimes \omega_D)^n \otimes (\Y \otimes \omega_{D^\prime}) \,\in\, 
\mathfrak{U}\g^\bbR(D) \otimes \mathfrak{U}\g^\bbR(D^\prime)\quad.
\end{flalign}
Applying the cochain map $\mathfrak{U}\g^\bbR\big(\iota_{(D, D^\prime)}^\bbR\big)$ 
we obtain $(\X \otimes \omega_D)^n \,(\Y \otimes \omega_{D^\prime}) \in \mathfrak{U}\g^\bbR(\bbR)$, and hence
\begin{flalign} \label{alpha on Xm Y}
\X^n \star \Y \,=\, \sum_{j \geq 0} p_\bbR \,\big( \dd_\CE \, h_\bbR \big)^j \,\big((\X \otimes \omega_D)^n \,
(\Y \otimes \omega_{D^\prime}) \big)\, \in\, \Sym(\g) \quad.
\end{flalign}
To determine the latter sum, we can make a computationally important simplification 
by using the freedom to choose our strong deformation retract in Lemma \ref{lem: deRhamSDR}
to satisfy $\omega_\bbR = \omega_{D_0}$ for some \textit{fixed} $D_0 \in \mathrm{Disk}_\bbR$. 
Using independence of the structure maps under the choice of tuples of
intervals with fixed ordering along $\bbR$ (see Proposition \ref{prop:strictstructuremaps}), we can then 
take without loss of generality $D=D_0$ and $D^\prime$ any interval to the right of $D_0$.
With this simplification, we obtain the following result.
\begin{lem} \label{lem: dCE H on Xm Y}
For any $\X, \Y \in \g$ and $j \in \{0, \ldots, n\}$, we have
\begin{flalign} \label{dCE H on Xm Y}
\big( \dd_\CE  \, h_\bbR \big)^j \,\big((\X \otimes \omega_D)^n \,(\Y \otimes \omega_{D^\prime}) \big) 
\,=\, (\X \otimes \omega_D)^{n-j} \,\big( \ad_{\X}^j (\Y) \otimes \eta^{\sb{(n)}}_j \big) \quad,
\end{flalign}
where $\ad_\X(\Y):= [\X,\Y]$ denotes the adjoint action and 
$\eta^{\sb{(n)}}_j \in \Omega^1_{\cc}(\bbR)$ is defined recursively by $\eta^{\sb{(n)}}_0 = \omega_{D^\prime}$ and
\begin{flalign} \label{recurrence eta}
\eta^{\sb{(n)}}_{j+1} \,=\, - (n-j)\, \omega_D\,\int \bigg(\eta^{\sb{(n)}}_j - \omega_D~\int_{\bbR}\eta^{\sb{(n)}}_j\bigg)\quad,
\end{flalign}
for $j \in \{ 0, \ldots, n-1 \}$.
\end{lem}
\begin{proof}
We will show that, for any $j \in \{ 0, \ldots, n-1 \}$ and $\eta^{\sb{(n)}}_j \in \Omega^1_{\cc}(\bbR)$, we have
\begin{flalign}
\dd_\CE \, h_\bbR \Big( (\X \otimes \omega_D)^{n-j} \,\big( \ad_{\X}^j (\Y) \otimes \eta^{\sb{(n)}}_j \big) \Big) \,=\, 
(\X \otimes \omega_D)^{n-j-1} \, \big( \ad_{\X}^{j+1} (\Y) \otimes \eta^{\sb{(n)}}_{j+1} \big) \quad,
\end{flalign}
with $\eta^{\sb{(n)}}_{j+1} \in \Omega^1_{\cc}(\bbR)$ given by \eqref{recurrence eta}.
The desired identity \eqref{dCE H on Xm Y} then follows by induction with the choice $\eta^{\sb{(n)}}_0 = \omega_{D^\prime}$.
\sk

Using the explicit expression for $h^{n-j+1}_\bbR$ given in \eqref{SymhD explicit}, we find that
\begin{flalign}
h^{n-j+1}_\bbR \Big((\X \otimes \omega_D)^{n-j} \,\big( \ad_{\X}^j(\Y) \otimes \eta^{\sb{(n)}}_j \big) \Big)
\,=\, (\X \otimes \omega_D)^{n-j} \,\bigg( \ad_{\X}^j (\Y) \otimes \int \bigg( \eta^{\sb{(n)}}_j - \omega_D \int_\bbR \eta^{\sb{(n)}}_j \bigg) \bigg)\quad.
\end{flalign}
In this calculation it is crucial to use the specific choice $\omega_\bbR = \omega_D$
for the strong deformation retract, which in particular implies that $i_\bbR \,  p_\bbR$ 
acts as the identity on $\X \otimes \omega_D$. Furthermore, if the factor 
$h_\bbR[1]$ in $h^{n-j+1}_\bbR$ from \eqref{SymhD explicit} acts on any of 
the $n{-}j$ factors $\X \otimes \omega_D$, then the corresponding term in the sum vanishes 
by virtue of the choice $\omega_\bbR = \omega_D$. Hence, $h_\bbR[1]$ must act 
on the last factor $ \ad_{\X}^j (\Y) \otimes \eta^{\sb{(n)}}_j$, and it can be computed
by using the explicit expression for $h_\bbR[1]$ in dimension $m=1$ given in Remark \ref{hD in 1d}.
The result now follows by applying $\dd_\CE$ given in \eqref{dCE def}.
\end{proof}

Applying $p_\bbR^{n-j+1} $ to both sides of the identity \eqref{dCE H on Xm Y}, 
we obtain
\begin{flalign}
p_\bbR \,( \dd_\CE \, h_\bbR )^j \big((\X \otimes \omega_D)^n \,(\Y \otimes \omega_{D^\prime}) \big) 
\,=\, \X^{n-j} \,\ad_{\X}^j (\Y) ~ \int_\bbR \eta^{\sb{(n)}}_j \quad,
\end{flalign}
so it remains to solve the recurrence relation \eqref{recurrence eta} 
for $\eta^{\sb{(n)}}_j \in \Omega^1_{\cc}(\bbR)$ and compute its integral over $\bbR$. 
The result is given by the following lemma.
\begin{lem}
For all $j \in \{0,\ldots,n\}$, we have
\begin{flalign}
\int_\bbR \eta^{\sb{(n)}}_j \, =\, (-1)^{j}\,\binom{n}{j}\, B_j \quad,
\end{flalign}
where $B_j$ are the Bernoulli numbers with sign convention $B_1 = -\frac{1}{2}$.
\end{lem}
\begin{proof}
First, we observe that the $n^{\rm th}$ iterated integral of $\omega_D$ 
is given by $\int_\bbR \omega_D \int \omega_D \cdots \int \omega_D = \frac{1}{n!}$ 
and that $\omega_D \int \omega_{D^\prime} = 0$ since $D$ is to the left of $D^\prime$.
This allows us to use the recurrence relation \eqref{recurrence eta} to express 
$\int_\bbR \eta^{\sb{(n)}}_j$ for any $j \in\{1,\ldots,n\}$ in terms of all the previous 
$\int_\bbR \eta^{\sb{(n)}}_r$ for $r = 0, \ldots, j-1$ as follows
\begin{flalign}
-(-1)^{j}\frac{1}{\binom{n}{j}} \int_\bbR \eta^{\sb{(n)}}_j \,=\,
\frac{1}{j+1} \,\sum_{r=0}^{j-1} \binom{j+1}{r}\, (-1)^{r} \,\frac{1}{\binom{n}{r}} \,\int_\bbR \eta^{\sb{(n)}}_r \quad.
\end{flalign}
The statement now follows from the fact that the Bernoulli numbers satisfy the same recurrence relation, namely,
for all $j>0$,
\begin{flalign}
\sum_{r=0}^j\binom{j+1}{r} \, B_r \,=\,0\quad \Longleftrightarrow\quad
-B_j \,=\, \frac{1}{j+1}\, \sum_{r=0}^{j-1} \binom{j+1}{r} \, B_r\quad,
\end{flalign}
and the same initial condition $B_0 = 1 = \int_{\bbR}\omega_{D^\prime} = \int_\bbR \eta^{\sb{(n)}}_0$.
\end{proof}

From the results above, we can deduce that
\eqref{alpha on Xm Y} is given explicitly by
\begin{flalign}
\X^n \star \Y \,=\, \sum_{j=0}^n (-1)^j\,\binom{n}{j} \, B_j ~ \X^{n-j} \,\ad_{\X}^j (\Y) \quad.
\end{flalign}
This identity characterizes the Gutt star-product \cite{Gutt} on $\Sym (\g)$, see e.g.\ 
\cite[Proposition 2.7]{ESW}, hence the associative and unital product \eqref{deformed product Sg} on $\Sym(\g)$ 
coincides with the Gutt star-product. Summing up, we have shown the following result, which 
was observed before in \cite[Proposition 3.4.1]{CostelloGwilliam} by using different techniques.
\begin{propo} \label{prop: Ug}
The $\PP_\bbR$-algebra structure on the factorization envelope
$\mathfrak{U}\g^\bbR \in \Alg_{\PP_\bbR}$ transfers along the strong deformation retract 
\eqref{SDR for Ug} to the associative and unital algebra structure on $\Sym(\g)$ that is given
by the Gutt star-product. The latter is isomorphic to the universal enveloping algebra $U\g$.
\end{propo}

\subsection{\label{subsec:universal1stMassey}Universal first-order Massey product}
On top of the underlying strict $(\PP_{\bbR^m})_\infty$-algebra
structure on the cohomology $\H\mathfrak{U}\g^{\bbR^m}$ that we have described
in Subsection \ref{sec: structure maps}, there may exist additional higher structures
in the sense of Subsections \ref{subsec:minmod} and \ref{sec: universal Massey}.
The first instance of such higher structures is given by the universal first-order
Massey product, which in the context of locally constant prefactorization algebras
was described in Subsection \ref{sec: lcPFA}. The aim of this subsection
is to compute and interpret the universal first-order Massey products for the factorization
envelopes $\mathfrak{U}\g^{\bbR^m}\in\Alg_{\PP_{\bbR^m}}$ on $\bbR^m$.
\sk

Following the approach in Subsection \ref{sec: lcPFA}, we can
use local constancy of the factorization envelope $\mathfrak{U}\g^{\bbR^m}$ 
to improve the strong deformation retract \eqref{SDR for Ug} as in Lemma \ref{lem:PFASDRimproved}.
Then the $1$-cocycle representing the universal first-order Massey product \eqref{eqn:improvedMassey}
becomes
\begin{flalign}\label{Massey example}
\widetilde{\beta}^{(2)}\big(\iota_{\und{D}}^D,\iota_{\und{D}_i}^{D_i}\big)\,=\,
\widehat{\beta}^{(2)}\big(\iota_{\und{D}}^{\bbR^m},\iota_{\und{D}_i}^{D_i}\big)
- \widehat{\mu}_{\und{D}}~\widehat{\beta}^{(2)}\big(\iota_{D_i}^{\bbR^m},\iota_{\und{D}_i}^{D_i}\big) \quad,
\end{flalign}
where $\widehat{\mu}_{\und{D}}$ denotes the structure maps from Subsection \ref{sec: structure maps}
and $\widehat{\beta}^{(2)}$ is defined by the commutative diagram \eqref{eqn:firstMasseyonRm} 
in terms of the $1$-cocycle $\beta^{(2)}$ associated with the original strong deformation retract \eqref{SDR for Ug}.
In our present example, the identification between $\beta^{(2)}$ and $\widehat{\beta}^{(2)}$ is trivial
because $\H\mathfrak{U}\g^{\bbR^m}\!(\iota_D^{D^\prime}) = \id$, for all $D\subset D^\prime$,
hence we obtain the explicit expression
\begin{flalign}\label{eqn:hatMasseyexample}
\widehat{\beta}^{(2)}\big(\iota_{\und{D}}^{\bbR^m},\iota_{\und{D}_i}^{D_i}\big) \,=\,
\widehat{p}_{\bbR^m}\, \mathfrak{U}\g^{\bbR^m}\!\big(\iota_{\und{D}}^{\bbR^m}\big)~\widehat{h}_{D_i}~ 
\mathfrak{U}\g^{\bbR^m}\!\big(\iota_{\und{D}_i}^{D_i}\big)~i_{(D_1,\dots,\und{D}_i,\dots, D_n)}\quad.
\end{flalign}
We observe that the behavior of this $1$-cocycle depends strongly on the dimension $m$. 
\begin{description}
\item[Dimension $m=1$:] The cohomology $\H\mathfrak{U}\g^{\bbR}\cong \Sym(\g)$ is concentrated
in cohomological degree $0$, hence the degree $-1$ operations
$\widehat{\beta}^{(2)}\big(\iota_{\und{D}}^{\bbR^m},\iota_{\und{D}_i}^{D_i}\big) : \Sym(\g)^{\otimes (n+n_i-1)}\to \Sym(\g)$
all vanish on degree grounds. More generally, by the same reason, 
the transferred $(\PP_{\bbR^m})_\infty$-algebra structure $\beta$ on $\H\mathfrak{U}\g^{\bbR^m}$
vanishes on all trees with more than one vertex. This implies that the 
factorization envelope $\mathfrak{U}\g^{\bbR}\in \Alg_{\PP_{\bbR}}$ in $m=1$ dimensions
is formal, i.e.\ it is $\infty$-quasi-isomorphic to its cohomology 
$\H\mathfrak{U}\g^{\bbR}\cong \Sym(\g)$ endowed with the strict prefactorization algebra structure.
By a similar argument, an analogous result is found in \cite[Section 5.3]{PFAKoszul}. 

\item[Dimension $m\geq 3$:] Using the concrete expression for the
strong deformation retract \eqref{SDR for Ug}, we can write \eqref{eqn:hatMasseyexample} more explicitly as
\begin{flalign}\label{eqn:degreeargument m>2}
\widehat{\beta}^{(2)}\big(\iota_{\und{D}}^{\bbR^m},\iota_{\und{D}_i}^{D_i}\big)
=\sum_{j,k \geq 0} \! p_{\bbR^m} (\dd_\CE \, h_{\bbR^m})^j 
\mathfrak{U}\g^{\bbR^m}\!\big(\iota_{\und{D}}^{\bbR^m}\big)\, h_{D_i} (\dd_\CE \, h_{D_i})^k \, 
\mathfrak{U}\g^{\bbR^m}\!\big(\iota_{\und{D}_i}^{D_i}\big)\, i_{(D_1,\dots,\und{D}_i,\dots, D_n)} \quad.
\end{flalign}
As a consequence of Remark \ref{rem: dCE action}, one would need at
least $m \geq 3$ homotopies $h$ preceding any $\dd_\CE$ to obtain a
non-trivial contribution to the sum. This is obviously impossible,
which implies that all $\widehat{\beta}^{(2)}\big(\iota_{\und{D}}^{\bbR^m},\iota_{\und{D}_i}^{D_i}\big)=0$
vanish in dimension $m\geq 3$, and consequently
all $\widetilde{\beta}^{(2)}\big(\iota_{\und{D}}^D,\iota_{\und{D}_i}^{D_i}\big)=0$ 
in \eqref{Massey example} vanish as well. 
This implies that the factorization envelope $\mathfrak{U}\g^{\bbR^{m}}\in \Alg_{\PP_{\bbR^m}}$ in $m\geq 3$ 
dimensions has a trivial universal first-order Massey product 
$[\widetilde{\beta}^{(2)}] = 0$, but it still may have non-trivial higher-order Massey products
in the sense of \cite{Dimitrova,UniversalMassey3}. See also the related discussion in Remark \ref{rem:PFAEmPm}.
\end{description}

As a consequence of this observation, we restrict our attention
in the remainder of this subsection to the case of dimension $m=2$.
Our goal is to show that the cohomology class $\big[\widetilde{\beta}^{(2)}\big]\in\H^1_{\Gamma}(\H\mathfrak{U}\g^{\bbR^2})$
is non-trivial, which implies that the cohomology $\H\mathfrak{U}\g^{\bbR^2}\cong \Sym(\g[-1])$ carries,
on top of its commutative algebra structure from Subsection \ref{sec: structure maps},
non-trivial higher structures. To show non-triviality of $\big[\widetilde{\beta}^{(2)}\big]$,
we will study the simpler $L$-invariant from Subsection \ref{subsec:2dinvariant}, which is available 
in our present case of $m=2$ dimensions.
\sk

We start by computing explicitly the value of the improved $1$-cocycle $\widetilde{\beta}^{(2)}$ in
\eqref{Massey example} on trees of the form \eqref{eqn:complicatedtrees}, and use this result to argue
that the $L$-invariant \eqref{eqn:Lcombination} and hence the
cohomology class $\big[\widetilde{\beta}^{(2)}\big] \neq 0$ is non-trivial.
For our argument it suffices to focus on trees with only two free edges
\begin{flalign} \label{binary tree}
t\big(\iota_{(D_1^{\prime},D_2)}^{D}, \iota_{D_1}^{D_1^{\prime}}\big)~~=~~
\parbox{2cm}{\begin{tikzpicture}[cir/.style={circle,draw=black,inner sep=0pt,minimum size=2mm},
        poin/.style={rectangle, inner sep=2pt,minimum size=0mm},scale=0.8, every node/.style={scale=0.8}]
\node[poin] (i)   at (0,0) {};
\node[poin] (v1)  at (0,-1) {{$\iota_{(D_1^{\prime},D_2)}^{D}$}};
\node[poin] (v2)  at (-0.75,-2) {{$\iota^{D_1^{\prime}}_{D_1}$}};
\node[poin] (o1)  at (0.75,-2) {};
\node[poin] (o22)  at (-0.75,-3) {};
\draw[thick] (i) -- (v1);
\draw[thick] (v1) -- (v2);
\draw[thick] (v1) -- (o1);
\draw[thick] (v2) -- (o22);
\end{tikzpicture}}\quad,
\end{flalign}
for all $D_1, D_1^{\prime}, D_2, D \in \mathrm{Disk}_{\bbR^2}$ 
such that $D_1 \subset D_1^{\prime}$, $D_2 \cap D_1^{\prime} = \emptyset$ and 
$D_1^{\prime} \sqcup D_2 \subset D$. The values of the $1$-cocycle \eqref{Massey example} on 
such trees are degree $-1$ binary operations
\begin{flalign}
\widetilde{\beta}^{(2)}\big(\iota_{(D_1^{\prime},D_2)}^{D}, \iota_{D_1}^{D_1^{\prime}}\big) 
\,:\, \Sym(\g[-1]) \otimes \Sym(\g[-1])\, \longrightarrow \, \Sym(\g[-1])
\end{flalign}
on $\Sym(\g[-1]) \in \Ch$, which we will compute in Lemma \ref{lem: Massey computation} below.
\sk

In order to ease the notation in what follows, we will always suppress the
extension maps $\mathfrak{U}\g^{\bbR^2}\!\big(\iota_{D}^{D^\prime}\big)$
for all disk inclusion $D \subset D^\prime$, which is justified since
they are simply extension by zero of compactly supported forms from $D$ to $D^\prime$.
Then the first term on the right-hand side of \eqref{Massey example} explicitly reads
\begin{flalign}
\nn \widehat{\beta}^{(2)}\big(\iota_{(D_1^{\prime},D_2)}^{\bbR^2},\iota_{D_1}^{D_1^{\prime}}\big) \,&=\,
\widehat{p}_{\bbR^2}\, \mathfrak{U}\g^{\bbR^2} \!\big(\iota_{(D_1^{\prime},D_2)}^{\bbR^2}\big)\, \widehat{h}_{D_1^{\prime}} 
\, i_{(D_1, D_2)}\\[4pt]
&=\,
\sum_{j,k \geq 0} p_{\bbR^2} \,(\dd_\CE \, h_{\bbR^2})^j \, \mathfrak{U}\g^{\bbR^2} \!\big(\iota_{(D_1^{\prime},D_2)}^{\bbR^2}\big)\, h_{D_1^{\prime}} \,(\dd_\CE \, h_{D_1^{\prime}})^k \, i_{(D_1,D_2)} \quad.
\end{flalign}
It follows from Remark \ref{rem: dCE action} 
that the only non-vanishing term in the sum is given by $k=0$ and $j=1$,
hence the above reduces to
\begin{flalign}
\widehat{\beta}^{(2)}\big(\iota_{(D_1^{\prime},D_2)}^{\bbR^2},\iota_{D_1}^{D_1^{\prime}}\big) \,=\,
p_{\bbR^2} \, \dd_\CE \, h_{\bbR^2} \, \mathfrak{U}\g^{\bbR^2}\! \big(\iota_{(D_1^{\prime},D_2)}^{\bbR^2}\big) 
\, h_{D_1^{\prime}} \, i_{(D_1,D_2)} \quad.
\end{flalign}
Similarly, for the $\widehat{\beta}^{(2)}$ factor in
the second term on the right-hand side of \eqref{Massey example}, we have
\begin{flalign}
\widehat{\beta}^{(2)}\big(\iota_{D_1^{\prime}}^{\bbR^2},\iota_{D_1}^{D_1^{\prime}}\big) \,=\,
\widehat{p}_{\bbR^2}\, \widehat{h}_{D_1^{\prime}} \, i_{D_1} \,=\,
p_{\bbR^2} \, \dd_\CE \, h_{\bbR^2} \, h_{D_1^\prime} \, i_{D_1} \quad.
\end{flalign}
Using also the result from Subsection \ref{sec: structure maps} that 
$\widehat{\mu}_{(D_1^\prime,D_2)} = \mu_2$
is the standard commutative binary multiplication of the symmetric algebra,
we can put the above together and find
that the value of the $1$-cocycle \eqref{Massey example} on the trees \eqref{binary tree} can be written explicitly as
\begin{flalign} \label{Massey explicit R2}
\widetilde{\beta}^{(2)}\big(\iota_{(D_1^{\prime},D_2)}^D,\iota_{D_1}^{D_1^{\prime}}\big)\, =\,
p_{\bbR^2} \, \dd_\CE \, h_{\bbR^2} \, \mathfrak{U}\g^{\bbR^2}\! \big(\iota_{(D_1^{\prime},D_2)}^{\bbR^2}\big) 
\, h_{D_1^{\prime}} \, i_{(D_1,D_2)} -\mu_2 \, p_{\bbR^2} \, \dd_\CE \, h_{\bbR^2} \, h_{D_1^{\prime}} \, i_{D_1} \quad. 
\end{flalign}
Our first result is the following
\begin{lem} \label{lem: Massey computation}
For every pair of elements $A, B \in \Sym(\g[-1])$, we have that
\begin{flalign}\label{eqn:explicitt21Masseypre}
\widetilde{\beta}^{(2)}\big(\iota_{(D_1^{\prime},D_2)}^D,\iota_{D_1}^{D_1^{\prime}}\big)(A \otimes B) \,=\, 
\varphi_0\big(\iota_{(D_1^{\prime},D_2)}^D,\iota_{D_1}^{D_1^{\prime}}\big) ~ \{ A, B \}_{(-1)}^{}\quad,
\end{flalign} 
where 
\begin{flalign} \label{shifted PB}
\{ - , - \}_{(-1)}^{} \, : \, \Sym(\g[-1]) \otimes \Sym(\g[-1]) \,\longrightarrow\, \Sym(\g[-1])
\end{flalign}
denotes the degree $-1$ Poisson bracket (i.e.\ $\mathbb{P}_2$-algebra structure) 
that is defined on the generators by $\{ \X, \Y \}_{(-1)}^{} := [\X, \Y]$, for all $\X,\Y \in \g[-1]$, 
and the numerical prefactor is given explicitly by 
\begin{flalign}\label{eqn:varphi0}
\varphi_0\big(\iota_{(D_1^{\prime},D_2)}^D,\iota_{D_1}^{D_1^{\prime}}\big) 
\,:=\, \frac{1}{2}\, \int_{\bbR^2} \Big( h^1_{\bbR^2} \big( h^2_{D_1^\prime}(\omega_{D_1}) \big) 
\wedge (\omega_{D_2} + \omega_{\bbR^2}) + h^2_{\bbR^2}(\omega_{D_2}) \wedge h^2_{D_1^\prime}(\omega_{D_1}) \Big) \quad.
\end{flalign}
\end{lem}
\begin{proof}
To simplify notation, we will denote in this proof the given tree by 
$t = t\big(\iota_{(D_1^{\prime},D_2)}^D,\iota_{D_1}^{D_1^{\prime}}\big)$.
By linearity, it suffices to consider the case where
$A = \prod_{a=1}^n \X_a \in \Sym^n(\g[-1])$ and $B = \prod_{b=1}^r \Y_b \in \Sym^r(\g[-1])$ 
are two arbitrary monomials in $\Sym(\g[-1])$.
The result follows from a direct computation of the evaluation of \eqref{Massey explicit R2} on $A \otimes B$. 
This computation is straightforward but rather lengthy, so we will only highlight the key points. 
Evaluating the first term on the right-hand side of \eqref{Massey explicit R2} on $A \otimes B$, 
we find
\begin{flalign}
\nn &p_{\bbR^2} \, \dd_\CE \, h_{\bbR^2} \, \mathfrak{U}\g^{\bbR^2} \big(\iota_{(D_1^{\prime},D_2)}^{\bbR^2}\big) 
\, h_{D_1^\prime} \, i_{(D_1,D_2)} (A \otimes B)\\
\nn &\qquad\qquad = \, - \psi(t)\, \sum_{\substack{i,j=1\\ i < j}}^n (-1)^{i+j} \, [\X_i, \X_j] \, 
\prod_{\substack{a=1\\ a\neq i, j}}^n \X_a\,  \prod_{b=1}^r \Y_b \\
&\qquad\qquad\qquad\qquad - \varphi_0(t)\, \sum_{i=1}^n \sum_{j=1}^r (-1)^{n+i+j} \, [\X_i, \Y_j]\, 
\prod_{\substack{a=1\\ a\neq i}}^n  \X_a \, \prod_{\substack{b=1\\ b \neq j}}^r \Y_b\quad,  \label{first term Massey}
\end{flalign}
where $\varphi_0(t) \in \bbK$ is defined as in the statement of the lemma and $\psi(t) \in \bbK$ is given by the integral
\begin{flalign}
\psi(t) \,:=\, \frac{1}{2}\, \int_{\bbR^2} \Big( h^1_{\bbR^2} \big( h^2_{D_1^{\prime}}(\omega_{D_1}) \big)
\wedge (\omega_{D_1} + \omega_{D_1^\prime} + 2 \omega_{\bbR^2}) + h^2_{\bbR^2}(\omega_{D_1} + \omega_{D_1^\prime}) 
\wedge h^2_{D_1^\prime}(\omega_{D_1}) \Big) \quad.
\end{flalign}
On the other hand, evaluating the second term on the right-hand side of \eqref{Massey explicit R2} on $A \otimes B$, 
we find
\begin{flalign}
\mu_2\, p_{\bbR^2} \, \dd_\CE \, h_{\bbR^2} \, h_{D_1^\prime} \, i_{D_1}(A \otimes B)
\, = \, - \psi(t)\, \sum_{\substack{i,j=1\\ i < j}}^n (-1)^{i+j} \, [\X_i, \X_j] \,
\prod_{\substack{a=1\\ a\neq i, j}}^n  \X_a \, \prod_{b=1}^r \Y_b \quad.
\end{flalign}
This coincides with the first term on the right-hand side of \eqref{first term Massey} 
and therefore they cancel in the expression for $\widetilde{\beta}^{(2)}(t)(A \otimes B)$. The latter 
then reads as
\begin{flalign}
\widetilde{\beta}^{(2)}(t)(A\otimes B) \,=\, - \varphi_0(t)\, \sum_{i=1}^n \sum_{j=1}^r (-1)^{n+i+j} \,[\X_i, \Y_j]\, \prod_{\substack{a=1\\ a\neq i}}^n \X_a\,  \prod_{\substack{b=1\\ b \neq j}}^r \Y_b \,=\, 
\varphi_0(t) \, \{ A, B \}_{(-1)}^{}\quad,
\end{flalign}
where in the last equality we have used the definition of the degree $-1$ Poisson bracket \eqref{shifted PB}
and its biderivation property. 
\end{proof}

In the next proposition, we obtain a simpler expression for 
\eqref{eqn:explicitt21Masseypre} by exploiting the freedom to add 
a suitable coboundary term to $\widetilde{\beta}^{(2)}$.
\begin{propo} \label{prop: Massey simplify}
There exists a $0$-cochain $\chi \in \Gamma^0\big(\Sym(\g[-1]) \big)$,
satisfying the conditions from Remark \ref{rem:computingMassey},
such that $\widetilde{\beta}^{\prime(2)} := \widetilde{\beta}^{(2)} + \partial_\Gamma \chi$ takes the values
\begin{flalign}
\widetilde{\beta}^{\prime(2)}\big(\iota_{(D_1^{\prime},D_2)}^D,\iota_{D_1}^{D_1^{\prime}}\big)(A \otimes B)\, =\,
\varphi\big(\iota_{(D_1^{\prime},D_2)}^D,\iota_{D_1}^{D_1^{\prime}}\big) ~ \{ A, B \}_{(-1)}^{}\quad,
\end{flalign}
for all trees $t\big(\iota_{(D_1^{\prime},D_2)}^D,\iota_{D_1}^{D_1^{\prime}}\big)$ and all 
$A, B \in \Sym(\g[-1])$, where the numerical prefactor is now given by the simpler integral
\begin{flalign}\label{eqn:varphi}
\varphi\big(\iota_{(D_1^{\prime},D_2)}^D,\iota_{D_1}^{D_1^{\prime}}\big) 
\,:=\, \int_{\bbR^2} h^1_{\bbR^2} \big( h^2_{D_1^{\prime}}(\omega_{D_1}) \big) \wedge \omega_{D_2} \quad.
\end{flalign}
\end{propo}
\begin{proof}
Recalling the explicit form \eqref{eqn:gaugetransformationPFA} of
the transformation formula $\widetilde{\beta}^{\prime(2)} = \widetilde{\beta}^{(2)} + \partial_\Gamma \chi$,
we find for our trees
\begin{flalign} \label{beta2 change gauge}
\widetilde{\beta}^{\prime(2)}\big(\iota_{(D_1^{\prime},D_2)}^D,\iota_{D_1}^{D_1^{\prime}}\big)  
\,=\, \widetilde{\beta}^{(2)}\big(\iota_{(D_1^{\prime},D_2)}^D,\iota_{D_1}^{D_1^{\prime}}\big) 
- \chi\big( \iota^D_{(D_1, D_2)} \big)
+ \chi\big( \iota^D_{(D_1^\prime, D_2)} \big) + \mu_2 \, \chi \big( \iota_{D_1}^{D_1^\prime} \big)\quad.
\end{flalign}
We define the $0$-cochain $\chi$ on trees with one and two free edges by
\begin{flalign}\label{eqn:gaugetransformationforFE}
\chi\big( \iota^{\widetilde{D}}_{D} \big) (A) = 0 \quad, \qquad
\chi\big( \iota^{\widetilde{D}}_{(D, D^\prime)} \big)(A \otimes B) \,:=\, 
\frac{1}{2} \int_{\bbR^2} h^2_{\bbR^2}(\omega_{D^\prime}) \wedge h^2_{\bbR^2}(\omega_D) ~ \{A, B\}_{(-1)}^{}\quad,
\end{flalign}
for all $A, B \in \Sym(\g[-1])$, and by zero on all other trees.
Note that this satisfies the conditions \eqref{eqn:improvedgauge2} from Remark \ref{rem:computingMassey}.
Then the last term in \eqref{beta2 change gauge} vanishes and evaluation on $A \otimes B$ gives
\begin{flalign}
\nn &\widetilde{\beta}^{\prime(2)}\big(\iota_{(D_1^{\prime},D_2)}^D,\iota_{D_1}^{D_1^{\prime}}\big)(A \otimes B)\\ 
\nn &\quad \,=\, \bigg( \varphi_0\big(\iota_{(D_1^{\prime},D_2)}^D,\iota_{D_1}^{D_1^{\prime}}\big)   
+ \frac{1}{2} \int_{\bbR^2} h^2_{\bbR^2}(\omega_{D_2}) \wedge \big( \! - h^2_{\bbR^2}(\omega_{D_1}) + h^2_{\bbR^2}(\omega_{D_1^{\prime}}) \big) \bigg)\, \{ A, B \}_{(-1)}^{}\\
\nn &\quad \,=\, \frac{1}{2} \int_{\bbR^2} \Big( h^1_{\bbR^2} \big( h^2_{D_1^\prime}(\omega_{D_1}) \big) 
\wedge (\omega_{D_2} + \omega_{\bbR^2}) \, + \\
&\qquad~\qquad~\qquad  h^2_{\bbR^2}(\omega_{D_2}) \wedge \big( h^2_{D_1^\prime}(\omega_{D_1})
- h^2_{\bbR^2}(\omega_{D_1}) + h^2_{\bbR^2}(\omega_{D_1^\prime}) \big) \Big) \,\{ A, B \}_{(-1)} \quad.
\end{flalign}
Now using the identities \eqref{dh ip id}, we have
\begin{flalign}
\nn h^2_{D_1^{\prime}}(\omega_{D_1}) \,&=\, h^2_{\bbR^2}\big( \dd_{\rm dR} h^2_{D_1^{\prime}}(\omega_{D_1}) \big) 
+ \dd_{\rm dR} h^1_{\bbR^2}\big( h^2_{D_1^\prime}(\omega_{D_1}) \big)\\
&=\, h^2_{\bbR^2}( \omega_{D_1} ) - h^2_{\bbR^2}( \omega_{D_1^\prime} ) 
+ \dd_{\rm dR} h^1_{\bbR^2}\big( h^2_{D_1^\prime}(\omega_{D_1}) \big) \quad,
\end{flalign}
where in the second line we have used $\int_{D^\prime_1} \omega_{D_1} = 1$.
The result now follows using this identity in the last line above, 
then using Stokes' theorem and the fact that 
$\dd_{\rm dR} h^2_{\bbR^2}(\omega_{D_2}) = \omega_{D_2} - \omega_{\bbR^2}$, 
which follows again by \eqref{dh ip id a}.
\end{proof}

\begin{theo}\label{theo:LinvariantFacEnvelope}
The $L$-invariant \eqref{eqn:Lcombination} 
for the factorization envelope $\mathfrak{U}\g^{\bbR^2}\in\Alg_{\PP_{\bbR^2}}$ on $\bbR^2$
coincides with the degree $-1$ Poisson bracket from \eqref{shifted PB}, i.e.\
\begin{flalign}
L \,= \, \{-,-\}_{(-1)}^{}\quad.
\end{flalign}
It follows that the cohomology class $\big[ \widetilde{\beta}^{(2)} \big] \in \H^1_\Gamma\big( \Sym(\g[-1])\big)$ 
defining the universal first-order Massey product is non-trivial for $\mathfrak{U}\g^{\bbR^2}\in\Alg_{\PP_{\bbR^2}}$.
\end{theo}
\begin{proof}
Using Proposition \ref{prop:Linvariant}, 
we can equivalently consider the simpler transformed $1$-cocycle $\widetilde{\beta}^{\prime(2)}$
from Proposition \ref{prop: Massey simplify}. Inserting this into the definition 
of $L$ from \eqref{eqn:Lcombination}, we find
\begin{flalign}
L(A \otimes B) \,=\, \int_{\bbR^2} \Big( h^1_{\bbR^2} \big( h^2_{D_u}(\omega_{D_1} - \omega_{D_2}) \big) -
h^1_{\bbR^2} \big( h^2_{D_d}(\omega_{D_1} - \omega_{D_2}) \big) \Big) \wedge \omega_{\widetilde{D}} ~
\{ A, B \}_{(-1)}^{} \quad.
\end{flalign}
Since $\int_{D_u} (\omega_{D_1} - \omega_{D_2}) = 0$, 
we have $\omega_{D_1} - \omega_{D_2} = \dd_{\rm dR} \beta_u$
for $\beta_u = h^2_{D_u} (\omega_{D_1} - \omega_{D_2}) \in \Omega^1_{\cc}(D_u)$ 
and $\omega_{D_1} - \omega_{D_2} = \dd_{\rm dR} \beta_d $
for $\beta_d = h^2_{D_d} (\omega_{D_1} - \omega_{D_2}) \in \Omega^1_{\cc}(D_d)$
by identity \eqref{dh ip id a}. 
Similarly, because $\dd_{\rm dR} (\beta_u - \beta_d) = 0$ 
we have $\beta_u - \beta_d = \dd_{\rm dR} g$
for $g = h^1_{\bbR^2} (\beta_u - \beta_d) \in \Omega^0_{\cc} (D)$ 
by \eqref{dh ip id b}, 
noting that because $\dd_{\rm dR} g = 0$ outside of $D_u \cup D_d$, 
$g \in \Omega^0_{\cc} (\widehat{D})$ for any disk $\widehat{D}$ containing $D_u$ and $D_d$. 
So  
\begin{flalign}
L(A \otimes B) 
\,=\, \int_{\bbR^2} h^1_{\bbR^2} ( \beta_u - \beta_d) \wedge \omega_{\widetilde{D}} ~ \{ A, B \}_{(-1)}
\,=\, \int_{\bbR^2} g \wedge \omega_{\widetilde{D}} ~ \{ A, B \}_{(-1)}^{} \quad .
\end{flalign}

Now, since $\dd_{\rm dR} g\vert_{\widetilde{D}}  = (\beta_u - \beta_d)\vert_{\widetilde{D}} = 0$, 
it follows that the restriction $g\vert_{\widetilde{D}}$ is constant and we claim that $g\vert_{\widetilde{D}}=1$.
Indeed, $g$ can be constructed explicitly as the line integral $g(p) = \int_{\gamma_p} (\beta_u - \beta_d)$ 
along any path $\gamma_p$ from a fixed base point $p_0 \not \in D$ to $p \in D$. The definition is independent 
of the choice of path by Stokes' theorem since $\dd_{\rm dR} (\beta_u - \beta_d) = 0$. 
Now let $p \in \widetilde{D}$ and pick any path $\gamma_p \in \bbR^2 \setminus D_d$. 
Then $g(p) = \int_{\gamma_p} \beta_u$. Let $\gamma^\prime_p \in \bbR^2 \setminus D_u$ be
another path such that $\gamma_p - \gamma^\prime_p$ forms a cycle
bounding a region $R \supset D_1$ with $R \cap D_2 = \emptyset$. Then $\int_{\gamma^\prime_p} \beta_u = 0$ 
and so
\begin{flalign}
g(p) \,=\, \int_{\gamma_p} \beta_u - \int_{\gamma^\prime_p} \beta_u \,=\, 
\int_{\partial R} \beta_u \,=\, \int_R \dd_{\rm dR} \beta_u \,=\, 
\int_R (\omega_{D_1} - \omega_{D_2}) \,=\, \int_R \omega_{D_1} \,=\, 1
\end{flalign}
as claimed, where in the third equality we used Stokes' theorem. 
So we have 
\begin{flalign}
L(A \otimes B) 
\,=\, \int_{\bbR^2} g \wedge \omega_{\widetilde{D}} ~ \{ A, B \}_{(-1)}^{} 
\,=\, \{ A, B \}_{(-1)}^{} 
\quad 
\end{flalign}
because $\int_{\bbR^2}  \omega_{\widetilde{D}} = 1$ by assumption. 
\end{proof}


\section{\label{sec:CS}Linear Chern-Simons theory}
Throughout this section, we fix the field $\bbK$ to be either the
real numbers $\bbR$ or the complex numbers $\bbC$. Recall
from \cite[Chapter 4.5]{CostelloGwilliam} that the prefactorization
algebra describing linear Chern-Simons theory (i.e.\ with structure group $\bbR$)
on an oriented $3$-manifold $M$ is given by
\begin{flalign}\label{eqn:CSPFA}
\FFF_{\mathrm{CS}}^{} \,:=\, 
\Big(\Sym\big( \Omega^\bullet_{M,\cc}[2] \big) , \dd_{\dR[2]} + \Delta_{\mathrm{BV}}\Big)\,\in \Alg_{\PP_{M}}\quad,
\end{flalign}
where $\Omega^\bullet_{M,\cc}$ is the cosheaf of compactly supported differential forms on $M$
and $\Delta_{\mathrm{BV}}$ is the BV Laplacian. The latter is the second-order differential operator
on $\Sym\big( \Omega^\bullet_{M,\cc}[2] \big)$ defined by
\begin{flalign}
\Delta_{\mathrm{BV}}\bigg( \prod_{a=1}^n \alpha_a \bigg) \,:= \, \sum_{\substack{i,j=1\\i<j}}^n (-1)^{n^{\leftarrow}_{i,j}} \, \bigg(\int_M\alpha_i\wedge \alpha_j\bigg)~ \prod_{\substack{a=1\\ a \neq i, j}}^n \alpha_a \quad,
\end{flalign}
for all $\alpha_a \in \Omega^\bullet_{M,\cc}[2]$ with $a \in \{ 1, \ldots, n\}$, 
where $(-1)^{n^\leftarrow_{i,j}}$ is the Koszul sign arising from bringing 
$\alpha_i$ and $\alpha_j$ to the front of the product.
This prefactorization algebra is locally constant.
\sk

The aim of this section is to describe the universal first-order Massey product
of the prefactorization algebra $\FFF_{\mathrm{CS}}^{}$. By the same argument
as around \eqref{eqn:degreeargument m>2}, one easily sees that the associated 
cohomology class $\big[\widetilde{\beta}^{(2)}\big]=0$ is trivial
for linear Chern-Simons theory on the $3$-dimensional Cartesian space $M=\bbR^3$. This 
is compatible with the more abstract point of view explained in Remark \ref{rem:PFAEmPm}:
Since $\FFF_{\mathrm{CS}}^{}$ on $\bbR^3$ is equivalent to an $\mathbb{E}_3$-algebra,
one does not expect a non-trivial first-order Massey product. In order to observe
a non-trivial cohomology class $\big[\widetilde{\beta}^{(2)}\big]\neq 0$, we will consider
the compactification of linear Chern-Simons theory on the $3$-manifold $M=\bbR^2\times \bbS^1$,
with $\bbS^1$ denoting the circle, which we regard as a $2$-dimensional prefactorization 
algebra on the $\bbR^2$-factor. This example can be treated with similar methods as 
in Section \ref{sec:examples}, hence we can be relatively brief in presenting the results.
\sk

As a first step, we observe that the compactified 
prefactorization algebra $\FFF_{\mathrm{CS}}^{}$ on $\bbR^2\times \bbS^1$
can be replaced by a weakly equivalent model that is more suitable for the
computations in this section. For every $2$-dimensional disk $D\in\mathrm{Disk}_{\bbR^2}$,
we have by K\"unneth's theorem a quasi-isomorphism
\begin{flalign}
\Omega^\bullet_\cc(D) \otimes \Omega^\bullet(\bbS^1) \,\stackrel{\sim}{\longrightarrow}\, \Omega^\bullet_\cc(D\times \bbS^1)\quad.
\end{flalign}
This family of quasi-isomorphisms is natural in $D$, hence we can consider instead of 
the compactification of \eqref{eqn:CSPFA} on $\bbR^2\times \bbS^1$
the weakly equivalent prefactorization algebra
\begin{flalign}
\FFF_{\mathrm{CS}}^{\bbS^1} \,:=\,
\Big(\Sym\big( \Omega^\bullet_{\bbR^2,\cc}[2]\otimes \Omega^\bullet(\bbS^1) \big) , \dd_{\dR[2]}^{\bbR^2} + \dd_{\dR}^{\bbS^1}+\Delta_{\mathrm{BV}}\Big)\,\in \Alg_{\PP_{\bbR^2}}\quad,
\end{flalign}
where now $\Omega^\bullet_{\bbR^2,\cc}$ denotes the cosheaf of compactly supported differential forms on
the $2$-dimensional Cartesian space $\bbR^2$. Using the strong deformation retract from Lemma \ref{lem: deRhamSDR},
we can define a strong deformation retract 
\begin{equation}\label{eqn:CSSDR}
\begin{tikzcd}
\Omega^\bullet(\bbS^1) \,\cong\, \bbK\otimes \Omega^\bullet(\bbS^1) \ar[rr,shift right=-1ex,"{i[2]\otimes \id}"] 
&& \ar[ll,shift right=-1ex,"{p[2]\otimes \id}"] 
\Omega^\bullet_{\bbR^2,\cc}[2] \otimes \Omega^\bullet(\bbS^1) \ar[loop,xshift=10,out=-23,in=20,distance=35,swap,"{h[2]\otimes \id}"]
\end{tikzcd}\qquad.
\end{equation}
To pass to the cohomology $\H^\bullet_{\mathrm{dR}}(\bbS^1)\cong \bbK\oplus\bbK[-1]$ 
of $\Omega^\bullet(\bbS^1)$, we further use the strong deformation retract 
\begin{equation}\label{eqn:SDRHodge}
\begin{tikzcd}
\H^\bullet_{\mathrm{dR}}(\bbS^1) \ar[r,shift right=-1ex,"i^{\bbS^1}"] & \ar[l,shift right=-1ex,"p^{\bbS^1}"] \Omega^\bullet(\bbS^1) \ar[loop,out=-25,in=25,distance=30,swap,"h^{\bbS^1}"]
\end{tikzcd}
\end{equation} 
that is obtained by choosing the standard metric $g = \dd t \otimes \dd t$ on $\bbS^1$
and using Hodge theory. Combining the strong deformation retracts \eqref{eqn:CSSDR} 
and \eqref{eqn:SDRHodge}, lifting along $\Sym$ and applying as in
Proposition \ref{prop: fac env SDR} the homological perturbation lemma, we obtain
the strong deformation retract
\begin{equation} \label{SDR for CS}
\begin{tikzcd}
\Sym\big( \H^\bullet_{\mathrm{dR}}(\bbS^1) \big) \arrow[r, "i", shift right=-1ex] & \FFF_{\mathrm{CS}}^{\bbS^1} 
\arrow[l, "\widehat{p}", shift right=-1ex] \arrow["\widehat{h}"', loop, out=-25,in=25, distance=30]
\end{tikzcd}
\end{equation}
which we use to determine a minimal model and the universal first-order Massey product 
for $\FFF_{\mathrm{CS}}^{\bbS^1}$.
\sk

Similarly to Subsection \ref{sec: structure maps}, one finds that 
the underlying strict $(\PP_{\bbR^2})_\infty$-algebra structure on 
$\Sym\big( \H^\bullet_{\mathrm{dR}}(\bbS^1) \big)$ that is obtained via 
homotopy transfer along the strong deformation retract \eqref{SDR for CS}
is given by the standard associative, unital and commutative algebra structure
on the symmetric algebra. The non-formality of $\FFF_{\mathrm{CS}}^{\bbS^1}\in\Alg_{\PP_{\bbR^2}}$ 
is established again by explicitly computing the $L$-invariant from Subsection \ref{subsec:2dinvariant}.
\begin{theo}
The $L$-invariant \eqref{eqn:Lcombination} 
for the prefactorization algebra $\FFF_{\mathrm{CS}}^{\bbS^1}\in\Alg_{\PP_{\bbR^2}}$
of compactified linear Chern-Simons theory
on $\bbR^2\times\bbS^1$ is given by the degree $-1$ Poisson bracket
\begin{flalign}
L \,= \, \{-,-\}_{(-1)}^{} \, : \, \Sym\big( \H^\bullet_{\mathrm{dR}}(\bbS^1) \big) \otimes \Sym\big( \H^\bullet_{\mathrm{dR}}(\bbS^1) \big) \,\longrightarrow\, \Sym\big( \H^\bullet_{\mathrm{dR}}(\bbS^1) \big)
\end{flalign}
that is defined on the generators $[1],[\dd t]\in \Sym\big( \H^\bullet_{\mathrm{dR}}(\bbS^1) \big)$ by
\begin{flalign}\label{eqn:CSbracketongenerators}
\big\{[1],[1]\big\}_{(-1)}^{} \,=\, 0 \,=\, \big\{[\dd t],[\dd t]\big\}_{(-1)}^{}\quad,\qquad
\big\{[1],[\dd t]\big\}_{(-1)}^{} \,=\,1\,=\,\big\{[\dd t],[1]\big\}_{(-1)}^{}\quad.
\end{flalign}
It follows that the cohomology class 
$\big[ \widetilde{\beta}^{(2)} \big] \in \H^1_\Gamma\big( \Sym\big( \H^\bullet_{\mathrm{dR}}(\bbS^1) \big)\big)$ 
defining the universal first-order Massey product is non-trivial for $\FFF_{\mathrm{CS}}^{\bbS^1}\in\Alg_{\PP_{\bbR^2}}$.
\end{theo}
\begin{proof}
The proof follows the same steps as the proof of Theorem \ref{theo:LinvariantFacEnvelope} in Subsection \ref{subsec:universal1stMassey}.
As in Lemma \ref{lem: Massey computation}, 
we find that for any $A,B \in \Sym\big( \H^\bullet_{\mathrm{dR}}(\bbS^1) \big)$
the universal first-order Massey product is
\begin{flalign}\label{eqn:explicitCSMasseyprod}
\widetilde{\beta}^{(2)}\big(\iota_{(D_1^{\prime},D_2)}^D,\iota_{D_1}^{D_1^{\prime}}\big)(A \otimes B) \,=\, 
\varphi_0\big(\iota_{(D_1^{\prime},D_2)}^D,\iota_{D_1}^{D_1^{\prime}}\big) ~ \{ A, B \}_{(-1)}^{}\quad,
\end{flalign} 
where $\{ A, B \}_{(-1)}$ is the bracket determined by \eqref{eqn:CSbracketongenerators}
and $\varphi_0\big(\iota_{(D_1^{\prime},D_2)}^D,\iota_{D_1}^{D_1^{\prime}}\big) $
is the prefactor given by \eqref{eqn:varphi0}. 
The reason that this is the same prefactor as in Lemma \ref{lem: Massey computation}
is that both the strong deformation retract \eqref{gdeRhamSDR} 
and the strong deformation retract \eqref{eqn:CSSDR}
are constructed using the strong deformation retract for de Rham forms
found in Lemma \ref{lem: deRhamSDR}. 

As in Proposition \ref{prop: Massey simplify}
one then finds that after a gauge transformation $\chi$
(the nonzero part of which is given by \eqref{eqn:gaugetransformationforFE},
where $\{ A, B \}_{(-1)}$ now is the bracket determined by \eqref{eqn:CSbracketongenerators})
the prefactor becomes $\varphi\big(\iota_{(D_1^{\prime},D_2)}^D,\iota_{D_1}^{D_1^{\prime}}\big) $, 
which is defined in \eqref{eqn:varphi}. 
Finally, as in the proof of Theorem \ref{theo:LinvariantFacEnvelope} 
one finds that the prefactor of $L$ is $1$, proving the Theorem. 
\end{proof}


\section*{Acknowledgments}
We would like to thank Marco Benini, Victor Carmona and Joost Nuiten
for valuable discussions about this work.
B.V.\ and S.B.\ gratefully acknowledge the support of the Leverhulme Trust 
through a Leverhulme Research Project Grant (RPG-2021-154).
A.S.\ gratefully acknowledges the support of 
the Royal Society (UK) through a Royal Society University 
Research Fellowship (URF\textbackslash R\textbackslash 211015)
and Enhancement Grants (RF\textbackslash ERE\textbackslash 210053 and 
RF\textbackslash ERE\textbackslash 231077).

\section*{Data availability statement}
All data generated or analyzed during this study are contained in this document.


\end{document}